\documentclass[12pt]{article}

\usepackage{amsthm,amsmath,mathrsfs,amsfonts,amssymb,mathtools,nicefrac}
\usepackage{bbm}
\usepackage{graphicx}
\usepackage{natbib}
\RequirePackage[colorlinks,citecolor=blue,urlcolor=blue]{hyperref}
\usepackage{url}
\usepackage{enumerate}
\usepackage{multirow}
\usepackage{tikz}
\usetikzlibrary{positioning}
\usepackage[normalem]{ulem}
\usepackage{etoolbox}
\AtEndEnvironment{example}{\hfill$\square$}
\definecolor{dgreen}{rgb}{0.3,0.6,0}

\newcommand{\argmin}{\operatorname*{argmin}}
\newcommand{\plim}{\operatorname*{plim}}

\newtheorem{assumption}{Assumption}
\newtheorem{definition}{Definition}
\newtheorem{proposition}{Proposition}
\newtheorem{theorem}{Theorem}
\newtheorem{lemma}{Lemma}
\newtheorem{corollary}{Corollary}
\theoremstyle{definition}
\newtheorem{example}{Example}
\newtheorem{remark}{Remark}

\usepackage[margin=0.99in]{geometry}
\usepackage{setspace}
\usepackage{lmodern}

\begin{document}

\def\spacingset#1{\renewcommand{\baselinestretch}{#1}\small\normalsize} \spacingset{1}

{
  \title{{\bf \large HIGHER-ORDER NEYMAN ORTHOGONALITY \\ IN MOMENT-CONDITION MODELS}\thanks{%
We thank seminar participants at UCLA, McGill and Rotterdam, as well as participants at the African Econometric Society Meeting in Cairo and the UvA--KU Leuven workshop on panel data econometrics in Amsterdam, for helpful comments. This research was funded by the European Union (ERC-NETWORK-101044319), by the French Government and the French National Research Agency under the Investissements d'Avenir program (ANR-17-EURE-0010), and through NSF Grant 2242447. Views and opinions expressed are however those of the authors only and do not necessarily reflect those of the European Union or the European Research Council. Neither the European Union nor the granting authority can be held responsible for them.
  }}
  \author{\setcounter{footnote}{2}%
  \parbox{.45\linewidth}{\centering St\'ephane Bonhomme\thanks{Department of Economics, University of Chicago. Email: sbonhomme@uchicago.edu}} \hfill
  \parbox{.45\linewidth}{\centering Koen Jochmans\thanks{Toulouse School of Economics, Universit\'e Toulouse Capitole. Email: koen.jochmans@tse-fr.eu}} \\[0.5em]
  \parbox{.45\linewidth}{\centering Whitney K.\ Newey\thanks{Department of Economics, MIT, and NBER. Email: wnewey@mit.edu}} \hfill
  \parbox{.45\linewidth}{\centering Martin Weidner\thanks{Department of Economics, University College London. Email: m.weidner@ucl.ac.uk}}
}

\date{\vspace{1cm} November 2026 \vspace{0.25cm}}
  \maketitle
}

\begin{abstract}
\noindent
We construct moment functions that are Neyman-orthogonal to a chosen order in moment-condition models that involve finite-dimensional nuisance parameters. These moment functions reduce sensitivity to estimation error in the nuisance parameters and, as such, offer a unified and tractable route to higher-order debiasing in a wide range of econometric models.
The number of additional nuisance parameters required by our construction, beyond those already present in the original moment conditions, is independent of the order of orthogonalization and can be reduced to a single scalar if desired.
\end{abstract}

\vspace{0.5cm}

\newpage

\onehalfspacing

\allowdisplaybreaks
\renewcommand{\theequation}{\arabic{equation}}  \setcounter{equation}{0}

\section{Introduction}

The generalized method of moments \citep[GMM,][]{Hansen1982} provides a unifying framework for a vast range of estimators, from ordinary least squares and instrumental variables to maximum likelihood and causal-inference procedures. These problems will typically involve nuisance parameters, and the accuracy with which those parameters are estimated can have a substantial effect on inference for the parameter of interest.
It is often helpful to work with estimating equations that are orthogonal to the nuisance parameters in the sense of \cite{Neyman1959}. Neyman orthogonality means that the first derivative of the expected estimating equation with respect to the nuisance parameter vanishes at the true parameter values.\footnote{Neyman-orthogonal estimating equations play a central role in semiparametric estimation \citep{BickelKlaassenRitovWellner1993,Newey1994}. They are also at the heart of the recent literature on debiased inference in high-dimensional models \citep{BelloniChernozhukovHansen2014,JavanmardMontanari2014,vandeGeerBuhlmannRitovDezeure2014,ZhangZhang2014} and on double/debiased machine learning \citep{ChernozhukovChetverikovDemirerDufloHansenNeweyRobins2018}.}
When combined with sample splitting, this property helps reduce the bias of an estimator to make it $\sqrt{n}$-consistent \citep{ChernozhukovChetverikovDemirerDufloHansenNeweyRobins2018,ChernozhukovNeweyRobins2022}.

A limitation of (first-order) Neyman orthogonality is that the remaining bias may still be important relative to the sampling variability. This happens when the dimension of the nuisance parameter is large or when it is estimated by an estimator with a slow convergence rate. Important examples are fixed effects in panel data \citep{NeymanScott1948,HahnNewey2004} and many control variables in high-dimensional regression \citep{MikushevaSolvsten2025}. Work by \citet{AngristFrandsen2022} and \citet{WuthrichZhu2023} also highlights these limitations.
In such cases, going beyond first-order orthogonality can help. Neyman orthogonality to order $q$ (\citealt{MackeySyrgkanisZadik2018}) requires that all derivatives of the expected estimating equation with respect to the nuisance parameter, up to order $q$, vanish at the true parameter values.
\citet{BonhommeJochmansWeidner2025} showed how this can be done in a likelihood setting; see also \citet{bonhomme2026orthogonal}.
However, many important problems fall outside of this parametric framework.

In this paper we consider models defined only through moment conditions. A first set of moment conditions identifies the target parameter given some nuisance parameter, and a second set pins down the nuisance parameter. Both these parameters are finite-dimensional, although the dimension of the nuisance may grow with the sample size. This structure nests a wide array of econometric applications, since it contains both models that are affine in the nuisance (such as linear instrumental variables) and fully nonlinear models (which are common in structural settings).
The main contribution of this paper is to propose a general procedure to construct higher-order orthogonal moments, based on closed-form expressions.

Our orthogonal moment function depends on multiple independent copies of the data, and it features a generalized inverse of a Jacobian matrix as an additional nuisance parameter. The same nuisance parameter appears at all orders of orthogonality. To provide intuition, we start by illustrating the construction in a linear grouped-data model with heterogeneous coefficients. In this setting, we outline a simple construction that combines two ingredients: a polynomial approximation of the target parameter, and the use of multiple independent copies to estimate each component of the polynomial without bias.

These two ingredients underlie our general approach. The polynomial approximation is achieved through a von Mises expansion of the target parameter, initially proposed in \citet{vonMises1947} and extended to arbitrary order in \citet{Withers1983} and \citet{RobinsLiTchetgenTchetgenvanderVaart2008}. In our setting based on parametric moment conditions, an important input into the expansion device is the choice of a baseline family of distributions. We index this family by the finite-dimensional parameters and require that it satisfies the moment conditions. However, the requirement that the components of the approximating polynomial be estimable without bias imposes an additional constraint, namely that the expected higher derivatives of the original moment function, up to the chosen order $q$, coincide under the baseline distribution and the data distribution, hence are estimable without bias. This key constraint leads to well-defined orthogonal moments. %

We first apply our von Mises approach to provide explicit moment functions in the case where the original moment function is affine in the nuisance parameter. This covers our heterogeneous-coefficient model, and also linear regression and instrumental-variable estimation with linear moment conditions. We then extend our construction to the case where moment functions are nonlinear, in which case additional correction terms are needed. Our explicit expression for orthogonal moment functions in the nonlinear case is based on rooted trees.
A noteworthy feature of our approach is that the dimension of the additional nuisance parameter can be dramatically reduced. In fact, the nuisance can be reduced to a single scalar, even under overidentification.

While the main goal of the paper is to develop a general construction of orthogonal moments, we provide additional results for grouped-data settings. In the linear model with heterogeneous coefficients, we illustrate the finite-sample behavior of the second-order orthogonal estimators of the mean and second moments of the heterogeneous coefficients in a simulation exercise. The model and design are motivated by the study of firm-specific hiring discrimination in \citet{KlineRoseWalters2022}. The simulations suggest that orthogonality, together with good (regularized) first-step estimates of the nuisance parameters, can lead to significant improvements relative to the naive plug-in approach. Additionally, in the appendix we develop asymptotic theory for the grouped data setting.

For the partially-linear regression model \cite{MackeySyrgkanisZadik2018} provide an explicit moment that is nonparametrically second-order orthogonal, but only when first-stage errors are non-Gaussian; they prove an impossibility result under Gaussianity. Our moment construction circumvents this impossibility because we rely on independent copies of the data and on ordinary finite-dimensional derivatives, rather than on functional derivatives. 
In recent independent work, \cite{ChetverikovSorensenTsyvinski2026} propose a recursive formula for Z-estimation problems and apply it to construct a ``triple-Lasso'' estimator for high-dimensional linear regression that is second-order orthogonal. Like our construction, theirs uses independent copies of the data but it differs from ours in several respects. In particular, their iterative construction is not closed form and, importantly, introduces additional nuisance parameters at each recursion step. As such, the number of nuisance parameters grows with the order of orthogonality. They also restrict attention to exactly-identified moment conditions.
In other related work, \cite{RobinsLiTchetgenTchetgenvanderVaart2008} and \cite{RobinsLiMukherjeeTchetgenvanderVaart2017} have developed higher-order influence functions for nonparametric models. Their results rely on nonparametric von Mises expansions whose higher-order terms can only be approximated, whereas, here, the nuisance parameter is finite-dimensional and the terms of the expansion can be estimated without bias. 

The outline for the rest of the paper is as follows. Section~\ref{sec:setup} describes the setup and defines $q$-th order orthogonality. Section~\ref{sec:example} provides intuition and supporting simulations in a model with heterogeneous coefficients. Section~\ref{sec:vonmises} presents our general approach to orthogonal moment construction based on von Mises expansions. Section~\ref{sec:affine} applies this approach and gives explicit formulas for the affine case. Section~\ref{sec:nonlinear} similarly describes the construction in detail for the nonlinear case, first at second order and then at arbitrary order in terms of rooted trees. Section~\ref{sec:nuisance} gives a nuisance-dimension reduction technique. Finally, Section~\ref{sec:con} concludes. The Appendix contains all proofs and additional results.

\section{Setup\label{sec:setup}}

\subsection{Model and parameter of interest}

We observe a sample $\{W_i : i = 1, \ldots, n\}$ from an unknown distribution $P_0$. Consider a model defined by two moment conditions,
\begin{align}
\mathbb{E}[g(W_i, \theta_0, \eta_0)] &= 0,
&
\mathbb{E}[m(W_i, \theta_0, \eta_0)] &= 0,
\label{eq:model-general}
\end{align}
where $\theta_0 \in \mathbb{R}^{d_\theta}$ is the parameter of interest and $\eta_0 \in \mathbb{R}^{d_\eta}$ is a finite-dimensional nuisance parameter. The functions $g(W_i, \theta, \eta) \in \mathbb{R}^{d_g}$ and $m(W_i, \theta, \eta) \in \mathbb{R}$ are known. All expectations are taken with respect to $P_0$.

\begin{assumption}[i.i.d.\ sampling]
\label{ass:iid}
The observations $W_1, \ldots, W_n$ are independent and identically distributed draws from $P_0$.
\end{assumption}

The first moment condition identifies $\eta_0$ for given $\theta_0$. We assume $d_g \geq d_\eta$, which allows the nuisance parameter to be either exactly identified or overidentified. The second moment condition is used to identify $\theta_0$. We take $m$ to be scalar throughout, purely for notational convenience. In practice one will often have $m(W_i, \theta, \eta) \in \mathbb{R}^{d_m}$ with $d_m > 1$, in which case our orthogonalization construction is applied to each component of $m$ separately. Identification of $\theta_0$ in a given application requires $d_m \geq d_\theta$, and the target-moment system in $\theta$ may itself be exactly identified or overidentified. These identification considerations do not affect the orthogonalization construction that we present.

Our goal is to construct a Neyman-orthogonal moment function $\psi$ that replaces the target moment function $m$.\footnote{In our construction, the parameter $\theta$ plays a passive role. It is carried through as an argument of $m$ and $g$. Only after the orthogonal moment function has been constructed does $\theta$ enter the analysis as the parameter to be estimated, through the orthogonal moment equations. Thus, one could drop $\theta$ from most of the subsequent formulas without changing their content, but we retain it throughout the main text to avoid notational confusion.}

\subsection{Nuisance parameters and the Jacobian inverse}

Let  $\partial_\eta$ be the gradient operator with respect to $\eta$. Our construction requires a left inverse of the $d_g \times d_\eta$ Jacobian matrix
\begin{equation}
J_0 := \mathbb{E}[\partial_\eta g(W_i, \theta_0, \eta_0)].
\label{eq:Jacobian}
\end{equation}
Let $\Lambda_0$ be a $d_\eta \times d_g$ matrix satisfying $\Lambda_0 J_0 = I_{d_\eta}$. When $d_g = d_\eta$ and $J_0$ is invertible, $\Lambda_0 = J_0^{-1}$. When $d_g > d_\eta$ and $J_0$ has full column rank, we can set $\Lambda_0 = (J_0'\, \Omega^{-1}\, J_0)^{-1} J_0'\, \Omega^{-1}$ for a positive-definite weight matrix $\Omega$. When $d_g > d_\eta$ the left inverse is not unique. Our results hold for any choice, with $\Lambda_0$ fixed throughout.

The matrix $\Lambda_0$ enters as an additional nuisance parameter in our construction. We parametrize it as $\Lambda(\lambda)$ for a finite-dimensional parameter $\lambda \in \mathbb{R}^{d_\lambda}$, with true value $\lambda_0$ satisfying $\Lambda(\lambda_0) = \Lambda_0$. The simplest choice is $\lambda = \mathrm{vec}(\Lambda)$, giving $d_\lambda = d_\eta \, d_g$.

\begin{assumption}[Jacobian left inverse]
\label{ass:jacobian-inverse}
The Jacobian matrix $J_0$ has full column rank, so a left inverse $\Lambda_0$ with $\Lambda_0 J_0 = I_{d_\eta}$ exists. The map $\lambda \mapsto \Lambda(\lambda)$ satisfies $\Lambda(\lambda_0) = \Lambda_0$.
\end{assumption}

While the dimension $d_\lambda$ does not enter the orthogonalization construction itself, $d_{\lambda}$ does matter in applications, because the accuracy with which $\widehat\Lambda$ can be estimated depends on $d_\lambda$, with a smaller $d_\lambda$ potentially leading to a more accurate preliminary estimate. For this reason, Section~\ref{sec:nuisance} develops techniques that reduce $d_\lambda$, to a single scalar if desired, by replacing the moment function $g$ with a suitably transformed version.

\begin{example}[Heterogeneous coefficients in grouped data]
\label{ex:grouped}
Our first, running example is a grouped data model with heterogeneous coefficients,
\begin{equation}
Y_{it} = X_{it}'\eta_{i0} + U_{it}, \quad i = 1, \ldots, N, \quad t = 1, \ldots, T,
\label{eq:RC-model}
\end{equation}
for a scalar outcome $Y_{it}$ and a $d_\eta$-dimensional covariate vector $X_{it}$, under the uncorrelatedness condition
\begin{equation}
\mathbb{E}[X_{it} U_{it}] = 0, \quad i = 1, \ldots, N, \quad t = 1, \ldots, T.
\label{eq:uncorr}
\end{equation}
Such models have been studied extensively under the mean independence condition $\mathbb{E}[U_{it} \mid X_{i1}, \ldots, X_{iT}] = 0$ (e.g., \citealp{Chamberlain1992}). Condition \eqref{eq:uncorr} is weaker, since it does not restrict the conditional mean of $Y_{it}$ given $X_{it}$ to be linear. The coefficients $\eta_{i0}$ are the best linear predictors,
\begin{equation}
\eta_{i0} = \bigg(\mathbb{E}\bigg[\sum_{t=1}^T X_{it} X_{it}'\bigg]\bigg)^{-1} \mathbb{E}\bigg[\sum_{t=1}^T X_{it} Y_{it}\bigg],\quad i=1,...,N,\label{eq;etai0}
\end{equation}
and the target parameter is $$\theta_0 = \frac{1}{N} \sum_{i=1}^N m(\eta_{i0})$$ for a smooth scalar function $m$. The mean and second moment of the components of $\eta_{i0}$ are natural choices, for example to quantify the level and dispersion of the coefficients across units.

For each unit $i$, write the observation at $t$ as $W_{it} = (Y_{it}, X_{it})$. The framework \eqref{eq:model-general} applies with
\[
g(W_{it}, \theta, \eta_i) = X_{it}(Y_{it} - X_{it}' \eta_i), \qquad m_i(W_{it}, \theta, \eta_i) = m(\eta_i) - \theta ,
\]
and condition \eqref{eq:uncorr} is equivalent to $\mathbb{E}[g(W_{it}, \theta_0, \eta_{i0})] = 0$ for every $t$. The Jacobian-inverse nuisance parameter is $\Lambda_{i0} = -(\mathbb{E}[X_{it} X_{it}'])^{-1}$.

We assume that, for each $i$, the observations $W_{i1}, \ldots, W_{iT}$ are i.i.d., so that Assumption~\ref{ass:iid} holds within each unit. We will apply our construction unit by unit, with $W_{it}$ in place of $W_i$, $(\eta_{i0}, \Lambda_{i0})$ in place of $(\eta_0, \Lambda_0)$, and the unit-level target $\theta_{i0} := m(\eta_{i0})$ in place of $\theta_0$, so that $\mathbb{E}[m_i(W_{it}, \theta_{i0}, \eta_{i0})] = 0$ holds as in~\eqref{eq:model-general}. Our construction will yield, for each unit, an orthogonal moment function $\psi_i$ that depends on several observations of unit $i$. The estimator of $\theta_0 = (1/N)\sum_{i=1}^N \theta_{i0}$ averages the unit-level orthogonal moments over $i$.

Each $\eta_{i0}$ is estimated from the $T$ observations of unit $i$ alone, so it converges at rate $T^{-\nicefrac{1}{2}}$, and the plug-in estimator of $\theta_0$ has a bias of order $T^{-1}$. Orthogonalizing to order $q$ reduces the bias to order $T^{-(q+1)/2}$.

\end{example}

\begin{example}[Linear instrumental variables, including linear regression as a special case]
\label{ex:iv}
To introduce our second example, let $W_i = (Y_i, D_i, Z_i, X_i)$ with $D_i \in \mathbb{R}$ a regressor of interest, $Z_i \in \mathbb{R}$ an instrument satisfying $\mathbb{E}[Z_i D_i] \neq \mathbb{E}[Z_i X_i'](\mathbb{E}[X_i X_i'])^{-1}\mathbb{E}[X_i D_i]$, and $X_i \in \mathbb{R}^K$ exogenous controls. Consider
\[
Y_i = D_i \theta_0 + X_i' \eta_0 + U_i, \qquad \mathbb{E}[X_i U_i] = 0, \qquad \mathbb{E}[Z_i U_i] = 0.
\]
The parameter of interest is $\theta_0 \in \mathbb{R}$. The framework~\eqref{eq:model-general} applies with
\[
g(W_i, \theta, \eta) = X_i (Y_i - D_i \theta - X_i' \eta), \qquad m(W_i, \theta, \eta) = Z_i (Y_i - D_i \theta - X_i' \eta).
\]
Both moment functions are affine in $\eta$, and the Jacobian is $J_0 = -\mathbb{E}[X_i X_i']$, so $\Lambda_0 = -(\mathbb{E}[X_i X_i'])^{-1}$ is minus the inverse of the population second-moment matrix of the controls. The special case $Z_i = D_i$ corresponds to linear regression with high-dimensional controls, in which $\theta_0$ is the coefficient on $D_i$ in the projection of $Y_i$ on $(D_i, X_i)$ and the additional moment condition $\mathbb{E}[D_i U_i] = 0$ is needed for identification.
\end{example}

\begin{remark}[$\Lambda$ as a by-product of estimating $\eta$]
\label{rem:lambda-byproduct}
Although Assumption~\ref{ass:jacobian-inverse} formally introduces $\Lambda(\lambda)$ as an additional nuisance parameter, in practice $\widehat\Lambda$ is often a by-product of estimating $\eta$. In Example~\ref{ex:iv}, computing the least-squares estimator of $\eta$ requires inverting the sample second-moment matrix $n^{-1}\sum_i X_i X_i'$, and that inversion already delivers an estimator for $\Lambda$ at no additional cost. The same is true in nonlinear settings, where the Jacobian inverse at the solution $\widehat\eta$ is computed as part of the numerical optimization.
\end{remark}

\subsection{Higher-order Neyman orthogonality}

Let $\psi(W_1, \ldots, W_L; \theta, \eta, \lambda) \in \mathbb{R}$ be a moment function that depends on $L \geq 1$ independent copies of $W_i$, the parameter of interest $\theta$, the nuisance parameter $\eta$, and the Jacobian-inverse parameter $\lambda$. Define the population moment
\begin{equation}
\Psi(\theta, \eta, \lambda) := \mathbb{E}[\psi(W_1, \ldots, W_L; \theta, \eta, \lambda)].
\label{eq:Psi}
\end{equation}

\begin{definition}[Higher-order Neyman orthogonality]
\label{def:orthogonality}
For $q \in \{1, 2, 3, \ldots\}$, the moment function $\psi$ is \emph{$q$-th order Neyman orthogonal} if there are a neighborhood $\mathcal{N}$ of $(\eta_0, \lambda_0)$ and a constant $C < \infty$ such that
\begin{equation}
\bigl| \Psi(\theta_0, \eta, \lambda)  - \Psi(\theta_0, \eta_0, \lambda_0) \bigr| \; \leq \; C \, \bigl( \|\eta - \eta_0\| + \|\lambda - \lambda_0\| \bigr)^{q+1}
\qquad \text{for all } (\eta, \lambda) \in \mathcal{N} .
\label{eq:ortho-bound}
\end{equation}
\end{definition}

 Notice that what matters here  is how $\Psi$ varies with $(\eta,\lambda)$, not its level. While we set $\mathbb{E}[m(W_i, \theta_0, \eta_0)] = 0$ in~\eqref{eq:model-general} to follow the usual GMM presentation, nothing below uses this normalization
to zero. A target $m(W_i, \theta, \eta)$ for which $\Psi(\theta_0,\eta_0,\lambda_0) = \mathbb{E}[m(W_i, \theta_0, \eta_0)]  \neq 0$ is 
equally allowed by the above definition and by our construction of $\psi$ below.

Orthogonality is more commonly stated through derivatives of $\Psi$ at the true values. For multi-indices $\alpha = (\alpha_1, \ldots, \alpha_{d_\eta})$ and $\beta = (\beta_1, \ldots, \beta_{d_\lambda})$ of non-negative integers, write $|\alpha| = \sum_j \alpha_j$ and $\partial_\eta^\alpha \partial_\lambda^\beta = \partial^{|\alpha|+|\beta|} / (\partial \eta_1^{\alpha_1} \cdots \partial \eta_{d_\eta}^{\alpha_{d_\eta}} \partial \lambda_1^{\beta_1} \cdots \partial \lambda_{d_\lambda}^{\beta_{d_\lambda}})$.

\begin{proposition}
\label{prop:reconcile}
Let $\psi$ be $q$-th order Neyman orthogonal, and suppose that $(\eta, \lambda) \mapsto \Psi(\theta_0, \eta, \lambda)$ is $q$ times differentiable at $(\eta_0, \lambda_0)$. Then
\begin{equation}
\partial_\eta^{\alpha} \partial_\lambda^{\beta} \Psi(\theta_0, \eta_0, \lambda_0) = 0
\quad \text{for all multi-indices $(\alpha, \beta)$ with } 1 \leq |\alpha| + |\beta| \leq q .
\label{eq:ortho}
\end{equation}
\end{proposition}

At $q = 1$ this is the standard Neyman orthogonality condition. The two formulations agree under enough smoothness: if $(\eta, \lambda) \mapsto \Psi(\theta_0, \eta, \lambda)$ is $q+1$ times continuously differentiable near $(\eta_0, \lambda_0)$, then \eqref{eq:ortho} implies \eqref{eq:ortho-bound} by Taylor's theorem, so the two possible formulations of $q$-th order orthogonality are then equivalent. \cite{MackeySyrgkanisZadik2018} and \cite{BonhommeJochmansWeidner2025} state higher-order orthogonality through \eqref{eq:ortho}.

The conditions \eqref{eq:ortho-bound} and \eqref{eq:ortho} impose orthogonality in $\eta$ and $\lambda$ jointly. Beyond the pure $\eta$-derivatives and pure $\lambda$-derivatives of $\Psi$, condition~\eqref{eq:ortho} also requires all mixed derivatives $\partial_\eta^\alpha \partial_\lambda^\beta \Psi$ with $|\alpha| \geq 1$, $|\beta| \geq 1$, and $|\alpha| + |\beta| \leq q$ to vanish at the true values. The pure $\eta$-conditions control bias from estimating the original nuisance parameter, the pure $\lambda$-conditions control bias from estimating the Jacobian inverse, and the mixed conditions control bias arising from the interaction of the two estimation errors.

\paragraph{Estimation.}
Given an i.i.d.\ sample $\{W_i : i = 1, \ldots, n\}$, the natural sample analogue of the population moment $\Psi(\theta, \eta, \lambda)$ is the U-statistic
\begin{align}
\widehat\Psi_n(\theta, \eta, \lambda) := \frac{1}{n(n-1)\cdots(n-L+1)} \sum_{\substack{i_1, \ldots, i_L = 1 \\ \text{all distinct}}}^{n} \psi(W_{i_1}, \ldots, W_{i_L}; \theta, \eta, \lambda).
     \label{UstatsPsi}
\end{align}
An estimator $\widehat\theta$ of $\theta_0$ is obtained by solving $\widehat\Psi_n(\widehat\theta, \widehat\eta, \widehat\lambda) = 0$, where $\widehat\eta$ and $\widehat\lambda$ are preliminary estimators of the nuisance parameters. To ensure that the estimation error in $(\widehat\eta, \widehat\lambda)$ does not contaminate the U-statistic, the sample is split into two independent parts. One part is used to construct $\widehat\eta$ and $\widehat\lambda$, and the other is used to evaluate $\widehat\Psi_n$. Efficiency can be improved by cross-fitting, that is by swapping the roles of the two sample parts and averaging, as in \cite{ChernozhukovChetverikovDemirerDufloHansenNeweyRobins2018}.

The key consequence of $q$-th order orthogonality is that, under appropriate regularity conditions, the estimator satisfies $$\widehat\theta - \theta_0 = O_p\bigl( (\|\widehat\eta - \eta_0\| + \|\widehat\lambda - \lambda_0\|)^{q+1} + n^{-\nicefrac{1}{2}} \bigr),$$ because \eqref{eq:ortho-bound} applied at $(\widehat\eta, \widehat\lambda)$ bounds the bias from the nuisance estimates, together with $\Psi(\theta_0,\eta_0,\lambda_0) = 0$, which holds under the normalization in~\eqref{eq:model-general}. A larger $q$ therefore makes the bias from the nuisance estimates of higher order, so that valid inference requires only $\|\widehat\eta - \eta_0\| + \|\widehat\lambda - \lambda_0\| = o_p\bigl(n^{-\nicefrac{1}{2(q+1)}}\bigr)$.

\section{Illustration in a heterogeneous coefficients model\label{sec:example}}

This section illustrates the main ideas in a concrete example, the heterogeneous coefficients model of Example \ref{ex:grouped}. Our motivating application is \cite{KlineRoseWalters2022}, who study employment discrimination by sending randomized fictitious job applications to firms and measuring callbacks $Y_{it}$. In this setting, independence across job applications $t$ for each firm $i$ is plausible by design; $\eta_{i0}$ captures how each firm $i$ responds to applicant characteristics $X_{it}$ such as race, gender, and age, and moments of $\eta_{i0}$ quantify the extent and variation of discrimination across firms.

\subsection{Main ideas\label{sec:example-ideas}}

The individual-specific ordinary least-squares estimates $\widehat\eta_i^{\rm OLS}$, for $i=1,...,N$, are noisy when $T$ is small, and the plug-in estimator $\widehat\theta_i^{\rm plug\text{-}in} = m(\widehat\eta_i^{\rm OLS})$ of $\theta_{i0} = m(\eta_{i0})$ is biased in general. Our approach constructs alternative estimators with reduced bias. Since the problem stratifies across units, we focus on a single unit $i$ and describe the construction at the unit level. 

A specific feature of this example is that the moment function is \emph{affine} in $\eta_i$. It is useful to introduce the reparameterization $$A_{i0} = -\Lambda_{i0}^{-1},\quad b_{i0} = -\Lambda_{i0}^{-1} \eta_{i0},$$ with $a_{i0} = \mathrm{vech}(A_{i0})$ denoting the vector that contains the unique elements of $A_{i0}$.

The construction in this setting relies on two insights. A preliminary observation is that $A_{i0}$ and $b_{i0}$ admit simple unbiased estimators, $X_{it} X_{it}'$ and $X_{it} Y_{it}$ respectively, since
\begin{equation}
\mathbb{E}[X_{it} X_{it}'] = A_{i0}, \qquad \mathbb{E}[X_{it} Y_{it}] = b_{i0}.
\label{eq:unbiased-Ab}
\end{equation}
The first insight is that any \emph{polynomial} $P(a_{i0}, b_{i0})$ admits an unbiased estimator as well. Indeed, since observations are independent across $t$, one can combine unbiased estimators of $(a_{i0}, b_{i0})$ from different $t$ observations to produce unbiased estimators of powers and products. For example, when $d_\eta = 1$, an unbiased estimator of $a_{i0}^2 = (\mathbb{E}[X_{it}^2])^2$ is $X_{i1}^2 X_{i2}^2$, since $$\mathbb{E}[X_{i1}^2 X_{i2}^2] = \mathbb{E}[X_{i1}^2]\, \mathbb{E}[X_{i2}^2] = a_{i0}^2.$$

The second insight is that $m(\eta_{i0})$ can be \emph{approximated} by a polynomial in $(a_{i0}, b_{i0})$. Writing $m(\eta_{i0}) = \varphi(a_{i0}, b_{i0})$ with $\varphi(a, b) = m(A^{-1} b)$, the function $\varphi$ has a singularity at $\det A = 0$ and so cannot be globally represented as a polynomial. However, given preliminary estimators $\widehat\Lambda_i$ and $\widehat\eta_i$ that are independent of the remaining data, one can expand locally: setting $\widehat a_i = -\mathrm{vech}(\widehat\Lambda_i^{-1})$ and $\widehat b_i = -\widehat\Lambda_i^{-1} \widehat\eta_i$,
\begin{equation}
m(\eta_{i0}) = \sum_{|\alpha| + |\beta| \leq q} \frac{1}{\alpha!\, \beta!}\, \partial_a^\alpha \partial_b^\beta \varphi(\widehat a_i, \widehat b_i)\, (a_{i0} - \widehat a_i)^\alpha (b_{i0} - \widehat b_i)^\beta + R_{q,i},
\label{eq:taylor-ab}
\end{equation}
where the remainder $R_{q,i}$ is of order $q+1$ in $a_{i0} - \widehat a_i$ and $b_{i0} - \widehat b_i$. Each term $(a_{i0} - \widehat a_i)^\alpha (b_{i0} - \widehat b_i)^\beta$ admits an unbiased estimator built from independent observations. This yields an approximately unbiased estimator of $m(\eta_{i0})$ based on a $q$-th order orthogonal moment function.

These two insights are sufficient to construct orthogonal moments in the heterogeneous coefficients model of Example \ref{ex:grouped}, as we will show next. Furthermore, they underlie our general constructions in Sections \ref{sec:affine} and \ref{sec:nonlinear} based on von Mises expansions. In practice, the preliminary estimators $\widehat\Lambda_i$ and $\widehat\eta_i$ are obtained by holding out a set of $t$ observations, and efficiency can be improved by cross-fitting. A feature of this construction is that the nuisance parameters $\widehat\Lambda_i$ and $\widehat\eta_i$ do not change with the order $q$ of orthogonalization.

\subsection{Orthogonal moment}

We now show how to construct a second-order orthogonal moment function (i.e., at $q = 2$) in model \eqref{eq:RC-model} with a scalar covariate (i.e., $d_\eta = 1$). This construction is a special case of the general formulas, for arbitrary $d_\eta$ and arbitrary order $q$, given in the next sections. 

A second-order Taylor expansion of $m(\eta_{i0})$ around $\widehat\eta_i$ gives
\begin{equation}
m(\eta_{i0}) \approx m(\widehat\eta_i) + \partial_\eta m(\widehat\eta_i)\,(\eta_{i0} - \widehat\eta_i) + \tfrac{1}{2}\, \partial_\eta^2 m(\widehat\eta_i)\,(\eta_{i0} - \widehat\eta_i)^2.\label{eq_taylor2}
\end{equation}
To express $\eta_{i0} - \widehat\eta_i$ in terms of quantities with unbiased estimators, write $a_{i0} =- \Lambda_{i0}^{-1} = \mathbb{E}[X_{it}^2]$ and $b_{i0} = -\Lambda_{i0}^{-1} \eta_{i0} = \mathbb{E}[X_{it} Y_{it}]$. Expanding $\eta_{i0} = b_{i0}/a_{i0}$ around $(\widehat a_i, \widehat b_i) = (-\widehat\Lambda_i^{-1}, -\widehat\Lambda_i^{-1}\widehat\eta_i)$ to second order,
\[
\eta_{i0} - \widehat\eta_i \approx -\widehat\Lambda_i(b_{i0} - \widehat\eta_i\, a_{i0}) - \widehat\Lambda_i(1 + \widehat\Lambda_i\, a_{i0})(b_{i0} - \widehat\eta_i\, a_{i0}).
\]
Substituting into (\ref{eq_taylor2}) and replacing $a_{i0}$ and $b_{i0}$ by their unbiased estimators ($X_{it}^2$ and $X_{it} Y_{it}$, respectively), we obtain the approximately unbiased estimator of $m(\eta_{i0})$
\begin{align*}m(\widehat{\eta}_i)&-\partial_\eta m(\widehat\eta_i)\left(	\widehat{\Lambda}_i(X_{i1}Y_{i1}-\widehat{\eta}_i X_{i1}^2)+\widehat{\Lambda}_i(1+\widehat{\Lambda}_iX_{i2}^2)(X_{i1}Y_{i1}-\widehat\eta_i X_{i1}^2)\right)\\&+\frac{1}{2}\partial_\eta^2 m(\widehat\eta_i)(	\widehat{\Lambda}_i(X_{i1}Y_{i1}-\widehat{\eta}_i X_{i1}^2))(	\widehat{\Lambda}_i(X_{i2}Y_{i2}-\widehat{\eta}_i X_{i2}^2)).
\end{align*}

The resulting moment function is
\begin{align}
\psi(W_i; \theta, \eta_i, \Lambda_i) &= m(\eta_i) - \partial_\eta m(\eta_i)\, \Lambda_i(X_{i1} Y_{i1} - \eta_i\, X_{i1}^2) \notag\\
&\quad - \partial_\eta m(\eta_i)\, \Lambda_i(1 + \Lambda_i\, X_{i2}^2)(X_{i1} Y_{i1} - \eta_i\, X_{i1}^2) \notag\\
&\quad + \tfrac{1}{2}\, \partial_\eta^2 m(\eta_i)\, \Lambda_i(X_{i1} Y_{i1} - \eta_i\, X_{i1}^2) \, \Lambda_i(X_{i2} Y_{i2} - \eta_i\, X_{i2}^2) - \theta.
\label{eq:psi-covariate}
\end{align}
One can verify that $\mathbb{E}[\psi(W_i; \theta_{i0}, \eta_i, \Lambda_i)]$ and all its first and second derivatives in $(\eta_i, \Lambda_i)$ at $(\eta_{i0}, \Lambda_{i0})$ vanish. The moment function is therefore second-order Neyman-orthogonal.

\subsection{Simulation evidence\label{sec:implementation}}

We now report simulations that illustrate the finite-sample behavior of the second-order orthogonal estimator based on (\ref{eq:psi-covariate}) (adapted to a vector-valued $\eta_i$), and compare it to the behavior of plug-in and first-order orthogonal estimators. While our goal in the paper is to present a general construction of orthogonal moments, we complement these simulations in Appendix \ref{sec:asymptotic} by presenting asymptotic properties for the grouped data setting.

The simulation design is a stylized version of the setup in \citet{KlineRoseWalters2022}, where we do not attempt to calibrate the design to their empirical results. There are $N$ firms, and $T$ job applications per firm. Applications feature the race $X_{it,1}\in\{0,1\}$ of applicants and their gender $X_{it,2}\in\{0,1\}$. Binary callback indicators are generated as
\begin{equation}
	Y_{it}=\boldsymbol{1}\{\varepsilon_{it}\geq \beta_{i0}+\beta_{i1}X_{it,1}+\beta_{i2}X_{it,2}\},\label{eq_callbacks}
\end{equation}
where $\varepsilon_{it}$ are i.i.d. standard logistic independent of $X_{it,1},X_{it,2}$, and observations are independent across $i$ and $t$. Job applications are randomly generated within firms but -- for example due to implementation constraints -- assignment is heterogeneous across firms.\footnote{Specifically, we assume that firms belong to one of two types $Z_i\in\{1,2\}$, i.i.d. with probability 1/2. When $Z_i=1$, $(X_{it,1}=1,X_{it,2}=1)$ and $(X_{it,1}=0,X_{it,2}=0)$ are each drawn with probability 3/8, and $(X_{it,1}=0,X_{it,2}=1)$ and $(X_{it,1}=1,X_{it,2}=0)$ are each drawn with probability 1/8. When $Z_i=2$, $(X_{it,1}=1,X_{it,2}=1)$ and $(X_{it,1}=0,X_{it,2}=0)$ are each drawn with probability 1/8, and $(X_{it,1}=0,X_{it,2}=1)$ and $(X_{it,1}=1,X_{it,2}=0)$ are each drawn with probability 3/8.} For $j\in\{0,1,2\}$, $\beta_{ij}=Z_i+V_{ij}$, where the $V_{ij}$ are independent standard normal, independent of $Z_i$.

Let $X_{it}=(1,X_{it,1},X_{it,2})'$, and let $\eta_{i0}$ be given by (\ref{eq;etai0}). We estimate the parameters
$$\theta_{0,1}=\frac{1}{N}\sum_{i=1}^N\eta_{i10},\quad \theta_{0,2}=\frac{1}{N}\sum_{i=1}^N\eta_{i10}^2,$$
which together are informative about the level and dispersion of race-based employment discrimination (here, the coefficient $\eta_{i10}$ corresponds to race). The true values in the DGP are $\theta_{0,1}=-0.0842$, which corresponds to an $8.42$ percentage point lower callback probability for blacks compared to whites, and $\theta_{0,2}=0.0177$, which corresponds to a standard deviation of $0.1030$. 

We report results based on the plug-in estimator (OLS) and the order-$1$ and order-$2$ orthogonal estimators (ORTH) based on $1000$ simulations, for $N=1000$ firms. For the latter we rely on sample splitting, leaving out one observation ($q=1$) or two observations ($q=2$), to obtain preliminary estimates of $\eta_i$ and $\Lambda_i$. In addition to the inverse of the sample second-moment matrix as an estimator of $\Lambda_i$, we report results that use a regularized estimator based on a Dirichlet hierarchical prior. In Appendix \ref{app_sim} we describe how these estimators are computed.

Figures \ref{fig:theta1} and \ref{fig:theta2} report the Monte Carlo performance of the OLS (in blue) and ORTH (in green at order $q=1$, and in red at order $q=2$) estimators for $\theta_1$ and $\theta_2$, respectively, as the number of applications per firm $T$ varies between $T=20$ and $T=100$. Each figure plots the mean estimate across simulations (solid line), the true value (dashed horizontal line), and 90\% simulation bands defined by the 5th and 95th quantiles across simulations (shaded regions). The plots on the left-hand side show the estimators based on regularized $\Lambda_i$'s, while the plots on the right-hand side correspond to un-regularized $\Lambda_i$'s.

\begin{figure}[htbp]
	\centering
	\begin{tabular}{cc}
   Regularized & Un-regularized\\
    \multicolumn{2}{c}{}\\
        \includegraphics[width=0.47\textwidth]{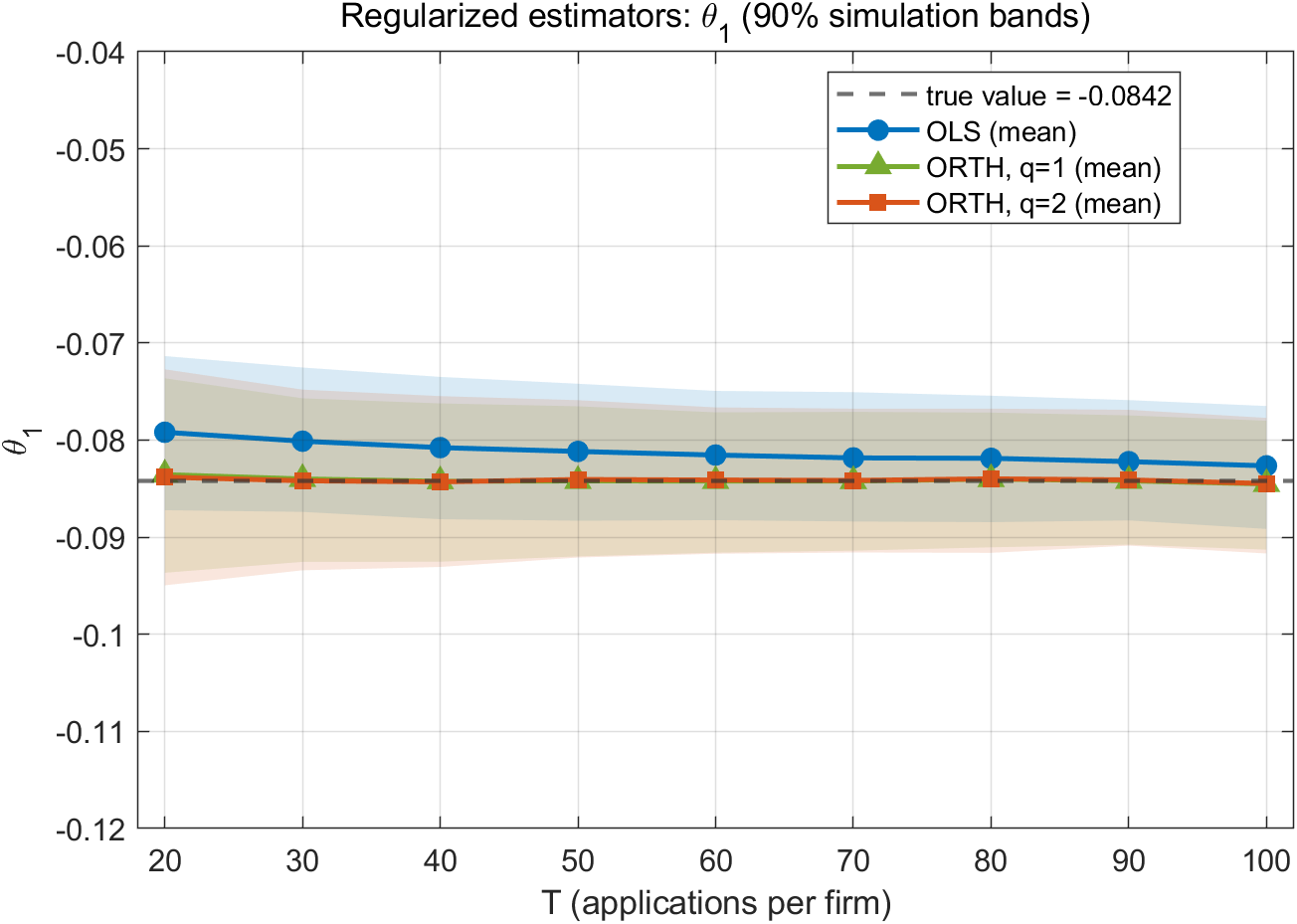} & \includegraphics[width=0.47\textwidth]{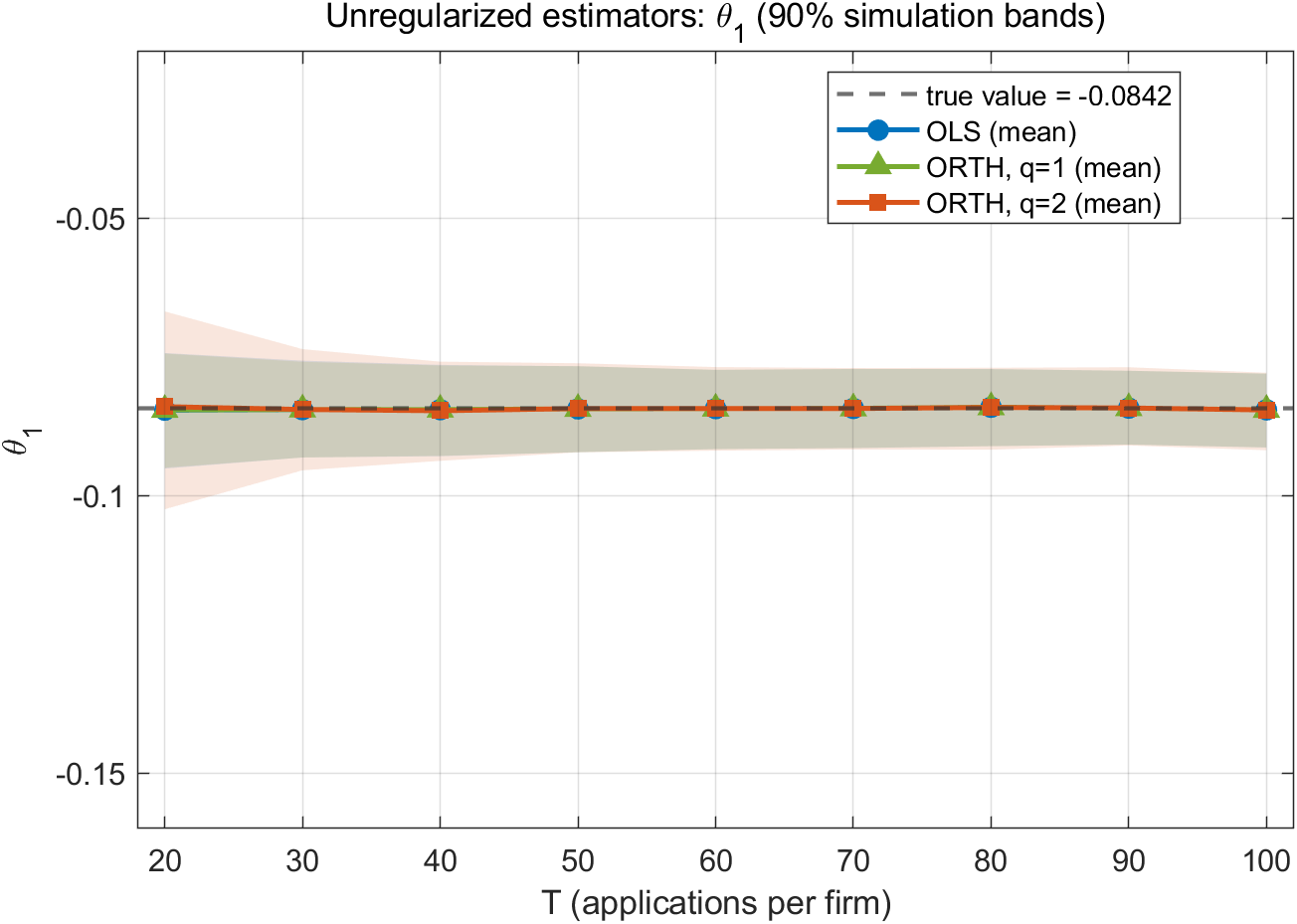}
        \end{tabular}
	\caption{Monte Carlo estimates of $\theta_1$ as a function of $T$. Solid lines show the mean across simulations for OLS and ORTH (at order $q=1$ and $q=2$); the dashed line is the true value; shaded regions are 90\% simulation bands.}
	\label{fig:theta1}
\end{figure}

\begin{figure}[h!]
	\centering
	\begin{tabular}{cc}
  Regularized & Un-regularized\\
  \multicolumn{2}{c}{}\\
        \includegraphics[width=0.47\textwidth]{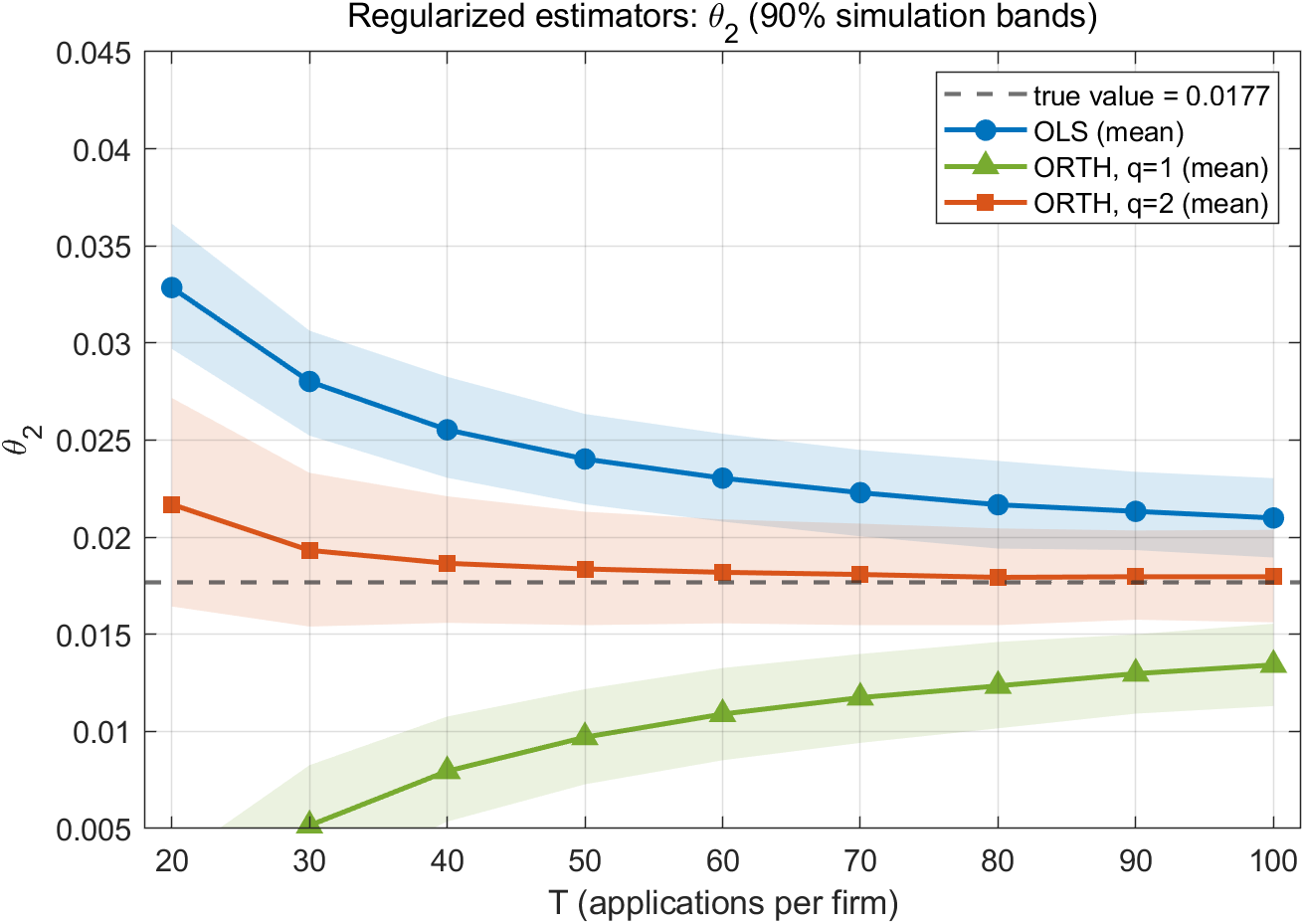} & \includegraphics[width=0.47\textwidth]{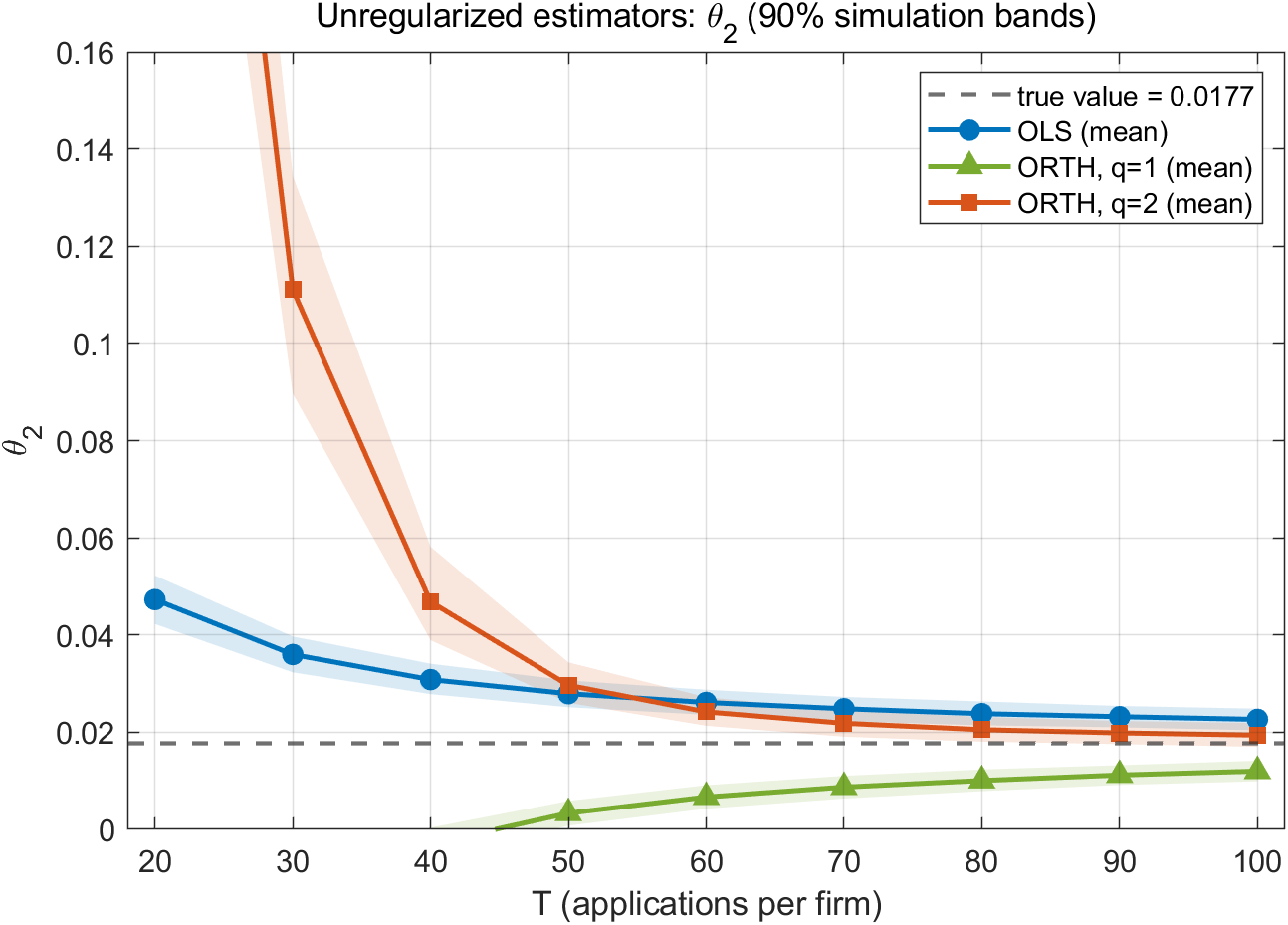}
        \end{tabular}
	\caption{Monte Carlo estimates of $\theta_2$ as a function of $T$. Solid lines show the mean across simulations for OLS and ORTH (at order $q=1$ and $q=2$); the dashed line is the true value; shaded regions are 90\% simulation bands.}
	\label{fig:theta2}
\end{figure}

Starting with $\theta_{0,1}$ -- the average of $\eta_{i10}$ -- the right graph in Figure \ref{fig:theta1} shows that the OLS and the orthogonal estimators are all close to unbiased irrespective of $T$. The left graph, in turn, shows that while OLS is slightly biased as a result of the regularization of $\Lambda_{i0}$, the orthogonal estimators at $q=1$ and $q=2$ remain virtually unbiased. In addition, the simulation bands show that the orthogonalization is associated with some increase in sampling variability.

Shifting attention to $\theta_{0,2}$ -- the average of $\eta_{i10}^2$ -- the right graph in Figure \ref{fig:theta2} shows that the three estimators are biased when one does not rely on a regularized estimator of $\Lambda_{i0}$. Moreover, the bias is larger in absolute value for the first-order orthogonal estimator compared to OLS, and the bias of the second-order orthogonal estimator only becomes lower than the bias of OLS when $T\geq 60$. Lastly, we observe a variance increase associated with the orthogonalization, especially for lower values of $T$. Hence, while orthogonalization reduces bias for large enough $T$, an estimator based on un-regularized estimates of $\Lambda_{i0}$ may perform poorly for lower values of $T$.

The left graph in Figure \ref{fig:theta2} shows that the situation improves greatly when relying on regularized estimates of $\Lambda_{i0}$. The bias of OLS is reduced relative to the un-regularized case, without a noticeable increase in variability. In turn, the bias of the first-order orthogonal estimator remains large, suggesting that first-order Neyman orthogonality is unable to correct incidental parameter bias. In contrast, the bias of the second-order orthogonal estimator decreases substantially relative to the un-regularized case, and it is lower than the one of OLS for all values of $T$. This suggests that the combination of orthogonal moments -- to reap the benefits of higher-order Neyman orthogonality -- and regularization of the nuisance parameters -- to ensure lower estimation error -- can be particularly appealing in applications.

\section{Orthogonal moments via von Mises expansions\label{sec:vonmises}}

We now provide a general approach to construct higher-order orthogonal moment functions in models defined by moment conditions. Our approach relies on von Mises expansions. We start by describing the general approach, and then explain how we apply it to our setting.

\subsection{The von Mises approach\label{sec:perturb}}

Throughout this section $P$ denotes distributions of a single observation $W_i$. Fix one of them as a reference and call it $\overline P$, and let $\mathcal{P}$ be a convex set of such distributions, each with a density relative to $\overline P$ that is square integrable under $\overline P$. Nothing below uses that the members of $\mathcal{P}$ are nonnegative, and the argument applies unchanged to finite signed measures. Moreover, the choice of $\overline P$ does not affect the expression of the orthogonal moments. 

For $P, Q \in \mathcal{P}$ write\footnote{This norm is convenient because $| \int f \, d(P-Q) | \leq \| f \|_{L^2(\overline P)} \| P - Q \|_{\overline P}$ by the Cauchy-Schwarz inequality, and the derivatives below are built from linear functionals of that kind.}
\begin{equation}
\| P - Q \|_{\overline P}
\;:=\; \biggl( \int \Bigl( \frac{dP}{d\overline P} - \frac{dQ}{d\overline P} \Bigr)^{2} d\overline P \biggr)^{1/2} .
\label{eq:Pnorm}
\end{equation}
Let $T(P)$ be a real-valued function of $P \in \mathcal{P}$, $q+1$ times continuously differentiable along the segments of $\mathcal{P}$ and with bounded derivative of order $q+1$.\footnote{For each $P \in \mathcal{P}$ and each $k \leq q+1$ there is a symmetric $k$-linear form $D^k T(P)$ on the differences of elements of $\mathcal{P}$ such that along every segment $P(s) = (1-s) \widehat P + s P_0$ in $\mathcal{P}$ we have $\frac{d^{\,k}}{ds^{\,k}} T(P(s)) = D^k T(P(s))[H, \ldots, H]$ with $H = P_0 - \widehat P$, and $s \mapsto D^{q+1} T(P(s))[H, \ldots, H]$ is continuous on $[0,1]$. Bounded means $M := \sup_{P} \sup_{\|H\|_{\overline P} \leq 1}$, the outer supremum over the segments on which $T$ is used, so over the ball $\{ P : \| P - P_0 \|_{\overline P} \leq c_{\widehat P}\, r \}$ once Lemma~\ref{lem:expansion} is applied, of $ | D^{q+1} T(P)[H, \ldots, H] | < \infty$, and~\eqref{eq:boundP} then holds with $C = M/(q+1)!$.} For $\widehat P \in \mathcal{P}$, and with the true distribution $P_0$ assumed to lie in $\mathcal{P}$ as well, so that the whole segment from $\widehat P$ to $P_0$ lies in $\mathcal{P}$, define
\begin{equation}
\Psi^{(q)}(\widehat P, P_0)
\;:=\; \sum_{k=0}^{q} \frac{1}{k!}\; \frac{d^{\,k}}{ds^{\,k}}\, T\bigl( (1-s)\, \widehat P + s\, P_0 \bigr) \Big|_{s=0} .
\label{eq:PsiP}
\end{equation}
Here, $P(s) := (1-s)\, \widehat P + s\, P_0$ interpolates between $\widehat P$ at $s = 0$ and $P_0$ at $s = 1$. Thus, $\Psi^{(q)}(\widehat P, P_0)$ is the Taylor polynomial of order $q$ of the function $s \mapsto T(P(s))$ around $s = 0$, evaluated at $s = 1$. Then Taylor's theorem with integral remainder guarantees existence of a constant $C < \infty$ such that
\begin{equation}
\bigl| \Psi^{(q)}(\widehat P, P_0) - T(P_0) \bigr|
\;\leq\; C\, \bigl\| P_0 - \widehat P \bigr\|_{\overline P}^{\,q+1} .
\label{eq:boundP}
\end{equation}

For the vector of nuisance parameters we now write $\nu$, which will be equal to $(\eta, \lambda)$ in the next subsection. Let $\nu \mapsto \widehat P(\nu)$ assign a distribution $\widehat P(\nu) \in \mathcal{P}$ to each $\nu$ with $\| \nu - \nu_0 \| \leq r$, for some $r > 0$, where $\nu_0$ is the true value of the nuisance parameter. We require that, for some $c_{\widehat P} < \infty$ and every $\nu$ such that $\| \nu - \nu_0 \| \leq r$,\footnote{The constant $c_{\widehat P}$ depends on $P_0$, $\overline P$, $q$ and the dimension of $\nu$, all of which are fixed here.}
\begin{equation}
\| \widehat P(\nu) - P_0 \|_{\overline P} \leq c_{\widehat P}\, \| \nu - \nu_0 \|.
\label{eq:lipschitz}
\end{equation}
The condition (\ref{eq:lipschitz}) implies that $\widehat P$ is equal to $P_0 $ at the truth, i.e., 
\begin{equation}
\widehat P(\nu_0) \;=\; P_0.
\label{eq:choice}
\end{equation}

\begin{lemma}[Expansion in the distribution of the data]
\label{lem:expansion}
Let $q \geq 0$ be an integer and let $r > 0$. Let $\mathcal{P}$, $\| \cdot \|_{\overline P}$ and $T$ be as above, and suppose that (\ref{eq:lipschitz}) holds. Let
$\Psi^{(q)}(\nu) := \Psi^{(q)}(\widehat P(\nu), P_0)$ be the expansion~\eqref{eq:PsiP} along the segment from $\widehat P(\nu)$ to $P_0$.
Then there is a constant $C < \infty$ such that
\begin{equation}
\bigl| \Psi^{(q)}(\nu) - T(P_0) \bigr|
\;\leq\; C\, \| \nu - \nu_0 \|^{\,q+1}
\label{eq:boundnu}
\end{equation}
for every $\nu$ with $\| \nu - \nu_0 \| \leq r$.
\end{lemma}

Notice that by construction, for given functional $T$ and distribution $\widehat P(\cdot)$, $\Psi^{(q)}(\nu)$ is a polynomial of degree at most $q$ in $P_0$. A polynomial of degree $L$ in $P_0$ is the expectation of a function of $L$ independent copies of the data, so $\Psi^{(q)}(\nu)$ is the population moment of a moment function in finitely many observations, and can be estimated by a U-statistic without bias.

Both the expansion~\eqref{eq:PsiP} and the remainder~\eqref{eq:boundnu} are known in the statistics literature, where the truncated sum is called a von Mises expansion. Expanding a function of the distribution along a segment goes back to \cite{vonMises1947}, and its first-order case is display (1.2) of \cite{vanderVaart2014}. \cite{RobinsLiTchetgenTchetgenvanderVaart2008} develop the order-$q$ version, which \cite{vanderVaart2014} states as a corrected plug-in with a bias of order $\| \nu - \nu_0 \|^{q+1}$. In that literature $\widehat P(\cdot)$ is the model, one distribution $P_\nu$ for each parameter value, and the expansion is taken at $P_\nu$, so that~\eqref{eq:choice} holds by construction.

Lastly, the presentation in this subsection is not specific to our moment condition model introduced in Section~\ref{sec:setup}. The von Mises expansion in \eqref{eq:PsiP} and Lemma~\ref{lem:expansion} are well-known results, presented here in a way that is convenient for our purposes. 

\subsection{Data perturbations\label{sec:vm-expansion}}

A key input to the von Mises approach is the family $\widehat P(\cdot)$. In our moment condition model, we require it to satisfy the two conditions
\begin{equation}
\mathbb{E}_{\widehat P(\eta,\lambda)}\bigl[ g(W_i, \theta_0, \eta) \bigr] = 0 ,
\qquad
\Lambda(\lambda)\, \mathbb{E}_{\widehat P(\eta,\lambda)}\bigl[ \partial_\eta g(W_i, \theta_0, \eta) \bigr] = I_{d_\eta}.
\label{eq:Phat-moments}
\end{equation}
In addition, we require the condition (\ref{eq:lipschitz}) to hold.

Since we work with the norm $\|\cdot\|_{\overline{P}}$ on $\mathcal{P}$ defined in (\ref{eq:Pnorm}), a natural construction defines $\widehat P(\eta,\lambda)$ as the distribution that is closest to $P_0$, with respect to $\|\cdot\|_{\overline{P}}$, subject to the moment conditions (\ref{eq:Phat-moments}) being satisfied. However, this particular family $\widehat P(\cdot)$ is impractical for our purpose in general. The reason is that the derivatives $\mathbb{E}_{\widehat P(\eta,\lambda)}\bigl[ \partial_\eta^p g(W_i, \theta_0, \eta) \bigr]$, for $p\geq 2$, appear as additional nuisance parameters in the expression of the orthogonal moments obtained by the von Mises expansion of Lemma \ref{lem:expansion}.

To address this challenge, given an orthogonality order $q$, our construction imposes the following additional restrictions, 
\begin{equation}
\mathbb{E}_{\widehat P(\eta,\lambda)}\bigl[ \partial_\eta^p g(W_i, \theta_0, \eta) \bigr] = \mathbb{E}_{ P_0}\bigl[ \partial_\eta^p g(W_i, \theta_0, \eta) \bigr],\text{ for all }p=2,3,...,q.
\label{eq:derivatives}
\end{equation}
When (\ref{eq:derivatives}) holds, the derivatives of order $2$ to $q$ of the moment under  $\widehat P(\cdot)$ coincide with the ones under $P_0$, so they can be estimated without bias, and so can polynomial functions of these derivatives using independent copies of the data, as anticipated in Section \ref{sec:example}. The zero- and first-order derivatives are pinned down by (\ref{eq:Phat-moments}). In turn, the derivatives of order $q+1$ and larger do not matter for achieving orthogonality to order $q$.

We then define $\widehat P(\cdot)$ as the solution to the following optimization problem
\begin{align}
		&\widehat P(\eta,\lambda)=\argmin_{P\in{\cal{P}}}\,  \|P-P_0\|_{\overline{P}}\quad \text{ s.t.}\quad \text{(\ref{eq:Phat-moments}) and (\ref{eq:derivatives})}.\label{eq:Phat}
	\end{align}
Lemma \ref{lem:projection} in Appendix \ref{app:projection} shows that, under a rank condition, the solution exists, that it is unique and available in closed form, and that, under an additional eigenvalue bound, it satisfies the condition (\ref{eq:lipschitz}).\footnote{With $\overline P = P_0$, the squared criterion in~\eqref{eq:Phat} is $\int (dP/dP_0 - 1)^2\, dP_0$, the $\chi^2$ divergence of $P$ from $P_0$ when $P$ is a probability measure. By Lemma~\ref{lem:projection}, the solution reweights $P_0$ by a function that is affine in the moment functions in the constraints. This is closely related to the implied probabilities of GMM \citep{BackBrown1993} and of generalized empirical likelihood with a quadratic criterion, which gives the continuously updated GMM estimator \citep{NeweySmith2004}.}

Importantly, denoting 
\begin{equation}
\overline G_P(\eta, \lambda) := \Lambda(\lambda)\, \mathbb{E}_{P}\bigl[ g(W_i, \theta_0, \eta) \bigr],\quad \overline G(\eta, \lambda):=\overline G_{P_0}(\eta, \lambda),\label{eq:Gbar}
\end{equation}
the moment condition at $\widehat{P}(\cdot)$ satisfies the key identity\footnote{For the final result on $\psi^{(q)}$, $\widehat P(\eta,\lambda)$ matters only through the identity~\eqref{eq:Gplug}, and we could modify the definition of $\widehat P(\eta,\lambda)$ in \eqref{eq:Phat} so that it uses $\overline G_P$ alone. Specifically, we could replace the first condition in~\eqref{eq:Phat-moments}
by $\Lambda(\lambda)\, \mathbb{E}_{\widehat P(\eta,\lambda)}[ g(W_i, \theta_0, \eta) ] = 0$,
and the conditions in~\eqref{eq:derivatives} by $\Lambda(\lambda)\, \mathbb{E}_{\widehat P(\eta,\lambda)}[ \partial^p_\eta g(W_i, \theta_0, \eta) ] = \Lambda(\lambda)\, \mathbb{E}_{P_0}[ \partial^p_\eta g(W_i, \theta_0, \eta) ]$ for $p = 2, \ldots, q$. In the just-identified case $d_g = d_\eta$ the matrix $\Lambda(\lambda)$ is invertible and the two definitions coincide. When $d_g > d_\eta$ the modified definition imposes fewer restrictions, but it still yields the same orthogonal moment function, provided all $\Lambda(\lambda)$ with $\lambda$ near $\lambda_0$ have the same null space, the condition stated after~\eqref{eq:nu-of-P}.}
\begin{equation}
\overline G_{\widehat P(\eta,\lambda)}(\widetilde\eta, \lambda)
\;=\; (\widetilde\eta - \eta) + \overline G(\widetilde\eta, \lambda) - \overline G(\eta, \lambda) - \partial_\eta \overline G(\eta, \lambda) (\widetilde\eta - \eta) + O\bigl( \| \widetilde\eta - \eta \|^{q+1} \bigr), 
\label{eq:Gplug}
\end{equation}
for $\eta,\widetilde{\eta}$ in a neighborhood of $\eta_0$ and all $\lambda$ in  a neighborhood of $\lambda_0$.  Note that $\widehat P(\eta,\lambda)$ satisfies (\ref{eq:Phat-moments}) and (\ref{eq:derivatives}), so that $\overline G_{\widehat P(\eta,\lambda)}(\cdot\,,\lambda)$ has value $0$ and first derivative $I_{d_\eta}$ at $\widetilde\eta = \eta$, and derivatives of order $2$ to $q$ equal to those of $\overline G(\cdot\,,\lambda)$. 

The expression of the moment condition (\ref{eq:Gplug}), together with the von Mises device, allows one to derive explicit orthogonal moments to any chosen order. As anticipated in Section \ref{sec:example} in a simple setting, the general construction relies on a polynomial approximation (via the expansion in Lemma \ref{lem:expansion}) and unbiased estimators of the components of the polynomial (through the use of multiple independent copies).

The approach applies Lemma \ref{lem:expansion} to the family $\widehat P(\cdot)$ constructed above that satisfies (\ref{eq:Gplug}), and to the functional\footnote{ 
Notice that in \eqref{eq:TP} the outer expectation is under $P_0$, not $P$. Taking $\mathbb{E}_{P}$ instead would make $\Psi^{(q)}$ depend on an expectation under $\widehat P(\eta,\lambda)$, from which we have no observations. Under $P_0$ it is estimable from the data, at the cost of raising the order $L$ of the U-statistic in~\eqref{UstatsPsi} by one.}
\begin{equation}
T(P)=\mathbb{E}_{P_0}\left[m(W_i,\theta_0,\eta(P))\right],\label{eq:TP}
\end{equation}
where $\eta(P)$ is the solution near $\eta_0$ of
\begin{equation}
\Lambda(\check{\lambda})\, \mathbb{E}_P\left[g(W_i,\theta_0,\eta(P))\right] \;=\; 0,
\label{eq:nu-of-P}
\end{equation}
for some $\check{\lambda}$ in a neighborhood of $\lambda_0$. Under correct specification, $\eta(P_0)=\eta_0$. The value of $\check{\lambda}$ only matters when the problem is over-identified, and it does not matter for the expression of the orthogonal moment implied by the von Mises expansion provided that all matrices $\Lambda(\lambda)$, for $\lambda$ near $\lambda_0$, have the same null space, since~\eqref{eq:nu-of-P} then has the same solution when $\Lambda(\check\lambda)$ is replaced by any such $\Lambda(\lambda)$. 

We compute orthogonal moments in the next two sections, starting with the special case where $g(W_i, \theta, \eta)$ is affine in the nuisance parameter $\eta$.

\begin{remark}[Comparison to parametric models] The need for the construction (\ref{eq:Phat}) arises from the fact that the moment conditions do not completely characterize the distribution of the data. In a parametric model with density $P_{\eta}$ one can simply take $\widehat P(\eta)=P_{\eta}$. In contrast, in our moment condition model the conditions (\ref{eq:Phat-moments}) do not uniquely pin down $\widehat P(\cdot)$, and our approach provides one family $\widehat P(\cdot)$ that additionally leads to feasible orthogonal estimating equations.
\end{remark}

\begin{remark}[Non-recursive construction] Our construction based on von Mises expansions, which yields closed-form expressions for orthogonal moment functions, differs from the recursive construction of higher-order influence functions \citep{pfanzagl1990estimation,RobinsLiTchetgenTchetgenvanderVaart2008}. Such a recursive construction, in a model defined by moment conditions, is developed in \citet{ChetverikovSorensenTsyvinski2026}.
\end{remark}

\begin{remark}[Expansion around the empirical distribution]\label{rem:empirical}
We can alternatively center the expansion~\eqref{eq:PsiP} at the empirical distribution $\widehat P_n$ of the part of the sample that is used for the preliminary estimators, which is independent of the observations in the U-statistic~\eqref{UstatsPsi}. For the functional~\eqref{eq:TP}, the expansion then depends on $\widehat P_n$ only through finitely many quantities, namely the root $\widehat\eta$ of~\eqref{eq:nu-of-P} at $P = \widehat P_n$, a normalized left inverse of the Jacobian under $\widehat P_n$, and the expected derivatives $\mathbb{E}_{\widehat P_n}[\partial^p_\eta g(W_i, \theta_0, \widehat\eta)]$ of orders $2 \leq p \leq q$. This amounts to dropping the restrictions~\eqref{eq:derivatives} and treating these quantities as nuisance parameters, which again gives a moment function that is orthogonal to order $q$. The restrictions~\eqref{eq:derivatives} avoid these additional nuisance parameters, so that the nuisance parameter remains $(\eta, \lambda)$ for every $q$, at the cost of additional independent copies of the data. See Appendix~\ref{app:empirical} for details.
\end{remark}

\section{Moment conditions affine in the nuisance parameter\label{sec:affine}}

We assume in this section that the moment function $g(W_i, \theta, \eta)$ is affine in $\eta$, so that $\partial_\eta g(W_i, \theta, \eta)$ does not depend on $\eta$, whereas the target moment function $m(W_i, \theta, \eta)$ can be nonlinear in $\eta$. This covers many important cases, such as linear models with heterogeneous coefficients (Example~\ref{ex:grouped}) and linear regression and linear instrumental variables (Example~\ref{ex:iv}). We derive the orthogonal moment functions by computing the von Mises expansion explicitly, and then formally verify that the resulting functions are indeed orthogonal.

\subsection{A general expression}

Fix $(\eta,\lambda)$ near $(\eta_0,\lambda_0)$. When $g$ is affine in $\eta$, \eqref{eq:Gplug} simplifies to
\begin{equation}
\overline G_{\widehat P(\eta,\lambda)}(\widetilde\eta,\lambda) \;=\; \widetilde\eta - \eta
\qquad \text{for all } \widetilde\eta .
\label{eq:Gplug-affine}
\end{equation}
We will use the shorthand $\overline G := \overline G(\eta,\lambda)= \Lambda(\lambda)\, \mathbb{E}\bigl[g(W_i,\theta_0,\eta)\bigr]$, and define
\begin{equation}
\overline{H} := \Lambda(\lambda)\, \mathbb{E}\bigl[\partial_\eta g(W_i,\theta_0,\eta)\bigr] - I_{d_\eta}.
\label{eq:GH}
\end{equation}
Along the segment $P(s) = (1-s)\widehat P(\eta,\lambda) + s P_0$, we have
\[
\overline G_{P(s)}(\widetilde\eta,\lambda)
= (I_{d_\eta} + s\overline H)(\widetilde\eta - \eta) + s\,\overline G .
\]
Now,  provided all matrices $\Lambda(\lambda)$ with $\lambda$ near $\lambda_0$ have the same null space, as in Subsection~\ref{sec:vm-expansion}, $\eta(P(s))$ in~\eqref{eq:nu-of-P} solves $\overline G_{P(s)}(\eta(P(s)),\lambda) = 0$. Hence, 
\begin{equation}
\eta(P(s)) = \eta - s\,(I_{d_\eta} + s\overline H)^{-1}\,\overline G
= \eta + \sum_{j \geq 1} s^j\, \delta_j,
\qquad
\delta_j := (-1)^j \overline H^{j-1}\, \overline G ,
\label{eq:affine-path}
\end{equation}
where the series converges for all $s \in [0,1]$ when $\|\overline H\| < 1$, which holds for $(\eta,\lambda)$ close enough to $(\eta_0,\lambda_0)$.

The functional~\eqref{eq:TP} along the segment is $T(P(s)) = \mathbb{E}[m(W_1,\theta_0,\eta(P(s)))]$. Its Taylor polynomial of degree $q$ at $s = 0$ is obtained by substituting~\eqref{eq:affine-path} into the Taylor expansion of $m$ around $\eta$ and discarding all powers of $s$ above $q$, which when evaluated at $s=1$ gives
\begin{equation}
\Psi^{(q)}(\theta_0,\eta,\lambda) = \sum_{r=0}^{q} \frac{1}{r!}
\sum_{\substack{j_1, \ldots, j_r \geq 1 \\ j_1 + \cdots + j_r \leq q}}
\Bigl\langle \mathbb{E}\bigl[\partial_\eta^r m(W_1, \theta_0, \eta)\bigr],\, \delta_{j_1} \otimes \cdots \otimes \delta_{j_r} \Bigr\rangle,
\label{eq:Psi-delta}
\end{equation}
where $\partial_\eta^r m(W_i, \theta, \eta)$ denotes the $r$-th derivative tensor of $m$ with respect to $\eta$, and $\langle\cdot,\cdot\rangle$ denotes the multilinear action of a derivative tensor.\footnote{For vectors $v^{(1)}, \ldots, v^{(r)} \in \mathbb{R}^{d_\eta}$, 
\[
\langle \partial_\eta^r m(W_i, \theta, \eta),\, v^{(1)} \otimes \cdots \otimes v^{(r)} \rangle
= \sum_{j_1, \ldots, j_r} \frac{\partial^r m(W_i, \theta, \eta)}{\partial \eta_{j_1} \cdots \partial \eta_{j_r}} v^{(1)}_{j_1} \cdots v^{(r)}_{j_r}.
\]
For $r = 1$ this gives the directional derivative $\partial_\eta m(W_i, \theta, \eta)' v^{(1)}$, while for $r = 2$ this gives the quadratic form $v^{(1)\prime} \partial_{\eta\eta'}^2 m(W_i, \theta, \eta)\, v^{(2)}$.}

The terms $\delta_j = (-1)^j \overline H^{j-1}\overline G$ are products of population moments, and can therefore be estimated without bias using independent copies of the data. Doing so, and collecting terms, then implies (see Appendix \ref{app:affine} for a complete derivation) 
\begin{equation}\Psi^{(q)}(\theta_0,\eta,\lambda) = \mathbb{E}[\psi^{(q)}(W;\theta_0,\eta,\lambda)],\label{eq:Psiq}\end{equation} for the moment function
\begin{align}
\psi^{(q)}(W; \theta, \eta, \lambda)
&:= \sum_{r=0}^{q} \frac{1}{r!} \hspace{-.20cm}
\sum_{\substack{k_1, \ldots, k_r \geq 1 \\ k_1 + \cdots + k_r \leq q}}
\hspace{-.5cm}
(-1)^{\sum_{b=1}^r k_b} \binom{q}{\sum_{b=1}^r k_b}
\left\langle \partial_\eta^r m(W_1, \theta, \eta),\,
\bigotimes_{b=1}^{r} Z_{k_b}(W^{(b)}; \theta, \eta, \lambda)
\right\rangle,
\label{eq:score}
\end{align}
where, for all $k \geq 1$, we have defined the factor
\begin{equation}
Z_{k}(W^{(b)}; \theta, \eta, \lambda)
= \left( \prod_{\ell=1}^{k - 1} \Lambda(\lambda)\, \partial_\eta g(W_{(b,\ell)}, \theta, \eta) \right) \Lambda(\lambda)\, g(W_{(b,k)}, \theta, \eta) ,
\label{eq:Zfactor}
\end{equation}
which uses a block $W^{(b)} = (W_{(b,1)}, \ldots, W_{(b,k)})$ of $k$ independent copies of $W_i$. In~\eqref{eq:score}, $W = (W_1, \ldots, W_{q+1})$ consists of $q+1$ independent copies of $W_i$: the copy $W_1$ enters through $m$, and in each term the blocks $W^{(1)}, \ldots, W^{(r)}$ are formed from distinct copies among $W_2, \ldots, W_{q+1}$.\footnote{How copies are assigned to blocks does not affect the expectation of $\psi^{(q)}$, and the U-statistic~\eqref{UstatsPsi} averages over all assignments. When $m$ does not depend on the data, $W_1$ is not needed and $\psi^{(q)}$ depends on $q$ copies only.}

\subsection{Explicit formulas for small $q$}

Equation \eqref{eq:score} provides a general, closed-form expression for an order-$q$ orthogonal moment function in models where moment conditions are affine in the nuisance parameter. To make it concrete, we next display $\psi^{(q)}$ for $q = 0, 1, 2$. Write $m := m(W_1, \theta, \eta)$, $m_\eta := \partial_\eta m(W_1, \theta, \eta)$, and $m_{\eta\eta} := \partial_{\eta\eta'}^2 m(W_1, \theta, \eta)$.

\paragraph{Case $q = 0$.}
\[
\psi^{(0)}(W_1; \theta, \eta) = m(W_1, \theta, \eta).
\]

\paragraph{Case $q = 1$.}
\[
\psi^{(1)}(W_1, W_2; \theta, \eta, \lambda) = m - m_\eta' \,\Lambda(\lambda)\, g(W_2, \theta, \eta).
\]

\paragraph{Case $q = 2$.}
\begin{equation}
\begin{split}
\psi^{(2)}(W_1, W_2, W_3; \theta, \eta, \lambda)
&= m - 2\, m_\eta' \,\Lambda(\lambda)\, g(W_2, \theta, \eta) \\
&\quad + m_\eta' \,\big(\Lambda(\lambda) \,\partial_\eta g(W_3, \theta, \eta)\big)\, \Lambda(\lambda)\, g(W_2, \theta, \eta) \\
&\quad + \tfrac{1}{2} \big(\Lambda(\lambda)\, g(W_2, \theta, \eta)\big)'\, m_{\eta\eta} \,\big(\Lambda(\lambda)\, g(W_3, \theta, \eta)\big).
   \label{AffinePsi2}
\end{split}
\end{equation}

\setcounter{example}{0}
\begin{example}[Heterogeneous coefficients, continued]
\label{ex:grouped-q2}

Returning to model \eqref{eq:RC-model} with a scalar covariate ($d_\eta = 1$), we now apply (\ref{AffinePsi2}) to obtain a second-order orthogonal moment function ($q = 2$). Recall that the Jacobian inverse $\Lambda_{i0} = -(\mathbb{E}[X_{it}^2])^{-1}$ enters as an additional nuisance parameter, and that the target moment function $m(\eta_i) - \theta$ does not depend on the observation, so the copy $W_1$ in~\eqref{eq:score} is absent and $\psi^{(2)}$ uses $L = 2$ observations of the same unit.

With $g(W_{it}, \theta, \eta_i) = X_{it}(Y_{it} - X_{it}\eta_i)$, the factors \eqref{eq:Zfactor} are
\[
Z_1(W_{it}; \theta, \eta_i, \Lambda_i) = \Lambda_i (X_{it} Y_{it} - \eta_i X_{it}^2),
\qquad
Z_2(W_{is}, W_{it}; \theta, \eta_i, \Lambda_i) = - \Lambda_i X_{is}^2 \, \Lambda_i (X_{it} Y_{it} - \eta_i X_{it}^2).
\]
Substituting into \eqref{AffinePsi2} gives
\begin{align}
\psi(W_i; \theta, \eta_i, \Lambda_i) &= m(\eta_i) - 2\, \partial_\eta m(\eta_i)\, \Lambda_i(X_{i1} Y_{i1} - \eta_i\, X_{i1}^2) \notag\\
&\quad - \partial_\eta m(\eta_i)\, \Lambda_i^2\, X_{i2}^2\,(X_{i1} Y_{i1} - \eta_i\, X_{i1}^2) \notag\\
&\quad + \tfrac{1}{2}\, \partial_\eta^2 m(\eta_i)\, \Lambda_i(X_{i1} Y_{i1} - \eta_i\, X_{i1}^2) \cdot \Lambda_i(X_{i2} Y_{i2} - \eta_i\, X_{i2}^2) - \theta,
\label{eq:psi-covariate2}
\end{align}
which coincides with (\ref{eq:psi-covariate}). 

\end{example}
\setcounter{example}{2}

\begin{remark}[Comparison with standard first-order orthogonality]
For $q = 1$, our moment function reduces to $\psi^{(1)} = m(W_1, \theta, \eta) - \partial_\eta m(W_1, \theta, \eta)' \Lambda(\lambda)\, g(W_2, \theta, \eta)$, which uses two independent copies and has the Jacobian inverse $\Lambda$ as nuisance parameter. The standard first-order Neyman-orthogonal moment function \citep{Newey1994,ChernozhukovChetverikovDemirerDufloHansenNeweyRobins2018} is $\psi_{\rm std}^{(1)}(W_i; \theta, \eta, \alpha) = m(W_i, \theta, \eta) - \alpha'\, g(W_i, \theta, \eta)$, which uses a single observation and has a $d_g$-dimensional nuisance vector $\alpha_0 = \Lambda_0'\, \partial_\eta \mathbb{E}[m(W_i, \theta_0, \eta_0)]$. The standard moment function is simpler. It combines $\Lambda_0$ and $\partial_\eta \mathbb{E}[m]$ into a single nuisance $\alpha_0$ and avoids the second copy. Our construction keeps the two separate. The derivative $\partial_\eta \mathbb{E}[m]$ is not a nuisance parameter, because $\partial_\eta m(W_1, \theta, \eta)$ estimates it without bias and $W_1$ is independent of the copy $W_2$ in $g$. With this separation the construction extends to every $q \geq 2$, with the same nuisance parameter $\lambda$ at every order.
\end{remark}

\subsection{Verifying orthogonality} Finally, we directly check that (\ref{eq:score}) is orthogonal to order $q$. For this purpose, we first state regularity conditions for the affine case.

\begin{assumption}
\label{ass:affine}
\leavevmode
\begin{enumerate}[(i)]
\item $g(W_i, \theta, \eta)$ is affine in $\eta$.
\item $m(W_i, \theta, \eta)$ is $(q+1)$-times continuously differentiable in $\eta$ in a neighborhood of $\eta_0$.
\item The map $\lambda \mapsto \Lambda(\lambda)$ in Assumption~\ref{ass:jacobian-inverse} is $(q+1)$-times continuously differentiable in a neighborhood of $\lambda_0$.
\item For all $r \leq q+1$, there is a neighborhood $\mathcal{N}$ of $\eta_0$ and an integrable function $M_r(W_i)$ with $\sup_{\eta \in \mathcal{N}} \|\partial_\eta^r m(W_i, \theta_0, \eta)\| \leq M_r(W_i)$ almost surely.
 \item $\mathbb{E}[\|g(W_i, \theta_0, \eta_0)\|] < \infty$ and $\mathbb{E}[\|\partial_\eta g(W_i, \theta_0, \eta_0)\|] < \infty$.
\end{enumerate}
\end{assumption}

\begin{theorem}[$q$-th order orthogonality, affine case]
\label{th:main}
Define the population moment associated with $\psi^{(q)}$ in \eqref{eq:score} by
$
\Psi^{(q)}(\theta, \eta, \lambda) := \mathbb{E}[\psi^{(q)}(W_1, \ldots, W_L; \theta, \eta, \lambda)],
$
where $L = q+1$, or $L = q$ when $m$ does not depend on the data. Let $q \geq 1$. Under Assumptions~\ref{ass:iid}, \ref{ass:jacobian-inverse} and~\ref{ass:affine}, the moment function $\psi^{(q)}$ is $q$-th order Neyman-orthogonal in the sense of Definition~\ref{def:orthogonality}.
\end{theorem}

\begin{corollary}
\label{cor:derivatives-affine}
Suppose in addition to the assumptions of Theorem~\ref{th:main} that parts (ii) and (iv) of
Assumption~\ref{ass:affine} hold with $q+1$ replaced by $2q$. Then
$(\eta,\lambda) \mapsto \Psi^{(q)}(\theta_0,\eta,\lambda)$ is $q$ times differentiable at
$(\eta_0,\lambda_0)$, and~\eqref{eq:ortho} holds.
\end{corollary}

\section{Nonlinear moment conditions\label{sec:nonlinear}}

When the moment function $g(W_i, \theta, \eta)$ is nonlinear in $\eta$,
additional correction terms are needed beyond the affine orthogonal moment
function. The target $m$ was already allowed to be nonlinear in the previous
section, but we now also allow for nonlinearity of $g$. We denote the
affine orthogonal moment function from~\eqref{eq:score} by
$\psi^{(q)}_{\rm aff}$ and write the corrected orthogonal moment function as
\begin{equation}
\psi^{(q)}(W, \theta, \eta, \lambda)
= \psi^{(q)}_{\rm aff}(W, \theta, \eta, \lambda) + \text{correction terms}.
\label{eq:nonlinear-form}
\end{equation}
The reason $\psi^{(q)}_{\rm aff}$ fails to be orthogonal for nonlinear $g$ is that its derivation in Section~\ref{sec:affine} relies on~\eqref{eq:Gplug-affine}, that is, on $\Lambda(\lambda)\, \mathbb{E}[g(W_i, \theta_0, \widetilde\eta)]$ being affine in $\widetilde\eta$. When $g$ is nonlinear, this expression
contains quadratic and higher-order terms in $\delta$ that the affine
construction does not cancel. The correction terms involve the higher-order
derivatives $\partial_\eta^p g(W_i, \theta, \eta)$ for $p \geq 2$. For $q = 1$ no correction is needed: $\psi^{(1)} = \psi^{(1)}_{\rm aff}$ is first-order Neyman-orthogonal even when $g$ is nonlinear.

We first set up notation for the higher derivatives of $g$, used throughout this section. For $p \geq 1$, the $p$-th derivative tensor $\partial_\eta^p g(W_i, \theta, \eta) \in \mathbb{R}^{d_g} \otimes (\mathbb{R}^{d_\eta})^{\otimes p}$ has one $d_g$-output slot and $p$ input slots in $\mathbb{R}^{d_\eta}$, with entries $[\partial_\eta^p g]_{k, j_1, \ldots, j_p} = \partial^p g_k(W_i, \theta, \eta) / \partial \eta_{j_1} \cdots \partial \eta_{j_p}$. For $p = 1$ this is the $d_g \times d_\eta$ Jacobian. To express contractions of such tensors with vectors in $\mathbb{R}^{d_\eta}$, we extend the tensor notation of Section~\ref{sec:affine} as follows: for a tensor $A$ with $p$ input slots in $\mathbb{R}^{d_\eta}$ and output in a vector space $V$, and for $v_1, \ldots, v_p \in \mathbb{R}^{d_\eta}$,
\[
A[v_1, \ldots, v_p] \;\in\; V
\]
denotes the contraction of $A$ along its inputs, leaving the $V$-output unchanged. For the scalar function $m$, the bracket and the angle-bracket of Section~\ref{sec:affine} agree, $\partial_\eta^r m\,[v_1, \ldots, v_r] = \langle \partial_\eta^r m, v_1 \otimes \cdots \otimes v_r \rangle \in \mathbb{R}$. For the vector-valued $g$, $\partial_\eta^p g\,[v_1, \ldots, v_p] \in \mathbb{R}^{d_g}$, and pre-multiplying by $\Lambda(\lambda)$ contracts the $d_g$-output slot to give a vector in $\mathbb{R}^{d_\eta}$.

In this notation, the affine moment function~\eqref{eq:score} reads
\begin{equation}
\psi^{(q)}_{\rm aff}(W; \theta, \eta, \lambda)
= \sum_{r=0}^{q} \frac{1}{r!} \hspace{-.20cm}
\sum_{\substack{k_1, \ldots, k_r \geq 1 \\ k_1 + \cdots + k_r \leq q}}
\hspace{-.5cm}
(-1)^{\sum_{b=1}^r k_b} \binom{q}{\sum_{b=1}^r k_b}\,
\partial_\eta^r m(W_1, \theta, \eta)\bigl[ Z_{k_1}, \ldots, Z_{k_r} \bigr] ,
\label{eq:score-bracket}
\end{equation}
with the factors $Z_{k_b} = Z_{k_b}(W^{(b)}; \theta, \eta, \lambda)$ of~\eqref{eq:Zfactor}.

\subsection{Second order orthogonality\label{sec:second-order}}

For illustration purposes, we first give the full formula for $\psi^{(q)}$ at $q = 2$, the smallest order at which an additional correction term is needed beyond $\psi^{(q)}_{\rm aff}(W; \theta, \eta, \lambda)$. The complete formula for $\psi^{(2)}$ is as follows:
\begin{align}
\psi^{(2)}(W_1, & W_2,  W_3,  W_4, \theta, \eta, \lambda)
\nonumber \\
&= m(W_1, \theta, \eta) - 2\, \partial_\eta m(W_1, \theta, \eta)'\, \Lambda(\lambda)\, g(W_2, \theta, \eta) \notag \\
&\quad + \partial_\eta m(W_1, \theta, \eta)'\, \big(\Lambda(\lambda)\, \partial_\eta g(W_3, \theta, \eta)\big)\, \Lambda(\lambda)\, g(W_2, \theta, \eta) \notag \\
&\quad + \tfrac{1}{2}\, \big(\Lambda(\lambda)\, g(W_2, \theta, \eta)\big)'\, \partial_{\eta\eta'}^2 m(W_1, \theta, \eta)\, \big(\Lambda(\lambda)\, g(W_3, \theta, \eta)\big) \notag \\
&\quad - \tfrac{1}{2}\, \partial_\eta m(W_1, \theta, \eta)'\, \Lambda(\lambda)\, \partial_\eta^2 g(W_2, \theta, \eta)\big[\Lambda(\lambda)\, g(W_3, \theta, \eta),\, \Lambda(\lambda)\, g(W_4, \theta, \eta)\big] \notag \\
&= \psi^{(2)}_{\rm aff}(W_1, W_2, W_3, \theta, \eta, \lambda) \notag \\
&\quad - \tfrac{1}{2}\, \partial_\eta m(W_1, \theta, \eta)'\, \Lambda(\lambda)\, \partial_\eta^2 g(W_2, \theta, \eta)\big[\Lambda(\lambda)\, g(W_3, \theta, \eta),\, \Lambda(\lambda)\, g(W_4, \theta, \eta)\big].
\label{eq:psi2-nonlinear}
\end{align}

Only the last term is new relative to the affine case. It is, up to its pre-factor, the simplest term one can write down that involves $\partial_\eta^2 g$ and is invariant under linear reparametrizations of $g$ and of $\eta$.  The pre-factor $-\tfrac12$ is then pinned down by imposing the orthogonality conditions of Proposition~\ref{prop:reconcile} at $q = 2$ and solving for it. It is worth noting that the additional term raises the number of independent copies of the data from $L = 3$ to $L = 4$ compared to the affine case, i.e. the moment function is evaluated at four observations drawn independently from the sample.%

Of course, \eqref{eq:psi2-nonlinear} has so far only been asserted. Deriving it, and finding the corresponding terms at every order $q$, is the goal of the remainder of this section.

\subsection{Rooted trees and the general formula\label{sec:trees}}

As in the affine case we rely on a von Mises expansion. Recall the key identity (\ref{eq:Gplug}) that our family of distributions $\widehat{P}(\cdot)$ satisfies, written here without a remainder with a slight abuse of notation;
\begin{equation}
\overline G_{\widehat P(\eta,\lambda)}(\widetilde\eta, \lambda)
\;=\; \widetilde\eta - \eta
\;+\; \overline G(\widetilde\eta, \lambda) - \overline G(\eta, \lambda) - \partial_\eta \overline G(\eta, \lambda)\, (\widetilde\eta - \eta).
\label{eq:Gexact}
\end{equation}
Note that the remainder does not affect derivatives of order $q$ or less at $\widetilde\eta = \eta$, and only those enter the expansion to order $q$. 

With~\eqref{eq:nu-of-P} and the target $T(P)$ given by (\ref{eq:TP}), equation~\eqref{eq:Gexact} and the von Mises expansion~\eqref{eq:PsiP} determine $\Psi^{(q)}(\eta,\lambda)$ at every order $q$. Evaluating each expectation under $P_0$ at its own copy of the data then gives an orthogonal moment function $\psi^{(q)}$. Our main result is a closed form for $\psi^{(q)}$ at arbitrary $q$. 

Along the segment  $P(s)=(1-s)\widehat P(\eta,\lambda)+sP_0$, we have
\begin{equation}
(1-s)\, \overline G_{\widehat P(\eta,\lambda)}\bigl( \eta(P(s)), \lambda \bigr)
\;+\; s\, \overline G\bigl( \eta(P(s)), \lambda \bigr) \;=\; 0 ,
\label{eq:root}
\end{equation}
Differentiating~\eqref{eq:root} produces products of derivatives of $m$ and of $\Lambda\, g$, each factor contracted with the factor whose differentiation produced it. Such a contraction pattern is a \emph{rooted tree}.

\subsubsection{Indexing the terms by rooted trees\label{sec:trees-index}}

We illustrate this on $\psi^{(2)}$ in~\eqref{eq:psi2-nonlinear}. Figure~\ref{fig:trees-q2} shows how the five terms of $\psi^{(2)}$ can each be represented by a rooted tree built from three kinds of nodes. The \emph{root} (shaded blue) carries either $m$ or one of its $\eta$-derivatives $\partial_\eta^r m$, where $r$ is the number of children of the root. The \emph{leaves} (shaded gray) each carry $\Lambda\, g$. The \emph{non-root, non-leaf nodes} (unshaded) each carry $\Lambda\, \partial_\eta^p g$, where $p$ is the number of children of the node. Each node is evaluated on its own independent copy of $W_i$. The tree, which we denote by $\tau$, then uniquely determines the corresponding term in $\psi^{(2)}$, which we denote by $\kappa_\tau(W, \theta, \eta, \lambda)$.

\begin{figure}[tb]
\centering
\renewcommand{\arraystretch}{1.6}
\begin{tabular}{@{}c@{\,}|@{\,}c@{\,}|@{\,}c@{\,}|@{\;}c}
\textbf{Rooted tree $\tau$} & \textbf{$d(\tau)$} & \textbf{$\tau$ with factors} & \textbf{$\kappa_\tau(W,\theta,\eta,\lambda)$} \\
\hline
\begin{tikzpicture}[baseline=(root.base),
    every node/.style={draw, circle, inner sep=2pt, minimum size=5pt}]
\node[fill=blue!15] (root) {};
\path (root.south) ++(0,-0.8em) coordinate (strut);
\end{tikzpicture}
&
$0$
&
\begin{tikzpicture}[baseline=(root.base),
    every node/.style={draw, rounded corners, inner sep=4pt, font=\small}]
\node[fill=blue!10] (root) {$m(W_1)$};
\path (root.south) ++(0,-0.8em) coordinate (strut);
\end{tikzpicture}
&
{ $ m(W_1, \theta, \eta)$} \\
\hline
\begin{tikzpicture}[baseline=(root.base),
    level distance=0.7cm, sibling distance=1cm,
    every node/.style={draw, circle, inner sep=2pt, minimum size=5pt},
    edge from parent/.style={draw, -}]
\node[fill=blue!15] (root) {} child { node[fill=gray!30] (leaf) {} };
\path (leaf.south) ++(0,-0.8em) coordinate (strut);
\end{tikzpicture}
&
$1$
&
\begin{tikzpicture}[baseline=(root.base),
    level distance=1.0cm, sibling distance=2.5cm,
    every node/.style={draw, rounded corners, inner sep=3pt, font=\small},
    edge from parent/.style={draw, -}]
\node[fill=blue!10] (root) {$\partial_\eta m(W_1)$}
  child { node[fill=gray!15] (leaf) {$\Lambda\, g(W_2)$} };
\path (leaf.south) ++(0,-0.8em) coordinate (strut);
\end{tikzpicture}
&
{$\partial_\eta m(W_1, \theta, \eta)'\, \Lambda(\lambda)\, g(W_2, \theta, \eta)$} \\
\hline
\begin{tikzpicture}[baseline=(root.base),
    level distance=0.55cm, sibling distance=1cm,
    every node/.style={draw, circle, inner sep=2pt, minimum size=5pt},
    edge from parent/.style={draw, -}]
\node[fill=blue!15] (root) {}
  child { node {}
    child { node[fill=gray!30] (leaf) {} }
  };
\path (leaf.south) ++(0,-0.8em) coordinate (strut);
\end{tikzpicture}
&
$2$
&
\begin{tikzpicture}[baseline=(root.base),
    level distance=1.0cm, sibling distance=2.5cm,
    every node/.style={draw, rounded corners, inner sep=3pt, font=\small},
    edge from parent/.style={draw, -}]
\node[fill=blue!10] (root) {$\partial_\eta m(W_1)$}
  child { node {$\Lambda\, \partial_\eta g(W_3)$}
    child { node[fill=gray!15] (leaf) {$\Lambda\, g(W_2)$} }
  };
\path (leaf.south) ++(0,-0.8em) coordinate (strut);
\end{tikzpicture}
&
{\footnotesize $\partial_\eta m(W_1, \theta, \eta)'\, \big(\Lambda(\lambda)\, \partial_\eta g(W_3, \theta, \eta)\big)\, \Lambda(\lambda)\, g(W_2, \theta, \eta)$} \\
\hline
\begin{tikzpicture}[baseline=(root.base),
    level distance=0.7cm, sibling distance=0.9cm,
    every node/.style={draw, circle, inner sep=2pt, minimum size=5pt},
    edge from parent/.style={draw, -}]
\node[fill=blue!15] (root) {}
  child { node[fill=gray!30] (leafL) {} }
  child { node[fill=gray!30] (leafR) {} };
\path (leafL.south) ++(0,-0.8em) coordinate (strut);
\end{tikzpicture}
&
$2$
&
\begin{tikzpicture}[baseline=(root.base),
    level distance=1.0cm, sibling distance=2.5cm,
    every node/.style={draw, rounded corners, inner sep=3pt, font=\small},
    edge from parent/.style={draw, -}]
\node[fill=blue!10] (root) {$\partial_{\eta\eta'}^2 m(W_1)$}
  child { node[fill=gray!15] (leafL) {$\Lambda\, g(W_2)$} }
  child { node[fill=gray!15] (leafR) {$\Lambda\, g(W_3)$} };
\path (leafL.south) ++(0,-0.8em) coordinate (strut);
\end{tikzpicture}
&
{\footnotesize $\big(\Lambda(\lambda)\, g(W_2, \theta, \eta)\big)'\, \partial_{\eta\eta'}^2 m(W_1, \theta, \eta)\, \big(\Lambda(\lambda)\, g(W_3, \theta, \eta)\big)$} \\
\hline
\begin{tikzpicture}[baseline=(root.base),
    level distance=0.55cm, sibling distance=0.9cm,
    every node/.style={draw, circle, inner sep=2pt, minimum size=5pt},
    edge from parent/.style={draw, -}]
\node[fill=blue!15] (root) {}
  child { node {}
    child { node[fill=gray!30] (leafL) {} }
    child { node[fill=gray!30] (leafR) {} }
  };
\path (leafL.south) ++(0,-0.8em) coordinate (strut);
\end{tikzpicture}
&
$2$
&
\begin{tikzpicture}[baseline=(root.base),
    level distance=1.0cm, sibling distance=2.5cm,
    every node/.style={draw, rounded corners, inner sep=3pt, font=\small},
    edge from parent/.style={draw, -}]
\node[fill=blue!10] (root) {$\partial_\eta m(W_1)$}
  child { node {$\Lambda\, \partial_\eta^2 g(W_2)$}
    child { node[fill=gray!15] (leafL) {$\Lambda\, g(W_3)$} }
    child { node[fill=gray!15] (leafR) {$\Lambda\, g(W_4)$} }
  };
\path (leafL.south) ++(0,-0.8em) coordinate (strut);
\end{tikzpicture}
&
{\footnotesize
\parbox{0.5\textwidth}{\raggedright
\phantom{a}\\
$\partial_\eta m(W_1, \theta, \eta)'$ \\
$\times  \Lambda(\lambda)\, \partial_\eta^2 g(W_2, \theta, \eta) \big[\Lambda(\lambda)\, g(W_3, \theta, \eta),\, \Lambda(\lambda)\, g(W_4, \theta, \eta)\big]$}} \\
\hline
\end{tabular}
\caption{The five terms of $\psi^{(2)}$ in~\eqref{eq:psi2-nonlinear}, indexed by rooted trees. The root (blue) carries $m$ or its $\eta$-derivatives. Non-root non-leaf nodes carry $\Lambda\,\partial_\eta^p g$ where $p$ is the node's number of children. Leaves (gray) carry $\Lambda g$.}
\label{fig:trees-q2}
\end{figure}
All terms in the orthogonal moment function $\psi^{(q)}$ can be indexed by rooted trees in the same way for any order $q \in \{1, 2, 3, \ldots\}$,
\begin{align}
\psi^{(q)}(W, \theta, \eta, \lambda) \;=\; \sum_{\tau \in \mathcal{T}_q} c_{q,\tau}\; \kappa_\tau(W, \theta, \eta, \lambda),
\label{eq:general-psi}
\end{align}
where $\mathcal{T}_q$ is a finite set of rooted trees, $c_{q,\tau} \in \mathbb{R}$ is a combinatorial coefficient, and the term $\kappa_\tau(W, \theta, \eta, \lambda)$ is determined by $\tau$ via the rules illustrated in Figure~\ref{fig:trees-q2}.

A \emph{rooted tree} $\tau$ is a finite tree (a connected acyclic graph) with one node singled out as the \emph{root}. Let $d(\tau)$ denote the number of non-root nodes of $\tau$ with at most one child. The middle column of Figure~\ref{fig:trees-q2} reports $d(\tau)$ for each of the five trees: the trivial tree has $d = 0$, the root with one leaf child has $d = 1$, and the remaining three trees all have $d = 2$. The set $\mathcal{T}_q$ in~\eqref{eq:general-psi} is then
\begin{equation}
\mathcal{T}_q \;:=\; \big\{\tau \text{ rooted tree} \;:\; d(\tau) \leq q\big\}.
\label{eq:Tq-def}
\end{equation}%
 The definition of $\mathcal{T}_q$  governs which terms $\kappa_\tau$ appear in  $\psi^{(q)}$,
which follows from the von Mises expansion that defines $\Psi^{(q)}$ (and thus $\psi^{(q)}$) above.
The relevant terms have $d(\tau) \leq q$, where the count $d(\tau)$ includes the leaves and the nodes with one child,
which carry $\Lambda\, g$ and $\Lambda\, \partial_\eta g$, but excludes the nodes with two or more children, which carry
$\Lambda\, \partial^p_\eta g$ with $p \geq 2$. The reason is that each factor $\Lambda\, g$ or $\Lambda\, \partial_\eta g$ raises the order of a term in the
expansion by one, as in the affine case of Section~\ref{sec:affine}, whereas a higher derivative $\Lambda\, \partial^p_\eta g$
only joins branches that are already there, without raising the order.\footnote{ In the notation of Section~\ref{sec:vonmises}, let
$\delta(s)$ be the displacement of the root of~\eqref{eq:nu-of-P} away from $\eta$ along the segment $P(s) = (1-s)\,\widehat
P(\eta,\lambda) + s\, P_0$. Writing all derivatives at $(\eta,\lambda)$, it solves $\delta = - s\, \overline G + s\, [
I_{d_\eta} - \partial_\eta \overline G ]\, \delta - \sum_{p \geq 2} \frac{1}{p!}\, \partial^p_\eta \overline G [
\delta^{\otimes p} ]$. Each leaf and each node with one child brings one factor $s$ and the higher derivatives bring none,
so a tree first appears at order $s^{d(\tau)}$ and the truncation at order $q$ in~\eqref{eq:PsiP} keeps those with
$d(\tau) \leq q$.}
For $q=0$, $\mathcal{T}_0$ only contains the single-node tree. For $q=1$, $\mathcal{T}_1$ contains the single-node tree and the tree with one root and one leaf child. For $q=2$, the five elements of $\mathcal{T}_2$ are shown in Figure~\ref{fig:trees-q2}. For $q = 3$ we have $|\mathcal{T}_3| = 13$, shown in Figure~\ref{fig:trees-q3}. The cardinalities grow rapidly, with $|\mathcal{T}_4| = 40$ and $|\mathcal{T}_5| = 130$.

\begin{figure}[tb]
\centering
\renewcommand{\arraystretch}{1.4}
\begin{tabular}{ccccccc}
\begin{tikzpicture}[baseline=(root.base),
    every node/.style={draw, circle, inner sep=2pt, minimum size=5pt}]
\node[fill=blue!15] (root) {};
\end{tikzpicture}
&
\begin{tikzpicture}[baseline=(root.base),
    level distance=0.55cm, sibling distance=0.9cm,
    every node/.style={draw, circle, inner sep=2pt, minimum size=5pt},
    edge from parent/.style={draw, -}]
\node[fill=blue!15] (root) {} child { node[fill=gray!30] {} };
\end{tikzpicture}
&
\begin{tikzpicture}[baseline=(root.base),
    level distance=0.55cm, sibling distance=0.9cm,
    every node/.style={draw, circle, inner sep=2pt, minimum size=5pt},
    edge from parent/.style={draw, -}]
\node[fill=blue!15] (root) {}
  child { node {}
    child { node[fill=gray!30] {} }
  };
\end{tikzpicture}
&
\begin{tikzpicture}[baseline=(root.base),
    level distance=0.55cm, sibling distance=0.9cm,
    every node/.style={draw, circle, inner sep=2pt, minimum size=5pt},
    edge from parent/.style={draw, -}]
\node[fill=blue!15] (root) {}
  child { node {}
    child { node {}
      child { node[fill=gray!30] {} }
    }
  };
\end{tikzpicture}
&
\begin{tikzpicture}[baseline=(root.base),
    level distance=0.55cm, sibling distance=0.6cm,
    every node/.style={draw, circle, inner sep=2pt, minimum size=5pt},
    edge from parent/.style={draw, -}]
\node[fill=blue!15] (root) {}
  child { node[fill=gray!30] {} }
  child { node[fill=gray!30] {} };
\end{tikzpicture}
&
\begin{tikzpicture}[baseline=(root.base),
    level distance=0.55cm, sibling distance=0.6cm,
    every node/.style={draw, circle, inner sep=2pt, minimum size=5pt},
    edge from parent/.style={draw, -}]
\node[fill=blue!15] (root) {}
  child { node {}
    child { node[fill=gray!30] {} }
  }
  child { node[fill=gray!30] {} };
\end{tikzpicture}
&
\begin{tikzpicture}[baseline=(root.base),
    level distance=0.55cm, sibling distance=0.5cm,
    every node/.style={draw, circle, inner sep=2pt, minimum size=5pt},
    edge from parent/.style={draw, -}]
\node[fill=blue!15] (root) {}
  child { node[fill=gray!30] {} }
  child { node[fill=gray!30] {} }
  child { node[fill=gray!30] {} };
\end{tikzpicture}
\\[1.5em]
\begin{tikzpicture}[baseline=(root.base),
    level distance=0.55cm, sibling distance=0.6cm,
    every node/.style={draw, circle, inner sep=2pt, minimum size=5pt},
    edge from parent/.style={draw, -}]
\node[fill=blue!15] (root) {}
  child { node {}
    child { node[fill=gray!30] {} }
    child { node[fill=gray!30] {} }
  };
\end{tikzpicture}
&
\begin{tikzpicture}[baseline=(root.base),
    level distance=0.55cm, sibling distance=0.6cm,
    every node/.style={draw, circle, inner sep=2pt, minimum size=5pt},
    edge from parent/.style={draw, -}]
\node[fill=blue!15] (root) {}
  child { node {}
    child { node {}
      child { node[fill=gray!30] {} }
      child { node[fill=gray!30] {} }
    }
  };
\end{tikzpicture}
&
\begin{tikzpicture}[baseline=(root.base),
    level distance=0.55cm, sibling distance=0.6cm,
    every node/.style={draw, circle, inner sep=2pt, minimum size=5pt},
    edge from parent/.style={draw, -}]
\node[fill=blue!15] (root) {}
  child { node {}
    child { node[fill=gray!30] {} }
    child { node {}
      child { node[fill=gray!30] {} }
    }
  };
\end{tikzpicture}
&
\begin{tikzpicture}[baseline=(root.base),
    level distance=0.55cm, sibling distance=0.5cm,
    every node/.style={draw, circle, inner sep=2pt, minimum size=5pt},
    edge from parent/.style={draw, -}]
\node[fill=blue!15] (root) {}
  child { node {}
    child { node[fill=gray!30] {} }
    child { node[fill=gray!30] {} }
    child { node[fill=gray!30] {} }
  };
\end{tikzpicture}
&
\begin{tikzpicture}[baseline=(root.base),
    level distance=0.55cm, sibling distance=0.6cm,
    every node/.style={draw, circle, inner sep=2pt, minimum size=5pt},
    edge from parent/.style={draw, -}]
\node[fill=blue!15] (root) {}
  child { node {}
    child { node {}
      child { node[fill=gray!30] {} }
      child { node[fill=gray!30] {} }
    }
    child { node[fill=gray!30] {} }
  };
\end{tikzpicture}
&
\begin{tikzpicture}[baseline=(root.base),
    level distance=0.55cm, sibling distance=0.6cm,
    every node/.style={draw, circle, inner sep=2pt, minimum size=5pt},
    edge from parent/.style={draw, -}]
\node[fill=blue!15] (root) {}
  child { node[fill=gray!30] {} }
  child { node {}
    child { node[fill=gray!30] {} }
    child { node[fill=gray!30] {} }
  };
\end{tikzpicture}
&
\\
\end{tabular}
\caption{The thirteen elements of $\mathcal{T}_3$. The first row shows the seven trees in which every non-root node has at most one child --- these correspond to the terms of the affine moment function $\psi^{(3)}_{\rm aff}$. The second row shows the six correction trees, each containing at least one non-root node with two or more children.}
\label{fig:trees-q3}
\end{figure}
Each rooted tree $\tau \in \mathcal{T}_q$ translates into a term $\kappa_\tau(W, \theta, \eta, \lambda)$ by populating the nodes with the algebraic objects already used in $\psi^{(2)}$. The root carries $m$ or one of its $\eta$-derivatives, with the order of differentiation equal to the number of children of the root. Each leaf carries $\Lambda\, g$. Each non-root, non-leaf node carries $\Lambda\, \partial_\eta^p g$, where $p$ is the number of children of the node. Each node is evaluated on its own independent copy of $W_i$, and the resulting expression is the contraction along the parent--child edges of the tree, as in Figure~\ref{fig:trees-q2}.

At $q = 3$, the seven trees in the top row of Figure~\ref{fig:trees-q3} give the affine moment function $\psi^{(3)}_{\rm aff}$ in~\eqref{eq:psi3-aff}, and the six trees in the bottom row give the correction terms in~\eqref{eq:psi3-nonlinear} that are needed when $g$ is nonlinear in $\eta$. See Appendix~\ref{app:formulas-q3} for both formulas. The affine pre-factors can be obtained from~\eqref{eq:score} by collecting the ordered chain-length tuples $(k_1, \ldots, k_r)$ that correspond to the same unordered tree. The correction term pre-factors are the values of the closed-form $c_{q,\tau}$ that we introduce next.
Note that the substantive content of the tree formalism is our closed-form expression for $c_{q,\tau}$, while the index set $\mathcal{T}_q$ and the  construction rule for $\kappa_\tau$ are purely notational bookkeeping tools.

\subsubsection{Closed form for $c_{q,\tau}$ and the main result\label{sec:nonlinear-cq}}

To state the closed form for $c_{q,\tau}$, we first need to define two further integers associated with each rooted tree $\tau$. We have already introduced $d(\tau)$ as the number of non-root nodes of $\tau$ with at most one child. Let $|\tau|$ denote the total number of non-root nodes of $\tau$ (equivalently, the number of edges of $\tau$). Let $\mathrm{Aut}(\tau)$ denote the automorphism group of $\tau$, that is, the set of bijections of the nodes of $\tau$ that fix the root and keep the tree structure unchanged. Finally, let $\sigma(\tau) := |\mathrm{Aut}(\tau)|$ denote its cardinality, called the symmetry coefficient of $\tau$. Then,
\begin{equation}
c_{q,\tau} \;:=\; \frac{(-1)^{|\tau|}}{\sigma(\tau)}\, \binom{q + |\tau| - d(\tau)}{|\tau|}.
\label{eq:cq-formula}
\end{equation}

In~\eqref{eq:cq-formula}, the sign $(-1)^{|\tau|}$ comes from solving~\eqref{eq:root} for $\eta(P(s))$ once per edge of $\tau$, the factor $1/\sigma(\tau)$ from grouping orderings of children by the automorphism group of $\tau$, and the binomial from counting how the $q$ derivatives are distributed across the edges of $\tau$, of which $d(\tau)$ are leaf-or-chain edges and $|\tau| - d(\tau)$ are branch edges. The one-child nodes come from expanding the inverse of the linearization of~\eqref{eq:root} along the segment in powers of $I_{d_\eta} - \Lambda(\lambda)\, \partial_\eta g$.

Figure~\ref{fig:invariants-q3} shows the values of $|\tau|$, $d(\tau)$, and $\sigma(\tau)$ for all rooted trees $\tau$ with $d(\tau) \leq 3$. The values of $|\tau|$ and $d(\tau)$ are found by counting the relevant nodes. The values of $\sigma(\tau)$ in this figure are also straightforward: for many of the trees, only the trivial bijection preserves the tree structure, so $\sigma(\tau) = 1$. For each of the remaining trees, exactly one node has $c$ isomorphic leaf children, which can be permuted in $c!$ ways, giving $\sigma(\tau) = c!$. For larger values of $d(\tau)$ the automorphism group can of course be more complicated.

\begin{figure}[p!]
\centering
\renewcommand{\arraystretch}{1.5}
\begin{tabular}{c|c|c|c}
\textbf{Rooted tree $\tau$} & $|\tau|$ & $d(\tau)$ & $\sigma(\tau)$ \\
\hline
\begin{tikzpicture}[baseline=(root.base),
    every node/.style={draw, circle, inner sep=2pt, minimum size=5pt}]
\node[fill=blue!15] (root) {};
\path (current bounding box.south) ++(0,-0.4em) -- ++(0.001em,0);
\end{tikzpicture}
& $0$ & $0$ & $1$ \\[0.4em]
\hline
\begin{tikzpicture}[baseline=(root.base),
    level distance=0.5cm, sibling distance=0.6cm,
    every node/.style={draw, circle, inner sep=2pt, minimum size=5pt},
    edge from parent/.style={draw, -}]
\node[fill=blue!15] (root) {} child { node[fill=gray!30] {} };
\path (current bounding box.south) ++(0,-0.4em) -- ++(0.001em,0);
\end{tikzpicture}
& $1$ & $1$ & $1$ \\[0.4em]
\hline
\begin{tikzpicture}[baseline=(root.base),
    level distance=0.5cm, sibling distance=0.6cm,
    every node/.style={draw, circle, inner sep=2pt, minimum size=5pt},
    edge from parent/.style={draw, -}]
\node[fill=blue!15] (root) {} child { node {} child { node[fill=gray!30] {} } };
\path (current bounding box.south) ++(0,-0.4em) -- ++(0.001em,0);
\end{tikzpicture}
& $2$ & $2$ & $1$ \\[0.4em]
\hline
\begin{tikzpicture}[baseline=(root.base),
    level distance=0.5cm, sibling distance=0.45cm,
    every node/.style={draw, circle, inner sep=2pt, minimum size=5pt},
    edge from parent/.style={draw, -}]
\node[fill=blue!15] (root) {}
  child { node[fill=gray!30] {} }
  child { node[fill=gray!30] {} };
\path (current bounding box.south) ++(0,-0.4em) -- ++(0.001em,0);
\end{tikzpicture}
& $2$ & $2$ & $2$ \\[0.4em]
\hline
\begin{tikzpicture}[baseline=(root.base),
    level distance=0.45cm, sibling distance=0.45cm,
    every node/.style={draw, circle, inner sep=2pt, minimum size=5pt},
    edge from parent/.style={draw, -}]
\node[fill=blue!15] (root) {} child { node {} child { node[fill=gray!30] {} } child { node[fill=gray!30] {} } };
\path (current bounding box.south) ++(0,-0.4em) -- ++(0.001em,0);
\end{tikzpicture}
& $3$ & $2$ & $2$ \\[0.4em]
\hline
\begin{tikzpicture}[baseline=(root.base),
    level distance=0.45cm, sibling distance=0.6cm,
    every node/.style={draw, circle, inner sep=2pt, minimum size=5pt},
    edge from parent/.style={draw, -}]
\node[fill=blue!15] (root) {} child { node {} child { node {} child { node[fill=gray!30] {} } } };
\path (current bounding box.south) ++(0,-0.4em) -- ++(0.001em,0);
\end{tikzpicture}
& $3$ & $3$ & $1$ \\[0.4em]
\hline
\begin{tikzpicture}[baseline=(root.base),
    level distance=0.45cm, sibling distance=0.5cm,
    every node/.style={draw, circle, inner sep=2pt, minimum size=5pt},
    edge from parent/.style={draw, -}]
\node[fill=blue!15] (root) {}
  child { node {} child { node[fill=gray!30] {} } }
  child { node[fill=gray!30] {} };
\path (current bounding box.south) ++(0,-0.4em) -- ++(0.001em,0);
\end{tikzpicture}
& $3$ & $3$ & $1$ \\[0.4em]
\hline
\begin{tikzpicture}[baseline=(root.base),
    level distance=0.45cm, sibling distance=0.4cm,
    every node/.style={draw, circle, inner sep=2pt, minimum size=5pt},
    edge from parent/.style={draw, -}]
\node[fill=blue!15] (root) {}
  child { node[fill=gray!30] {} }
  child { node[fill=gray!30] {} }
  child { node[fill=gray!30] {} };
\path (current bounding box.south) ++(0,-0.4em) -- ++(0.001em,0);
\end{tikzpicture}
& $3$ & $3$ & $6$ \\[0.4em]
\hline
\begin{tikzpicture}[baseline=(root.base),
    level distance=0.4cm, sibling distance=0.45cm,
    every node/.style={draw, circle, inner sep=2pt, minimum size=5pt},
    edge from parent/.style={draw, -}]
\node[fill=blue!15] (root) {} child { node {} child { node {} child { node[fill=gray!30] {} } child { node[fill=gray!30] {} } } };
\path (current bounding box.south) ++(0,-0.4em) -- ++(0.001em,0);
\end{tikzpicture}
& $4$ & $3$ & $2$ \\[0.4em]
\hline
\begin{tikzpicture}[baseline=(root.base),
    level distance=0.4cm, sibling distance=0.45cm,
    every node/.style={draw, circle, inner sep=2pt, minimum size=5pt},
    edge from parent/.style={draw, -}]
\node[fill=blue!15] (root) {} child { node {}
  child { node[fill=gray!30] {} }
  child { node {} child { node[fill=gray!30] {} } } };
\path (current bounding box.south) ++(0,-0.4em) -- ++(0.001em,0);
\end{tikzpicture}
& $4$ & $3$ & $1$ \\[0.4em]
\hline
\begin{tikzpicture}[baseline=(root.base),
    level distance=0.4cm, sibling distance=0.4cm,
    every node/.style={draw, circle, inner sep=2pt, minimum size=5pt},
    edge from parent/.style={draw, -}]
\node[fill=blue!15] (root) {} child { node {}
  child { node[fill=gray!30] {} }
  child { node[fill=gray!30] {} }
  child { node[fill=gray!30] {} } };
\path (current bounding box.south) ++(0,-0.4em) -- ++(0.001em,0);
\end{tikzpicture}
& $4$ & $3$ & $6$ \\[0.4em]
\hline
\begin{tikzpicture}[baseline=(root.base),
    level distance=0.4cm, sibling distance=0.5cm,
    every node/.style={draw, circle, inner sep=2pt, minimum size=5pt},
    edge from parent/.style={draw, -}]
\node[fill=blue!15] (root) {}
  child { node[fill=gray!30] {} }
  child { node {} child { node[fill=gray!30] {} } child { node[fill=gray!30] {} } };
\path (current bounding box.south) ++(0,-0.4em) -- ++(0.001em,0);
\end{tikzpicture}
& $4$ & $3$ & $2$ \\[0.4em]
\hline
\begin{tikzpicture}[baseline=(root.base),
    level distance=0.35cm, sibling distance=0.4cm,
    every node/.style={draw, circle, inner sep=2pt, minimum size=5pt},
    edge from parent/.style={draw, -}]
\node[fill=blue!15] (root) {} child { node {}
  child { node[fill=gray!30] {} }
  child { node {} child { node[fill=gray!30] {} } child { node[fill=gray!30] {} } } };
\path (current bounding box.south) ++(0,-0.4em) -- ++(0.001em,0);
\end{tikzpicture}
& $5$ & $3$ & $2$ \\[0.4em]
\hline
\end{tabular}
\caption{The integers $|\tau|$, $d(\tau)$, and $\sigma(\tau)$ for the 13 trees with $d(\tau) \leq 3$.}
\label{fig:invariants-q3}
\end{figure}
As a check, applying~\eqref{eq:cq-formula} at $q = 3$ to the values in Figure~\ref{fig:invariants-q3} reproduces the 13 coefficients in~\eqref{eq:psi3-aff} and~\eqref{eq:psi3-nonlinear}. For example, the root with three leaf children gives $c_{3,\tau} = -\binom{3}{3}/3! = -1/6$ (matching the coefficient of $\partial_\eta^3 m[\Lambda g, \Lambda g, \Lambda g]$ in~\eqref{eq:psi3-aff}), etc.\footnote{%
It is also instructive to verify the consistency between~\eqref{eq:cq-formula} and the affine moment function in~\eqref{eq:score}. The trees that contribute to $\psi^{(q)}_{\rm aff}$ are those in which only the root has more than one child, equivalently those with $|\tau| = d(\tau)$. For such a tree, let $r \geq 0$ be the number of children of the root, and let $k_b \geq 1$ be the chain length of each child for $b \in \{1, \ldots, r\}$. We then have $|\tau| = d(\tau) = k_1 + \cdots + k_r$. The automorphism group permutes branches of equal length, so if the multiset $\{k_1, \ldots, k_r\}$ has $n_\ell$ elements equal to $\ell$, then $\sigma(\tau) = \prod_\ell n_\ell!$, and~\eqref{eq:cq-formula} simplifies to
\begin{equation}
c_{q,\tau} \;=\; \frac{(-1)^{|\tau|}}{\prod_\ell n_\ell!}\, \binom{q}{|\tau|}.
\label{eq:cq-formula-affine}
\end{equation}
The same unordered tree $\tau$ is the image of $r! / \prod_\ell n_\ell!$ ordered branch tuples $(k_1, \ldots, k_r)$, each contributing the term $\frac{(-1)^{|\tau|}}{r!}\, \binom{q}{|\tau|}\, \kappa_\tau$ to the affine moment function $\psi^{(q)}_{\rm aff}$ in~\eqref{eq:score}. Summing across these tuples cancels the $r!$ and reproduces $c_{q,\tau}\, \kappa_\tau$, with $c_{q,\tau}$ as in~\eqref{eq:cq-formula-affine}, confirming that $\psi^{(q)}_{\rm aff}$ in~\eqref{eq:score} matches the affine part of~\eqref{eq:general-psi}.
}

Tree-indexed expansions of this form, including the appearance of the symmetry factor $1/\sigma(\tau)$, are a well-known structure in the literature on multivariate Taylor expansions, where $\sigma(\tau)$ is the symmetry coefficient of $\tau$, see Appendix~\ref{app:primer} for a discussion of the Butcher $B$-series. That appendix also gives a simple recursion for computing $\sigma(\tau)$.\footnote{
Using this recursion we have numerically verified that \eqref{eq:cq-formula} gives orthogonal moment functions for $q$ up to $10$ in the scalar case $d_\eta = d_g = 1$, and for $q$ up to $8$ in the multivariate case. Enumerating all $|\mathcal{T}_{10}| = 110{,}135$ trees and calculating all corresponding $c_{q,\tau}$ is possible within seconds.}

With $c_{q,\tau}$ in hand, we can state our main result for the nonlinear case.

\begin{assumption}
\label{ass:nonlinear}
\leavevmode
\begin{enumerate}[\rm (i)]
\item Assumption~\ref{ass:affine}(ii)--(iv) hold.
\item $g(W_i, \theta, \eta)$ is $(q+1)$-times continuously differentiable in $\eta$ in a neighborhood of $\eta_0$.
\item  For all $p \leq q+1$, there is a neighborhood $\mathcal{N}$ of $\eta_0$ and an integrable function $M^g_p(W_i)$ with $\sup_{\eta \in \mathcal{N}} \|\partial_\eta^p g(W_i, \theta_0, \eta)\| \leq M^g_p(W_i)$ almost surely.
\end{enumerate}
\end{assumption}

\begin{theorem}[$q$-th order Neyman orthogonality, nonlinear case]
\label{th:nonlinear}
Let $q \geq 1$ be an integer, and let $\Psi^{(q)}(\theta, \eta, \lambda) := \mathbb{E}[\psi^{(q)}(W_1, \ldots, W_L; \theta, \eta, \lambda)]$ be the population moment associated with $\psi^{(q)}$ in~\eqref{eq:general-psi}, with $c_{q,\tau}$ given by~\eqref{eq:cq-formula}, where $L = 1 + \max_{\tau \in \mathcal{T}_q} |\tau| \leq 2q$, or $L = \max_{\tau \in \mathcal{T}_q} |\tau|$ when $m$ does not depend on the data. Under Assumptions~\ref{ass:iid}, \ref{ass:jacobian-inverse} and~\ref{ass:nonlinear}, the moment function $\psi^{(q)}$ is $q$-th order Neyman-orthogonal in the sense of Definition~\ref{def:orthogonality}.
\end{theorem}

\begin{corollary}
\label{cor:derivatives-nonlinear}
Suppose in addition to the assumptions of Theorem~\ref{th:nonlinear} that the conditions on $g$ in
parts (ii) and (iii), and the conditions on $m$ in part (i), of Assumption~\ref{ass:nonlinear} hold
with $q+1$ replaced by $2q$. Then
$(\eta,\lambda) \mapsto \Psi^{(q)}(\theta_0,\eta,\lambda)$ is $q$ times differentiable at
$(\eta_0,\lambda_0)$, and~\eqref{eq:ortho} holds.
\end{corollary}

The expression~\eqref{eq:cq-formula} applies to both the affine and correction parts of $\psi^{(q)}$. Trees with $|\tau| = d(\tau)$ are those in which every non-root node is either a leaf or a one-child node; for these, \eqref{eq:cq-formula} reduces to (\ref{eq:cq-formula-affine}), the binomial weights of the affine score~\eqref{eq:score}. Trees with $|\tau| > d(\tau)$ --- those with at least one non-root node with at least two children --- are the correction trees, which contribute terms involving the higher-order derivatives $\partial_\eta^p g$ for $p \geq 2$. When $g$ is affine in $\eta$, all such terms vanish, \eqref{eq:general-psi} reduces to $\psi^{(q)} = \psi^{(q)}_{\rm aff}$, and Theorem~\ref{th:nonlinear} recovers Theorem~\ref{th:main}. Assumption~\ref{ass:affine} also implies Assumption~\ref{ass:nonlinear} in that case, because $\partial_\eta g(W_i, \theta, \eta)$ does not depend on $\eta$ and $\partial_\eta^p g(W_i, \theta, \eta) = 0$ for $p \geq 2$. Theorem~\ref{th:main} is therefore a corollary of Theorem~\ref{th:nonlinear}; see Appendix~\ref{app:proof-thm1-short} for a shorter direct proof of it. At order $2q$ the same argument makes Corollary~\ref{cor:derivatives-affine} the affine case of Corollary~\ref{cor:derivatives-nonlinear}.

The total number of independent copies of $W_i$ used by $\psi^{(q)}$ equals one plus the maximum of $|\tau|$ over $\tau \in \mathcal{T}_q$, which is bounded by $L \leq 2q$ for $q \geq 2$ (compared with $L = 1 + q$ in the affine case).
The bound is attained by every tree with $d(\tau) = q$ in which the root has a single child and every non-root, non-leaf node has two children.

\section{Reducing the nuisance dimension}\label{sec:nuisance}

To reduce the dimension of $\lambda$, the key observation is that the moment function $g(W_i, \theta, \eta)$ in \eqref{eq:model-general} can be replaced by a different function $\widetilde g(W, \theta, \eta)$ that may depend on multiple independent copies of $W_i$.\footnote{As before, here we write $W$ without subscript when the argument requires several copies.} The replacement is valid as long as $\mathbb{E}[\widetilde g(W, \theta_0, \eta_0)] = 0$ and the corresponding Jacobian $\widetilde J_0 = \mathbb{E}[\partial_\eta \widetilde g(W, \theta_0, \eta_0)]$ has a left inverse $\Lambda_0$ with $\Lambda_0 \widetilde J_0 = I_{d_\eta}$. The orthogonality construction of the previous sections applies with $g$ replaced by $\widetilde g$ throughout. The advantage is that $\widetilde g$ can be designed so that the nuisance parameter $\lambda$ has much lower dimension.

There is, however, a trade-off: reducing $d_\lambda$ comes at the cost of increasing the number of independent copies of the data, $L$, needed by the orthogonal moment function. The original moment function $g(W_i, \theta, \eta)$ depends on a single observation, so the $q$-th order affine orthogonal moment function uses $L = 1 + q$ copies. The transformed function $\widetilde g(W, \theta, \eta)$ may itself require multiple copies (e.g., $d_\eta$ copies for the determinant construction below), and each appearance of $\widetilde g$ in the orthogonal moment function then uses that many copies, increasing $L$ accordingly. In practice, the choice of $\widetilde g$ balances the difficulty of estimating a high-dimensional nuisance $\lambda$ (large $d_\lambda$, small $L$) against the  possibly larger variance of a higher-order U-statistic (small $d_\lambda$, large $L$).

In Sections~\ref{sec:affine} and~\ref{sec:nonlinear}, the orthogonal moment functions were stated in terms of a moment function $g$ that depends on a single observation $W_j$. When $\widetilde g$ requires a block of several independent copies, the construction is applied in the same way, with one change: each occurrence of $\partial^p_\eta g(W_j,\theta,\eta)$, $p \geq 0$, in the formulas of those sections is replaced by $\partial^p_\eta \widetilde g$ evaluated on a fresh block of independent copies of the size required by $\widetilde g$, and the assumptions of those sections are imposed on $\widetilde g$. We give examples of constructions for $\widetilde g$ below.

\paragraph{Determinant construction (exactly identified case).}
Suppose $d_g = d_\eta$ and $J_0$ is invertible. Define $\widetilde g(W_1, \ldots, W_{d_\eta}; \theta, \eta)$ componentwise by
\begin{multline*}
\widetilde g_r(W_1, \ldots, W_{d_\eta}; \theta, \eta)
= \det\Big[
\partial_{\eta_1} g(W_1, \theta, \eta), \ldots,
\partial_{\eta_{r-1}} g(W_{r-1}, \theta, \eta),\,
g(W_r, \theta, \eta),\\
\partial_{\eta_{r+1}} g(W_{r+1}, \theta, \eta),
\ldots,
\partial_{\eta_{d_\eta}} g(W_{d_\eta}, \theta, \eta)
\Big],
\end{multline*}
where $r = 1,\ldots,d_\eta$, and the $d_\eta \times d_\eta$ matrix inside the determinant has $g(W_r, \theta, \eta)$ in its $r$-th column and $\partial_{\eta_j} g(W_j, \theta, \eta)$ in its $j$-th column for $j \neq r$. Then $\mathbb{E}[\widetilde g_r] = 0$ follows from $\mathbb{E}[g(W_i, \theta_0, \eta_0)] = 0$ and the multilinearity of the determinant. The Jacobian satisfies $\mathbb{E}[\partial_\eta \widetilde g(W, \theta_0, \eta_0)] = \det(J_0) \cdot I_{d_\eta}$, so $\Lambda(\lambda) = \lambda \cdot I_{d_\eta}$ with
\[
\lambda_0 = \frac{1}{\det(J_0)}.
\]
This gives $d_\lambda = 1$. The trade-off is that each evaluation of $\widetilde g$ uses $d_\eta$ independent copies of $W_i$, so the total number of copies for the $q$-th order moment function becomes $L = 1 + q \cdot d_\eta$ when $g$ is affine in $\eta$ (then $\widetilde g$ is affine as well), and $L = 1 + (2q-1)\, d_\eta$ in general.

\paragraph{Determinant construction (overidentified case).}
When $d_g > d_\eta$, the Jacobian $J_0$ is no longer square, so the construction above does not directly apply. The idea is to first project the $d_g$-dimensional vectors $g$ and $\partial_{\eta_j} g$ down to $d_\eta$ dimensions using independent copies of $\partial_\eta g$, and then take the determinant. Concretely, define $\widetilde g(W_1, \ldots, W_{2d_\eta}; \theta, \eta)$ componentwise by
\begin{multline*}
\widetilde g_r(W_1, \ldots, W_{2d_\eta}; \theta, \eta)
= \det\Big[
(\partial_\eta g(W_1, \theta, \eta))' \partial_{\eta_1} g(W_2, \theta, \eta), \ldots,\\
(\partial_\eta g(W_{2r-1}, \theta, \eta))' g(W_{2r}, \theta, \eta), \ldots,
(\partial_\eta g(W_{2d_\eta - 1}, \theta, \eta))' \partial_{\eta_{d_\eta}} g(W_{2d_\eta}, \theta, \eta)
\Big],
\end{multline*}
where $r = 1, \ldots, d_\eta$. Each column of the $d_\eta \times d_\eta$ matrix is formed by pre-multiplying a $d_g$-vector ($\partial_{\eta_j} g$ or $g$) by the $d_\eta \times d_g$ matrix $(\partial_\eta g)'$ from an independent copy, yielding a $d_\eta$-vector. The $r$-th column uses $g(W_{2r}, \theta, \eta)$ (without a $\partial_{\eta_r}$ derivative), while all other columns use $\partial_{\eta_j} g(W_{2j}, \theta, \eta)$. Then $\mathbb{E}[\widetilde g_r] = 0$ by the same multilinearity argument as before, the Jacobian satisfies $\mathbb{E}[\partial_\eta \widetilde g(W, \theta_0, \eta_0)] = \det(J_0' J_0) \cdot I_{d_\eta}$, and
\[
\Lambda(\lambda) = \lambda \cdot I_{d_\eta}, \qquad \lambda_0 = \frac{1}{\det(J_0' J_0)},
\]
again giving $d_\lambda = 1$. Each evaluation of $\widetilde g$ now uses $2 d_\eta$ independent copies of $W_i$.

\section{Conclusion\label{sec:con}}

In this paper we have derived closed-form expressions for orthogonal moment functions to any chosen order, in models that are defined through a finite-dimensional set of moment conditions. The construction relies on polynomial approximations based on von Mises expansions, and on the use of multiple independent copies to estimate the components of the polynomials without bias. While the construction introduces an additional nuisance parameter, the latter does not change with the order of orthogonality and it can be reduced to a scalar.   

The paper belongs to a growing agenda on higher-order Neyman orthogonality. An important question is the choice of the order $q$ of orthogonality, which we have taken as given in the analysis. Also, while we have derived asymptotic properties of orthogonal estimators in a grouped data setting, more work is needed for other settings. Lastly, our analysis focuses on the case where the nuisance parameter is finite-dimensional, and extending our approach to cases where it contains nonparametric components is important too.

\section*{Acknowledgment}
This paper benefited from the use of generative AI tools to assist with language editing, \LaTeX{} formatting, and checking proofs. All output was carefully reviewed by the authors. All substantive content, results, and any remaining errors are the authors' responsibility.

\clearpage

%\bibliographystyle{chicago3}
%\bibliography{refs}

\clearpage

\appendix

\begin{center}
{\Large \textbf{APPENDIX}}
\end{center}

\section{Additional material\label{app:additional}}

\subsection{Details on estimation in the simulation exercise\label{app_sim}}

We now describe how the OLS and ORTH estimators in the simulations of Section~\ref{sec:implementation} are computed, with and without regularization of the estimator of $\Lambda_i$. We compute the OLS estimators
as
$$\widehat\eta_i=\left(\sum_{t=1}^TX_{it}X_{it}'\right)^{-1} \sum_{t=1}^TX_{it}Y_{it},$$
and we report the resulting estimates of $\theta_1$ and $\theta_2$ (denoted as OLS). Note that, given the design, $\sum_{t=1}^TX_{it}X_{it}'$ is singular with positive probability. We report averages of $\widehat\eta_{i1}$ and $\widehat\eta_{i1}^2$ on the subset of units for which singularity does not occur, and compare those to the averages of $\eta_{i10}$ and $\eta_{i10}^2$ for the same subset of units. We proceed similarly for the orthogonal estimators described below, for which we discard the units for which at least one of the matrices to be inverted, computed on subsets of observations, is singular.

Next, we compute orthogonal estimators to order 1 and 2 (denoted as ORTH) as follows. For order 1, let, for all $s$ in $\{1,\ldots,T\}$,
$$\widehat\eta_{i,-s}=\left(\sum_{t=1}^T\boldsymbol{1}\{t\neq s\}X_{it}X_{it}'\right)^{-1} \sum_{t=1}^T\boldsymbol{1}\{t\neq s\}X_{it}Y_{it},\quad \widehat\Lambda_{i,-s}=-\left(\frac{1}{T-1}\sum_{t=1}^T\boldsymbol{1}\{t\neq s\}X_{it}X_{it}'\right)^{-1}.$$
For $\theta_1$, the first-order orthogonal estimator is the second element of
$$\frac{1}{T}\sum_{s=1}^T \Big[\widehat{\eta}_{i,-s}-\widehat{\Lambda}_{i,-s}X_{i,s}(Y_{i,s}-X_{i,s}'\widehat{\eta}_{i,-s})\Big],$$
averaged across units. For $\theta_2$, it is
$$\frac{1}{T}\sum_{s=1}^T \Big[\widehat{\eta}_{i,-s}'D\widehat{\eta}_{i,-s}-2\widehat{\eta}_{i,-s}'D\widehat{\Lambda}_{i,-s}X_{i,s}(Y_{i,s}-X_{i,s}'\widehat{\eta}_{i,-s})\Big],$$
also averaged across units, where $D$ is the $3\times 3$ matrix with a one in the position $(2,2)$ and zeroes everywhere else.

For order 2, let, for all $s_1 \neq s_2$ in $\{1,\ldots,T\}$,
$$\widehat\eta_{i,-(s_1,s_2)}=\left(\sum_{t=1}^T\boldsymbol{1}\{t\neq s_1,s_2\}X_{it}X_{it}'\right)^{-1} \sum_{t=1}^T\boldsymbol{1}\{t\neq s_1,s_2\}X_{it}Y_{it},$$
and
\begin{equation}\widehat\Lambda_{i,-(s_1,s_2)}=-\left(\frac{1}{T-2}\sum_{t=1}^T\boldsymbol{1}\{t\neq s_1,s_2\}X_{it}X_{it}'\right)^{-1}.\label{eq_Lambda_hat}
\end{equation}
For $\theta_1$, the second-order orthogonal estimator is the second element (i.e., the one corresponding to $\eta_{i1}$) of
\begin{align*}\frac{2}{T(T-1)}\sum_{s_1=1}^T\sum_{s_2=s_1+1}^T \bigg[ &\widehat{\eta}_{i,-(s_1,s_2)}\\&-	\left(2I_3+\widehat{\Lambda}_{i,-(s_1,s_2)}X_{i,s_2}X_{i,s_2}'\right)\widehat{\Lambda}_{i,-(s_1,s_2)}X_{i,s_1}(Y_{i,s_1}-X_{i,s_1}'\widehat{\eta}_{i,-(s_1,s_2)}) \bigg],
\end{align*}
averaged across units, where $I_3$ is the $3\times 3$ identity matrix. For $\theta_2$, the second-order orthogonal estimator is
\begin{align*}\frac{2}{T(T-1)} & \sum_{s_1=1}^T\sum_{s_2=s_1+1}^T \bigg[ \widehat{\eta}_{i,-(s_1,s_2)}'D\widehat{\eta}_{i,-(s_1,s_2)}\\&-2	\widehat{\eta}_{i,-(s_1,s_2)}'D\left(2I_3+\widehat{\Lambda}_{i,-(s_1,s_2)}X_{i,s_2}X_{i,s_2}'\right)\widehat{\Lambda}_{i,-(s_1,s_2)}X_{i,s_1}(Y_{i,s_1}-X_{i,s_1}'\widehat{\eta}_{i,-(s_1,s_2)})\\
	&+\left(\widehat{\Lambda}_{i,-(s_1,s_2)}X_{i,s_2}(Y_{i,s_2}-X_{i,s_2}'\widehat{\eta}_{i,-(s_1,s_2)})\right)'D\left(\widehat{\Lambda}_{i,-(s_1,s_2)}X_{i,s_1}(Y_{i,s_1}-X_{i,s_1}'\widehat{\eta}_{i,-(s_1,s_2)})\right) \bigg],
\end{align*}
also averaged across units.

In addition to estimators based on the plug-in estimates (\ref{eq_Lambda_hat}) of $\Lambda_i$ (and their leave-one-out counterparts), we also report estimators based on empirical Bayes regularization. Let $p_{ik\ell}=\Pr(X_{it1}=k,X_{it2}=\ell)$ for all $(k,\ell)\in\{0,1\}^2$. To regularize $\Lambda_i$, we endow $(p_{i00},p_{i01},p_{i10},p_{i11})$ with a Dirichlet prior with parameters $(\alpha_{00},\alpha_{01},\alpha_{10},\alpha_{11})$. Let $\alpha=\alpha_{00}+\alpha_{01}+\alpha_{10}+\alpha_{11}$, $\pi_{k\ell}=\alpha_{k\ell}/\alpha$, and
	$$\Pi=\left(\begin{array}{ccc}1&\pi_{10}+\pi_{11} & \pi_{01}+\pi_{11}\\\pi_{10}+\pi_{11}&\pi_{10}+\pi_{11}&\pi_{11} \\\pi_{01}+\pi_{11}&\pi_{11}&\pi_{01}+\pi_{11}\end{array}\right).$$
	We estimate $-\Lambda_i^{-1}=\mathbb{E}[X_{it}X_{it}']$ as the posterior mean
	$$
    -\widehat \Lambda_i^{-1}=\frac{T}{T+\alpha}\left(\frac{1}{T}\sum_{t=1}^TX_{it}X_{it}'\right)+\frac{\alpha}{T+\alpha}\Pi.
    $$
	Finally, to select $(\alpha_{00},\alpha_{01},\alpha_{10},\alpha_{11})$, we maximize the marginal likelihood across firms with respect to $\alpha$, while fixing $\pi_{k\ell}$ at the pooled mean $\widehat{\pi}_{k\ell}=\frac{1}{NT}\sum_i\sum_t \boldsymbol{1}\{X_{it1}=k,X_{it2}=\ell\}$, so that $\alpha_{k\ell}=\widehat\pi_{k\ell}\alpha$. We estimate $\alpha$ as
	$$\widehat{\alpha}=\underset{\alpha>0}{\mbox{argmax}}\, N[\ln\Gamma (\alpha)-\ln\Gamma (\alpha+T)]+\sum_{i=1}^N\sum_{(k,\ell)\in\{0,1\}^2}\left[\ln\Gamma (\widehat\pi_{k\ell}\alpha+n_{ik\ell})-\ln\Gamma (\widehat\pi_{k\ell}\alpha)\right],$$
	where $n_{ik\ell}=\sum_{t=1}^T\boldsymbol{1}\{X_{it1}=k,X_{it2}=\ell\}$. The estimates $\widehat\alpha$ and $\widehat\pi_{k\ell}$ are computed once, from all $NT$ observations. We apply the regularization scheme to both the OLS and ORTH estimates, replacing $\widehat\Lambda_i$ by its regularized version and $\widehat\eta_i$ by $-\widehat\Lambda_i\frac{1}{T}\sum_{t=1}^TX_{it}Y_{it}$. In the case of ORTH, we apply the above formulas to subsets of observations, leaving out one $t$ index (order 1) or two $t$ indices (order 2), with $T$ replaced by $T-1$ or $T-2$, and averaging across indices \emph{ex post}. With regularization, no matrix is singular and all units are used.\footnote{We also computed versions of the estimators in which (i) the hyperparameters $\widehat\alpha$ and $\widehat\pi_{k\ell}$ are cross-fitted across firms, so that the observations of each firm are excluded from the hyperparameters used for that firm, (ii) the order-2 summands are averaged over both orderings of each pair $(s_1,s_2)$, as in~\eqref{UstatsPsi}, and (iii) in the un-regularized case, only the leave-out terms with a singular matrix are discarded instead of the whole firm. The results are similar to the ones we report in the text. In addition, since we redraw firm coefficients in each simulation, the simulation bands in Figures~\ref{fig:theta1} and~\ref{fig:theta2} also reflect variation in the finite-$N$ averages of $\eta_{i10}$ and $\eta_{i10}^2$ around their population values. Recentering each estimate on the corresponding in-sample average gives very similar figures.}

\subsection{The distribution $\widehat{P}(\eta,\lambda)$\label{app:projection}}

In Section~\ref{sec:vm-expansion} we defined $\widehat P(\eta,\lambda)$ by~\eqref{eq:Phat} and used two of its properties, the bound~\eqref{eq:lipschitz} and the identity~\eqref{eq:Gplug}.  From here on we write $\nu = (\eta, \lambda)$ again. The following lemma makes the choice of $\widehat P(\nu)$ in~\eqref{eq:Phat} precise. For this lemma we take $\mathcal{P}$ to be the set of all finite signed measures of total mass one whose density relative to $\overline P$ is square integrable under $\overline P$. This set is convex. As in Section~\ref{sec:perturb}, we assume that $P_0$ has such a density, so that $P_0 \in \mathcal{P}$. We require $\widehat P(\nu)$  to satisfy the constraints~\eqref{eq:Phat-moments} and~\eqref{eq:derivatives}, which are linear in the density. The criterion is an $L^2(\overline P)$ norm, so the member of $\mathcal{P}$ closest to $P_0$ among those satisfying them is an orthogonal projection and has a closed form. Let
\begin{equation}
h_q(w; \nu) \;:=\; \Bigl( g(w, \theta_0, \eta)',\; \mathrm{vec}[\Lambda(\lambda)\, \partial_\eta g(w, \theta_0, \eta)]',\; \mathrm{vec}[\partial^2_\eta g(w, \theta_0, \eta)]',\; \ldots,\; \mathrm{vec}[\partial^q_\eta g(w, \theta_0, \eta)]' \Bigr)',
\label{eq:hq}
\end{equation}
and write $$\overline h_q(\nu) := \mathbb{E}_{\overline P}[ h_q(W_i; \nu) ],\quad \overline\Sigma_q(\nu) := \mathbb{E}_{\overline P}[ (h_q(W_i; \nu) - \overline h_q(\nu)) (h_q(W_i; \nu) - \overline h_q(\nu))' ].$$ Let $\overline g(\eta) := \mathbb{E}[g(W_i, \theta_0, \eta)]$, and let
\begin{equation}
\zeta_q(\nu) \;:=\; \Bigl( -\overline g(\eta)',\; \mathrm{vec}[ I_{d_\eta} - \partial_\eta \overline G(\eta, \lambda) ]',\; 0',\; \ldots,\; 0' \Bigr)'.
\label{eq:zetaq}
\end{equation}
Only the first two blocks of $\zeta_q(\nu)$ are nonzero, because $P_0$ satisfies~\eqref{eq:derivatives} at every $\nu$, and $\zeta_q(\nu_0) = 0$ because $\overline g(\eta_0) = 0$ and $\partial_\eta \overline G(\eta_0, \lambda_0) = \Lambda_0 J_0 = I_{d_\eta}$. The first block does not depend on $\lambda$, and the second is affine in $\Lambda(\lambda)$. Write $\gamma_q(\nu)$ for the smallest nonzero eigenvalue of $\overline\Sigma_q(\nu)$ and $\overline\Sigma_q(\nu)^{+}$ for its Moore--Penrose inverse. Finally, let
\[
\overline R_\nu(\widetilde\eta) := \overline G(\widetilde\eta, \lambda) - \overline G(\eta, \lambda) - \partial_\eta \overline G(\eta, \lambda)(\widetilde\eta - \eta)
\]
denote the nonlinear part of $\overline G(\cdot\,, \lambda)$ around $\eta$, as it appears in~\eqref{eq:Gplug}.

\begin{lemma}[The distribution $\widehat{P}(\eta,\lambda)$]
\label{lem:projection}
Let $q \geq 1$, let Assumption~\ref{ass:nonlinear} hold with envelopes $M^g_p$, $p \leq q+1$, that dominate $\sup_{\eta \in \mathcal{N}} \| \partial^p_\eta g(\cdot\,, \theta_0, \eta) \|$ $\overline P$-almost surely and are square integrable under $\overline P$, and suppose that for every $\nu$ in a neighborhood of $\nu_0$ every vector $c$ for which $c' h_q(\cdot\,; \nu)$ is $\overline P$-almost surely constant has zero entries in its first two blocks (those matching $g(\cdot\,, \theta_0, \eta)$ and $\Lambda(\lambda)\, \partial_\eta g(\cdot\,, \theta_0, \eta)$). Then there is a neighborhood of $\nu_0$ on which:
\begin{enumerate}[\rm (i)]
\item the solution to $\min_{P\in\mathcal{P}}\, \| P - P_0 \|_{\overline P}$ subject to~\eqref{eq:Phat-moments} and~\eqref{eq:derivatives} is unique, and it is given by
\begin{equation}
\frac{d\widehat P(\nu)}{d\overline P}(w)
\;=\; \frac{dP_0}{d\overline P}(w) + \zeta_q(\nu)'\, \overline\Sigma_q(\nu)^{+} \bigl[ h_q(w; \nu) - \overline h_q(\nu) \bigr].
\label{eq:projection}
\end{equation}
Moreover, $\| \widehat P(\nu) - P_0 \|_{\overline P} = ( \zeta_q(\nu)'\, \overline\Sigma_q(\nu)^{+}\, \zeta_q(\nu) )^{1/2}$ and $\widehat P(\nu_0) = P_0$;
\item if $\gamma_q(\nu) \geq \gamma > 0$ on the neighborhood, then $\| \widehat P(\nu) - P_0 \|_{\overline P} \leq \gamma^{-1/2} c_\zeta \| \nu - \nu_0 \|$ there, where $c_\zeta$ is a Lipschitz constant of $\nu \mapsto \zeta_q(\nu)$;
\item under the condition of (ii), the function $\widetilde\eta \mapsto \overline G_{\widehat P(\nu)}(\widetilde\eta, \lambda)$ has value $0$, Jacobian $I_{d_\eta}$ and derivatives $\partial^p_\eta \overline G(\eta, \lambda)$, $2 \leq p \leq q$, at $\widetilde\eta = \eta$, and there are $C < \infty$ and $r > 0$ such that
\begin{equation}
\bigl\| \overline G_{\widehat P(\nu)}(\widetilde\eta, \lambda) - (\widetilde\eta - \eta) - \overline R_\nu(\widetilde\eta) \bigr\|
\;\leq\; C\, \| \widetilde\eta - \eta \|^{q+1}
\label{eq:jet-agree}
\end{equation}
for every $\| \widetilde\eta - \eta \| \leq r$ and every $\nu$ in the neighborhood, and $\| \partial^p_{\widetilde\eta} \overline G_{\widehat P(\nu)}(\widetilde\eta, \lambda) \| \leq C$ for $p \leq q+1$ on the same set;
\item under the condition of (ii), if in addition $\sup_{\nu} \| h_q(\cdot\,; \nu) - \overline h_q(\nu) \| \leq B$ $\overline P$-almost surely for some $B < \infty$, the supremum over the $\nu$ in the neighborhood, and $dP_0/d\overline P$ is bounded away from zero, then the density in~\eqref{eq:projection} is nonnegative for $\nu$ close enough to $\nu_0$, so that $\widehat P(\nu)$ is a probability distribution.
\end{enumerate}
\end{lemma}%
The proof is in Appendix~\ref{app:proof-projection}. Part~(i) gives existence, uniqueness, the closed form, and $\widehat P(\nu_0) = P_0$, and part~(ii) gives the condition~\eqref{eq:lipschitz}. Parts~(ii) to~(iv) assume the eigenvalue bound $\gamma_q(\nu) \geq \gamma > 0$ in addition to the rank condition. Part~(iii) gives~\eqref{eq:Gplug}, with a remainder that is uniform over the neighborhood. The value and Jacobian at $\eta$ follow from~\eqref{eq:Phat-moments}, and the derivatives of orders $2$ to $q$ from~\eqref{eq:derivatives}. Only these derivatives enter the expansion, since the $k$-th derivative in $s$ at $s=0$, $k \leq q$, is obtained by differentiating the root equation~\eqref{eq:root} $k$ times and involves derivatives of $\overline G_{\widehat P(\nu)}(\cdot\,, \lambda)$ at $\eta$ of order at most $k$. In particular, the choice of the reference distribution $\overline P$ affects $\widehat P(\nu)$ but not $\psi^{(q)}$. Part~(iv) gives conditions under which $\widehat P(\nu)$ is a probability distribution. In the affine case, the blocks of $h_q$ beyond the second vanish and~\eqref{eq:jet-agree} holds with zero remainder.

The rank condition of Lemma~\ref{lem:projection} is sufficient but not necessary for the constraints to be feasible, and it can fail. For instance, in Example~\ref{ex:grouped} with $d_\eta \geq 2$ the Jacobian $\partial_\eta g = -X_{it} X_{it}'$ is symmetric, so some nonzero linear combination of the entries of $\Lambda(\lambda)\, \partial_\eta g$ vanishes identically. In that example the constraints can also be infeasible. Since $\mathbb{E}_P[X_{it} X_{it}']$ is symmetric for every $P \in \mathcal{P}$, the second condition in~\eqref{eq:Phat-moments} has no solution when $\Lambda(\lambda)$ is not symmetric, and then $\widehat P(\nu)$ does not exist. This does not affect our results, since the expansion only uses the derivatives of $\overline G_{\widehat P(\nu)}(\cdot\,, \lambda)$ at $\eta$ up to order $q$ (which the constraints determine), and the proof of Theorem~\ref{th:nonlinear} in Appendix~\ref{app:proof-thm2} establishes the orthogonality bound directly, without constructing $\widehat P(\nu)$.

\subsection{Expansion around the empirical distribution\label{app:empirical}}

We now explain the construction behind Remark~\ref{rem:empirical}. It uses the von Mises expansion of Lemma~\ref{lem:expansion} with a different choice of baseline distribution and nuisance parameter.

In Subsection~\ref{sec:vm-expansion}, the restrictions~\eqref{eq:derivatives} fix the derivatives of orders $2$ to $q$ of the moment function under $\widehat P(\eta,\lambda)$ at their values under $P_0$. These values are estimated without bias using independent copies of the data, so they do not appear as nuisance parameters. Alternatively, we can leave these derivatives free and treat them as additional nuisance parameters. Let $\omega := (\eta, \lambda, b_2, \ldots, b_q)$, where each $b_p$ is a symmetric $p$-linear map from $(\mathbb{R}^{d_\eta})^p$ to $\mathbb{R}^{d_\eta}$. Let $\widetilde P(\omega)$ be defined as in~\eqref{eq:Phat}, but with~\eqref{eq:derivatives} replaced by
\[
\Lambda(\lambda)\, \mathbb{E}_{\widetilde P(\omega)}\bigl[ \partial^p_\eta g(W_i, \theta_0, \eta) \bigr] \;=\; b_p , \qquad p = 2, \ldots, q .
\]
When $d_g > d_\eta$, the first restriction in~\eqref{eq:Phat-moments} is used in the weighted form of the footnote before~\eqref{eq:Gplug}, with $g$ replaced by $\Lambda(\lambda)\, g$. The true value $\omega_0$ of the nuisance parameter consists of $\eta_0$, $\lambda_0$ and $b_{p,0} := \Lambda_0\, \mathbb{E}_{P_0}[\partial^p_\eta g(W_i, \theta_0, \eta_0)]$, and $\widetilde P(\omega_0) = P_0$. The expansion~\eqref{eq:PsiP} along the segment from $\widetilde P(\omega)$ to $P_0$ then gives a moment function with nuisance parameter $\omega$ that is orthogonal to order $q$. This follows from Lemma~\ref{lem:expansion} when the distributions $\widetilde P(\omega)$ exist and satisfy~\eqref{eq:lipschitz}, and it can also be shown directly along the lines of the proof of Theorem~\ref{th:nonlinear}.

The expansion along the segment from a distribution $P$ to $P_0$ depends on $P$ only through the root $\eta(P)$ of~\eqref{eq:nu-of-P}, the Jacobian $\mathbb{E}_P[\partial_\eta g(W_i, \theta_0, \eta(P))]$, and the derivatives $\mathbb{E}_P[\partial^p_\eta g(W_i, \theta_0, \eta(P))]$ for $2 \leq p \leq q$, as explained after Lemma~\ref{lem:projection}. Let $\widehat P_n$ be the empirical distribution of the part of the sample that is used for the preliminary estimators, so that $\widehat P_n$ is independent of the observations in the U-statistic~\eqref{UstatsPsi}. Let $\widehat\omega := (\widehat\eta, \widehat\lambda, \widehat b_2, \ldots, \widehat b_q)$ collect the corresponding sample quantities, namely the root $\widehat\eta$ of~\eqref{eq:nu-of-P} at $P = \widehat P_n$, a $\widehat\lambda$ with $\Lambda(\widehat\lambda)\, \mathbb{E}_{\widehat P_n}[\partial_\eta g(W_i, \theta_0, \widehat\eta)] = I_{d_\eta}$ such that $\Lambda(\widehat\lambda)$ has the same null space as $\Lambda(\check\lambda)$, and $\widehat b_p := \Lambda(\widehat\lambda)\, \mathbb{E}_{\widehat P_n}[\partial^p_\eta g(W_i, \theta_0, \widehat\eta)]$. Whenever $\widetilde P(\widehat\omega)$ exists, it has the same root as $\widehat P_n$ and the same derivatives up to order $q$ after premultiplication by $\Lambda(\widehat\lambda)$, so the expansions around the two distributions agree. The resulting moment function is defined by these quantities even when $\widetilde P(\widehat\omega)$ does not exist. Centering at the empirical distribution therefore amounts to the moment function of the previous paragraph, with $\omega$ estimated by $\widehat\omega$.

For $q = 2$, and with the arguments $\theta$ and $\eta$ of $m$ and $g$ suppressed, this moment function is
\begin{align*}
& m(W_1) - 2\, \partial_\eta m(W_1)'\, \Lambda\, g(W_2) + \partial_\eta m(W_1)'\, \Lambda\, \partial_\eta g(W_3)\, \Lambda\, g(W_2) \\
& \quad + \tfrac{1}{2}\, \bigl[ \partial^2_\eta m(W_1) - \partial_\eta m(W_1)'\, b_2 \bigr] \bigl[ \Lambda\, g(W_2),\, \Lambda\, g(W_3) \bigr] ,
\end{align*}
with $\Lambda = \Lambda(\lambda)$. Compared with~\eqref{eq:psi2-nonlinear}, the factor $\Lambda\, \partial^2_\eta g(W_2)$, which is evaluated at an additional copy of the data, is replaced by the nuisance parameter $b_2$. In general, the moment function with nuisance parameter $\omega$ uses at most $q+1$ independent copies of the data instead of up to $2q$, but it requires the additional nuisance parameters $b_2, \ldots, b_q$, whose number of unknown components can grow with $q$. Our restrictions~\eqref{eq:derivatives} avoid these additional nuisance parameters.

\subsection{Tree-indexed Taylor expansions\label{app:primer}}

We now relate the tree expansion of Section~\ref{sec:trees} to the literature on Taylor expansions indexed by rooted trees, and give a recursion that computes $\sigma(\tau)$. The closed-form coefficient
\[
  c_{q,\tau} \;=\; \frac{(-1)^{|\tau|}}{\sigma(\tau)}\,
  \binom{q+|\tau|-d(\tau)}{|\tau|}
\]
of Theorem~\ref{th:nonlinear} and the trees it weights involve three ingredients, and two of them are standard in that literature. The central object of that literature is the $B$-series of \citet{Butcher1963}, which writes the derivatives of the solution of an ordinary differential equation (ODE) as sums over rooted trees, with the factor $1/\sigma(\tau)$ in its coefficients. An early version of this picture is in \citet{Cayley1857}, and textbook treatments are \citet{HairerLubichWanner2006}, Ch.~III.1, and the survey of \citet{McLachlanModinMuntheKaasVerdier2017}.

$\sigma(\tau)$ is the symmetry coefficient, and it is defined either as the order of the group of automorphisms of $\tau$ or by the recursion we give at the end of this subsection, and the two definitions agree (\citealp{HairerLubichWanner2006}, Ch.~III.1; \citealp{ChartierHairerVilmart2010}, Sect.~2.1).

The one-child nodes, which carry the matrix $\Lambda(\lambda)\, \partial_\eta \overline g$, are also a known device, though not one from ordinary $B$-series. The order conditions for Runge--Kutta methods that replace the Jacobian by an approximation are indexed by trees with two kinds of nodes, in which the nodes that carry the approximate Jacobian each have a single child (\citealp{HairerWanner1996}, Sect.~IV.7). In our construction such nodes arise in the same way, from expanding the inverse of the linearized moment condition in powers of $I_{d_\eta} - \Lambda(\lambda)\, \partial_\eta \overline g$.

The binomial $\binom{q+|\tau|-d(\tau)}{|\tau|}$ and the index set $\mathcal{T}_q$ are specific to the orthogonalization problem. In the numerical analysis literature the coefficient of a tree does not depend on the order at which the series is truncated, and the truncation is by the number of nodes. Here the coefficient depends on $q$, and the truncation keeps the trees with $d(\tau) \leq q$, so that trees with as many as $2q-1$ non-root nodes still contribute. We are not aware of a tree coefficient of this form elsewhere.

\paragraph{Computing $\sigma(\tau)$.}

The symmetry coefficient $\sigma(\tau)$ admits a simple recursion. Let the root of $\tau$ have $r$ children. Deleting the root splits $\tau$ into $r$ rooted subtrees, one at each child. Group these subtrees into isomorphism classes (as rooted trees). For $i = 1, \ldots, s$, let $\tau_i$ denote a representative subtree of the $i$-th class and $n_i$ the number of subtrees in that class, so $n_1 + \cdots + n_s = r$. An automorphism of $\tau$ fixes the root, permutes the $n_i$ subtrees within each class freely, and acts on each subtree by one of its own automorphisms. This gives
\[
 \sigma(\tau) \;=\; \prod_{i=1}^s n_i! \; [\sigma(\tau_i)]^{n_i},
\]
with the recursion bottoming out at the trivial single-node tree, where $\sigma(\bullet) = 1$. The tree-isomorphism test that determines the classes admits linear-time implementations (\citealp{Valiente2002}, Ch.~4).

\subsection{Explicit formulas at $q = 3$}\label{app:formulas-q3}

We now display the orthogonal moment function at $q = 3$, which Section~\ref{sec:nonlinear} only describes. The formulas at $q = 4$ are given in Appendix~\ref{app:formulas}.

At $q = 3$ the affine moment function~\eqref{eq:score} has seven terms, where for compactness we suppress the arguments $(\theta, \eta)$ in $m$ and $g$ and the argument $\lambda$ in $\Lambda$,
\begin{align}
\psi^{(3)}_{\rm aff}
&= m(W_1) - 3\, \partial_\eta m(W_1)' \,\Lambda\, g(W_2) \notag\\
&\quad + 3\, \partial_\eta m(W_1)' \,\big(\Lambda \,\partial_\eta g(W_3)\big)\, \Lambda\, g(W_2) \notag\\
&\quad - \partial_\eta m(W_1)' \,\big(\Lambda \,\partial_\eta g(W_3)\big) \big(\Lambda \,\partial_\eta g(W_4)\big)\, \Lambda\, g(W_2) \notag\\
&\quad + \tfrac{3}{2}\, \partial_\eta^2 m(W_1)\big[\Lambda\, g(W_2),\, \Lambda\, g(W_3)\big] \notag\\
&\quad - \partial_\eta^2 m(W_1)\big[\Lambda\, g(W_2),\, \big(\Lambda \,\partial_\eta g(W_4)\big)\, \Lambda\, g(W_3)\big] \notag\\
&\quad - \tfrac{1}{6}\, \partial_\eta^3 m(W_1)\big[\Lambda\, g(W_2),\, \Lambda\, g(W_3),\, \Lambda\, g(W_4)\big].
\label{eq:psi3-aff}
\end{align}
With the six correction terms that nonlinearity of $g$ requires at $q = 3$, the orthogonal moment function is 
\begin{align}
\psi^{(3)}
&= \psi^{(3)}_{\rm aff} \notag\\
&\quad -\, 2\, \partial_\eta m(W_1)'\, \Lambda\, \partial_\eta^2 g(W_2)\big[\Lambda g(W_3),\, \Lambda g(W_4)\big] \notag\\
&\quad +\, \tfrac{1}{2}\, \partial_\eta m(W_1)'\, \Lambda\, \partial_\eta g(W_2)\, \Lambda\, \partial_\eta^2 g(W_3)\big[\Lambda g(W_4),\, \Lambda g(W_5)\big] \notag\\
&\quad +\, \partial_\eta m(W_1)'\, \Lambda\, \partial_\eta^2 g(W_2)\Big[\Lambda\, \partial_\eta g(W_3)\, \Lambda g(W_4),\; \Lambda g(W_5)\Big] \notag\\
&\quad +\, \tfrac{1}{6}\, \partial_\eta m(W_1)'\, \Lambda\, \partial_\eta^3 g(W_2)\big[\Lambda g(W_3),\, \Lambda g(W_4),\, \Lambda g(W_5)\big] \notag\\
&\quad -\, \tfrac{1}{2}\, \partial_\eta m(W_1)'\, \Lambda\, \partial_\eta^2 g(W_2)\Big[\Lambda\, \partial_\eta^2 g(W_3)\big[\Lambda g(W_4),\, \Lambda g(W_5)\big],\; \Lambda g(W_6) \Big] \notag\\
&\quad +\, \tfrac{1}{2}\, \partial_\eta^2 m(W_1)\big[\Lambda\, g(W_2),\, \Lambda\, \partial_\eta^2 g(W_3)\big[\Lambda g(W_4),\, \Lambda g(W_5)\big]\big].
\label{eq:psi3-nonlinear}
\end{align}
The combinatorial pre-factors in~\eqref{eq:psi3-aff} and~\eqref{eq:psi3-nonlinear} are the values needed for third-order Neyman orthogonality. The affine pre-factors can be obtained from~\eqref{eq:score} by collecting the tuples $(k_1, \ldots, k_r)$ that give rise to the same term. The correction pre-factors are the values of $c_{3,\tau}$ in~\eqref{eq:cq-formula}, computed from the invariants in Figure~\ref{fig:invariants-q3}. The number of independent copies is $L = 4$ for the affine terms and $L = 6$ once the corrections are included.

\subsection{Explicit formulas at $q = 4$\label{app:formulas}}

This appendix gives the explicit orthogonal moment function $\psi^{(4)}$ at $q = 4$, in both the affine and nonlinear cases. The construction follows Section~\ref{sec:trees}: $\psi^{(4)} = \sum_{\tau \in \mathcal{T}_4} c_{4,\tau}\, \kappa_\tau$, where $\mathcal{T}_4$ contains $|\mathcal{T}_4| = 40$ trees and the closed-form coefficient~\eqref{eq:cq-formula} is
\[
c_{q,\tau} \;=\; \frac{(-1)^{|\tau|}}{\sigma(\tau)}\, \binom{q + |\tau| - d(\tau)}{|\tau|}.
\]
We display $\mathcal{T}_4$, list $c_{4,\tau}$ for every tree, and compute $\psi^{(4)}$. Throughout we write $g(W_j)$ for $g(W_j, \theta, \eta)$, $\partial_\eta^p g(W_j)$ for $\partial_\eta^p g(W_j, \theta, \eta)$, $m_\eta = \partial_\eta m(W_1, \theta, \eta)$, $m_{\eta\eta} = \partial^2_{\eta\eta'} m(W_1, \theta, \eta)$, $m_{\eta\eta\eta} = \partial_\eta^3 m(W_1, \theta, \eta)$, $m_{\eta\eta\eta\eta} = \partial_\eta^4 m(W_1, \theta, \eta)$, and $\Lambda = \Lambda(\lambda)$.

\subsubsection*{The 40 trees in $\mathcal{T}_4$}

Among the 40 trees, $12$ have $|\tau| = d(\tau)$. In these trees every non-root node has at most one child, and they give the affine terms (Figure~\ref{fig:trees-q4-aff}). The other $28$ have $|\tau| > d(\tau)$, so at least one non-root node has two or more children, and they give the correction terms (Figure~\ref{fig:trees-q4-corr}).

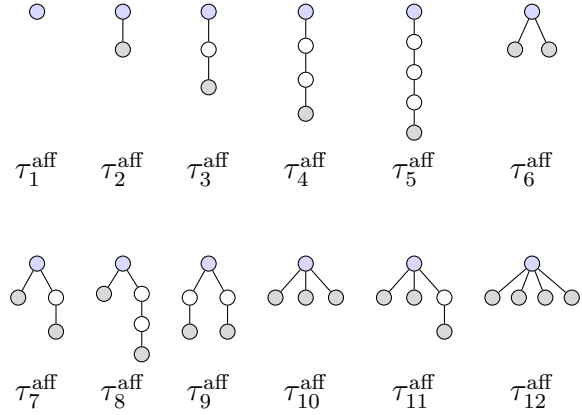
\begin{figure}[ht]
\centering
\renewcommand{\arraystretch}{1.4}
\begin{tabular}{cccccc}
\begin{tikzpicture}[baseline=(root.base),
    every node/.style={draw, circle, inner sep=2pt, minimum size=5pt}]
\node[fill=blue!15] (root) {};
\end{tikzpicture}
&
\begin{tikzpicture}[baseline=(root.base),
    level distance=0.5cm, sibling distance=0.6cm,
    every node/.style={draw, circle, inner sep=2pt, minimum size=5pt},
    edge from parent/.style={draw, -}]
\node[fill=blue!15] (root) {} child { node[fill=gray!30] {} };
\end{tikzpicture}
&
\begin{tikzpicture}[baseline=(root.base),
    level distance=0.5cm, sibling distance=0.6cm,
    every node/.style={draw, circle, inner sep=2pt, minimum size=5pt},
    edge from parent/.style={draw, -}]
\node[fill=blue!15] (root) {} child { node {} child { node[fill=gray!30] {} } };
\end{tikzpicture}
&
\begin{tikzpicture}[baseline=(root.base),
    level distance=0.45cm, sibling distance=0.6cm,
    every node/.style={draw, circle, inner sep=2pt, minimum size=5pt},
    edge from parent/.style={draw, -}]
\node[fill=blue!15] (root) {} child { node {} child { node {} child { node[fill=gray!30] {} } } };
\end{tikzpicture}
&
\begin{tikzpicture}[baseline=(root.base),
    level distance=0.4cm, sibling distance=0.6cm,
    every node/.style={draw, circle, inner sep=2pt, minimum size=5pt},
    edge from parent/.style={draw, -}]
\node[fill=blue!15] (root) {} child { node {} child { node {} child { node {} child { node[fill=gray!30] {} } } } };
\end{tikzpicture}
&
\begin{tikzpicture}[baseline=(root.base),
    level distance=0.5cm, sibling distance=0.45cm,
    every node/.style={draw, circle, inner sep=2pt, minimum size=5pt},
    edge from parent/.style={draw, -}]
\node[fill=blue!15] (root) {}
  child { node[fill=gray!30] {} }
  child { node[fill=gray!30] {} };
\end{tikzpicture}
\\
$\tau^{\rm aff}_1$ & $\tau^{\rm aff}_2$ & $\tau^{\rm aff}_3$ & $\tau^{\rm aff}_4$ & $\tau^{\rm aff}_5$ & $\tau^{\rm aff}_6$
\\[1em]
\begin{tikzpicture}[baseline=(root.base),
    level distance=0.45cm, sibling distance=0.5cm,
    every node/.style={draw, circle, inner sep=2pt, minimum size=5pt},
    edge from parent/.style={draw, -}]
\node[fill=blue!15] (root) {}
  child { node[fill=gray!30] {} }
  child { node {} child { node[fill=gray!30] {} } };
\end{tikzpicture}
&
\begin{tikzpicture}[baseline=(root.base),
    level distance=0.4cm, sibling distance=0.5cm,
    every node/.style={draw, circle, inner sep=2pt, minimum size=5pt},
    edge from parent/.style={draw, -}]
\node[fill=blue!15] (root) {}
  child { node[fill=gray!30] {} }
  child { node {} child { node {} child { node[fill=gray!30] {} } } };
\end{tikzpicture}
&
\begin{tikzpicture}[baseline=(root.base),
    level distance=0.45cm, sibling distance=0.5cm,
    every node/.style={draw, circle, inner sep=2pt, minimum size=5pt},
    edge from parent/.style={draw, -}]
\node[fill=blue!15] (root) {}
  child { node {} child { node[fill=gray!30] {} } }
  child { node {} child { node[fill=gray!30] {} } };
\end{tikzpicture}
&
\begin{tikzpicture}[baseline=(root.base),
    level distance=0.45cm, sibling distance=0.4cm,
    every node/.style={draw, circle, inner sep=2pt, minimum size=5pt},
    edge from parent/.style={draw, -}]
\node[fill=blue!15] (root) {}
  child { node[fill=gray!30] {} }
  child { node[fill=gray!30] {} }
  child { node[fill=gray!30] {} };
\end{tikzpicture}
&
\begin{tikzpicture}[baseline=(root.base),
    level distance=0.45cm, sibling distance=0.4cm,
    every node/.style={draw, circle, inner sep=2pt, minimum size=5pt},
    edge from parent/.style={draw, -}]
\node[fill=blue!15] (root) {}
  child { node[fill=gray!30] {} }
  child { node[fill=gray!30] {} }
  child { node {} child { node[fill=gray!30] {} } };
\end{tikzpicture}
&
\begin{tikzpicture}[baseline=(root.base),
    level distance=0.45cm, sibling distance=0.35cm,
    every node/.style={draw, circle, inner sep=2pt, minimum size=5pt},
    edge from parent/.style={draw, -}]
\node[fill=blue!15] (root) {}
  child { node[fill=gray!30] {} }
  child { node[fill=gray!30] {} }
  child { node[fill=gray!30] {} }
  child { node[fill=gray!30] {} };
\end{tikzpicture}
\\
$\tau^{\rm aff}_7$ & $\tau^{\rm aff}_8$ & $\tau^{\rm aff}_9$ & $\tau^{\rm aff}_{10}$ & $\tau^{\rm aff}_{11}$ & $\tau^{\rm aff}_{12}$
\end{tabular}
\caption{The 12 trees in $\mathcal{T}_4$ that correspond to affine terms in $\psi^{(4)}$.}
\label{fig:trees-q4-aff}
\end{figure}

\begin{figure}[ht]
\centering
\renewcommand{\arraystretch}{1.4}
\begin{tabular}{ccccccc}
\begin{tikzpicture}[baseline=(root.base),
    level distance=0.45cm, sibling distance=0.45cm,
    every node/.style={draw, circle, inner sep=2pt, minimum size=5pt},
    edge from parent/.style={draw, -}]
\node[fill=blue!15] (root) {} child { node {} child { node[fill=gray!30] {} } child { node[fill=gray!30] {} } };
\end{tikzpicture}
&
\begin{tikzpicture}[baseline=(root.base),
    level distance=0.45cm, sibling distance=0.4cm,
    every node/.style={draw, circle, inner sep=2pt, minimum size=5pt},
    edge from parent/.style={draw, -}]
\node[fill=blue!15] (root) {} child { node {} child { node[fill=gray!30] {} } child { node[fill=gray!30] {} } child { node[fill=gray!30] {} } };
\end{tikzpicture}
&
\begin{tikzpicture}[baseline=(root.base),
    level distance=0.45cm, sibling distance=0.45cm,
    every node/.style={draw, circle, inner sep=2pt, minimum size=5pt},
    edge from parent/.style={draw, -}]
\node[fill=blue!15] (root) {} child { node[fill=gray!30] {} } child { node {} child { node[fill=gray!30] {} } child { node[fill=gray!30] {} } };
\end{tikzpicture}
&
\begin{tikzpicture}[baseline=(root.base),
    level distance=0.4cm, sibling distance=0.4cm,
    every node/.style={draw, circle, inner sep=2pt, minimum size=5pt},
    edge from parent/.style={draw, -}]
\node[fill=blue!15] (root) {} child { node {} child { node[fill=gray!30] {} } child { node {} child { node[fill=gray!30] {} } child { node[fill=gray!30] {} } } };
\end{tikzpicture}
&
\begin{tikzpicture}[baseline=(root.base),
    level distance=0.4cm, sibling distance=0.45cm,
    every node/.style={draw, circle, inner sep=2pt, minimum size=5pt},
    edge from parent/.style={draw, -}]
\node[fill=blue!15] (root) {} child { node {} child { node {} child { node[fill=gray!30] {} } child { node[fill=gray!30] {} } } };
\end{tikzpicture}
&
\begin{tikzpicture}[baseline=(root.base),
    level distance=0.4cm, sibling distance=0.45cm,
    every node/.style={draw, circle, inner sep=2pt, minimum size=5pt},
    edge from parent/.style={draw, -}]
\node[fill=blue!15] (root) {} child { node {} child { node[fill=gray!30] {} } child { node {} child { node[fill=gray!30] {} } } };
\end{tikzpicture}
&
\begin{tikzpicture}[baseline=(root.base),
    level distance=0.4cm, sibling distance=0.35cm,
    every node/.style={draw, circle, inner sep=2pt, minimum size=5pt},
    edge from parent/.style={draw, -}]
\node[fill=blue!15] (root) {} child { node {} child { node[fill=gray!30] {} } child { node[fill=gray!30] {} } child { node[fill=gray!30] {} } child { node[fill=gray!30] {} } };
\end{tikzpicture}
\\
$\tau^{\rm corr}_1$ & $\tau^{\rm corr}_2$ & $\tau^{\rm corr}_3$ & $\tau^{\rm corr}_4$ & $\tau^{\rm corr}_5$ & $\tau^{\rm corr}_6$ & $\tau^{\rm corr}_7$
\\[1em]
\begin{tikzpicture}[baseline=(root.base),
    level distance=0.4cm, sibling distance=0.4cm,
    every node/.style={draw, circle, inner sep=2pt, minimum size=5pt},
    edge from parent/.style={draw, -}]
\node[fill=blue!15] (root) {}
  child { node[fill=gray!30] {} }
  child { node {} child { node[fill=gray!30] {} } child { node[fill=gray!30] {} } child { node[fill=gray!30] {} } };
\end{tikzpicture}
&
\begin{tikzpicture}[baseline=(root.base),
    level distance=0.4cm, sibling distance=0.4cm,
    every node/.style={draw, circle, inner sep=2pt, minimum size=5pt},
    edge from parent/.style={draw, -}]
\node[fill=blue!15] (root) {}
  child { node[fill=gray!30] {} }
  child { node[fill=gray!30] {} }
  child { node {} child { node[fill=gray!30] {} } child { node[fill=gray!30] {} } };
\end{tikzpicture}
&
\begin{tikzpicture}[baseline=(root.base),
    level distance=0.35cm, sibling distance=0.4cm,
    every node/.style={draw, circle, inner sep=2pt, minimum size=5pt},
    edge from parent/.style={draw, -}]
\node[fill=blue!15] (root) {} child { node {}
  child { node[fill=gray!30] {} }
  child { node {} child { node[fill=gray!30] {} } child { node[fill=gray!30] {} } child { node[fill=gray!30] {} } } };
\end{tikzpicture}
&
\begin{tikzpicture}[baseline=(root.base),
    level distance=0.4cm, sibling distance=0.35cm,
    every node/.style={draw, circle, inner sep=2pt, minimum size=5pt},
    edge from parent/.style={draw, -}]
\node[fill=blue!15] (root) {} child { node {}
  child { node[fill=gray!30] {} }
  child { node[fill=gray!30] {} }
  child { node {} child { node[fill=gray!30] {} } child { node[fill=gray!30] {} } } };
\end{tikzpicture}
&
\begin{tikzpicture}[baseline=(root.base),
    level distance=0.35cm, sibling distance=0.4cm,
    every node/.style={draw, circle, inner sep=2pt, minimum size=5pt},
    edge from parent/.style={draw, -}]
\node[fill=blue!15] (root) {}
  child { node[fill=gray!30] {} }
  child { node {}
    child { node[fill=gray!30] {} }
    child { node {} child { node[fill=gray!30] {} } child { node[fill=gray!30] {} } } };
\end{tikzpicture}
&
\begin{tikzpicture}[baseline=(root.base),
    level distance=0.4cm,
    level 1/.style={sibling distance=0.7cm},
    level 2/.style={sibling distance=0.32cm},
    every node/.style={draw, circle, inner sep=2pt, minimum size=5pt},
    edge from parent/.style={draw, -}]
\node[fill=blue!15] (root) {}
  child { node {} child { node[fill=gray!30] {} } child { node[fill=gray!30] {} } }
  child { node {} child { node[fill=gray!30] {} } child { node[fill=gray!30] {} } };
\end{tikzpicture}
&
\begin{tikzpicture}[baseline=(root.base),
    level distance=0.32cm, sibling distance=0.35cm,
    every node/.style={draw, circle, inner sep=2pt, minimum size=5pt},
    edge from parent/.style={draw, -}]
\node[fill=blue!15] (root) {} child { node {}
  child { node[fill=gray!30] {} }
  child { node {}
    child { node[fill=gray!30] {} }
    child { node {} child { node[fill=gray!30] {} } child { node[fill=gray!30] {} } } } };
\end{tikzpicture}
\\
$\tau^{\rm corr}_8$ & $\tau^{\rm corr}_9$ & $\tau^{\rm corr}_{10}$ & $\tau^{\rm corr}_{11}$ & $\tau^{\rm corr}_{12}$ & $\tau^{\rm corr}_{13}$ & $\tau^{\rm corr}_{14}$
\\[1em]
\begin{tikzpicture}[baseline=(root.base),
    level distance=0.35cm,
    level 2/.style={sibling distance=0.7cm},
    level 3/.style={sibling distance=0.32cm},
    every node/.style={draw, circle, inner sep=2pt, minimum size=5pt},
    edge from parent/.style={draw, -}]
\node[fill=blue!15] (root) {} child { node {}
  child { node {} child { node[fill=gray!30] {} } child { node[fill=gray!30] {} } }
  child { node {} child { node[fill=gray!30] {} } child { node[fill=gray!30] {} } } };
\end{tikzpicture}
&
\begin{tikzpicture}[baseline=(root.base),
    level distance=0.4cm, sibling distance=0.35cm,
    every node/.style={draw, circle, inner sep=2pt, minimum size=5pt},
    edge from parent/.style={draw, -}]
\node[fill=blue!15] (root) {} child { node {} child { node {} child { node[fill=gray!30] {} } child { node[fill=gray!30] {} } child { node[fill=gray!30] {} } } };
\end{tikzpicture}
&
\begin{tikzpicture}[baseline=(root.base),
    level distance=0.35cm, sibling distance=0.35cm,
    every node/.style={draw, circle, inner sep=2pt, minimum size=5pt},
    edge from parent/.style={draw, -}]
\node[fill=blue!15] (root) {} child { node {}
  child { node[fill=gray!30] {} }
  child { node[fill=gray!30] {} }
  child { node {} child { node[fill=gray!30] {} } } };
\end{tikzpicture}
&
\begin{tikzpicture}[baseline=(root.base),
    level distance=0.35cm, sibling distance=0.4cm,
    every node/.style={draw, circle, inner sep=2pt, minimum size=5pt},
    edge from parent/.style={draw, -}]
\node[fill=blue!15] (root) {}
  child { node {} child { node[fill=gray!30] {} } }
  child { node {} child { node[fill=gray!30] {} } child { node[fill=gray!30] {} } };
\end{tikzpicture}
&
\begin{tikzpicture}[baseline=(root.base),
    level distance=0.35cm, sibling distance=0.4cm,
    every node/.style={draw, circle, inner sep=2pt, minimum size=5pt},
    edge from parent/.style={draw, -}]
\node[fill=blue!15] (root) {}
  child { node[fill=gray!30] {} }
  child { node {} child { node {} child { node[fill=gray!30] {} } child { node[fill=gray!30] {} } } };
\end{tikzpicture}
&
\begin{tikzpicture}[baseline=(root.base),
    level distance=0.35cm, sibling distance=0.4cm,
    every node/.style={draw, circle, inner sep=2pt, minimum size=5pt},
    edge from parent/.style={draw, -}]
\node[fill=blue!15] (root) {}
  child { node[fill=gray!30] {} }
  child { node {} child { node[fill=gray!30] {} } child { node {} child { node[fill=gray!30] {} } } };
\end{tikzpicture}
&
\begin{tikzpicture}[baseline=(root.base),
    level distance=0.32cm, sibling distance=0.35cm,
    every node/.style={draw, circle, inner sep=2pt, minimum size=5pt},
    edge from parent/.style={draw, -}]
\node[fill=blue!15] (root) {} child { node {} child { node {}
  child { node[fill=gray!30] {} }
  child { node {} child { node[fill=gray!30] {} } child { node[fill=gray!30] {} } } } };
\end{tikzpicture}
\\
$\tau^{\rm corr}_{15}$ & $\tau^{\rm corr}_{16}$ & $\tau^{\rm corr}_{17}$ & $\tau^{\rm corr}_{18}$ & $\tau^{\rm corr}_{19}$ & $\tau^{\rm corr}_{20}$ & $\tau^{\rm corr}_{21}$
\\[1em]
\begin{tikzpicture}[baseline=(root.base),
    level distance=0.32cm,
    level 2/.style={sibling distance=0.6cm},
    level 3/.style={sibling distance=0.32cm},
    every node/.style={draw, circle, inner sep=2pt, minimum size=5pt},
    edge from parent/.style={draw, -}]
\node[fill=blue!15] (root) {} child { node {}
  child { node {} child { node[fill=gray!30] {} } }
  child { node {} child { node[fill=gray!30] {} } child { node[fill=gray!30] {} } } };
\end{tikzpicture}
&
\begin{tikzpicture}[baseline=(root.base),
    level distance=0.32cm, sibling distance=0.35cm,
    every node/.style={draw, circle, inner sep=2pt, minimum size=5pt},
    edge from parent/.style={draw, -}]
\node[fill=blue!15] (root) {} child { node {}
  child { node[fill=gray!30] {} }
  child { node {} child { node {} child { node[fill=gray!30] {} } child { node[fill=gray!30] {} } } } };
\end{tikzpicture}
&
\begin{tikzpicture}[baseline=(root.base),
    level distance=0.32cm, sibling distance=0.32cm,
    every node/.style={draw, circle, inner sep=2pt, minimum size=5pt},
    edge from parent/.style={draw, -}]
\node[fill=blue!15] (root) {} child { node {}
  child { node[fill=gray!30] {} }
  child { node {}
    child { node[fill=gray!30] {} }
    child { node {} child { node[fill=gray!30] {} } } } };
\end{tikzpicture}
&
\begin{tikzpicture}[baseline=(root.base),
    level distance=0.32cm, sibling distance=0.35cm,
    every node/.style={draw, circle, inner sep=2pt, minimum size=5pt},
    edge from parent/.style={draw, -}]
\node[fill=blue!15] (root) {} child { node {} child { node {} child { node {}
  child { node[fill=gray!30] {} }
  child { node[fill=gray!30] {} } } } };
\end{tikzpicture}
&
\begin{tikzpicture}[baseline=(root.base),
    level distance=0.32cm, sibling distance=0.35cm,
    every node/.style={draw, circle, inner sep=2pt, minimum size=5pt},
    edge from parent/.style={draw, -}]
\node[fill=blue!15] (root) {} child { node {}
  child { node[fill=gray!30] {} }
  child { node {} child { node {} child { node[fill=gray!30] {} } } } };
\end{tikzpicture}
&
\begin{tikzpicture}[baseline=(root.base),
    level distance=0.32cm, sibling distance=0.35cm,
    every node/.style={draw, circle, inner sep=2pt, minimum size=5pt},
    edge from parent/.style={draw, -}]
\node[fill=blue!15] (root) {} child { node {} child { node {}
  child { node[fill=gray!30] {} }
  child { node {} child { node[fill=gray!30] {} } } } };
\end{tikzpicture}
&
\begin{tikzpicture}[baseline=(root.base),
    level distance=0.32cm, sibling distance=0.35cm,
    every node/.style={draw, circle, inner sep=2pt, minimum size=5pt},
    edge from parent/.style={draw, -}]
\node[fill=blue!15] (root) {} child { node {}
  child { node {} child { node[fill=gray!30] {} } }
  child { node {} child { node[fill=gray!30] {} } } };
\end{tikzpicture}
\\
$\tau^{\rm corr}_{22}$ & $\tau^{\rm corr}_{23}$ & $\tau^{\rm corr}_{24}$ & $\tau^{\rm corr}_{25}$ & $\tau^{\rm corr}_{26}$ & $\tau^{\rm corr}_{27}$ & $\tau^{\rm corr}_{28}$
\end{tabular}
\caption{The 28 trees in $\mathcal{T}_4$ that correspond to non-linear correction terms in $\psi^{(4)}$.}
\label{fig:trees-q4-corr}
\end{figure}
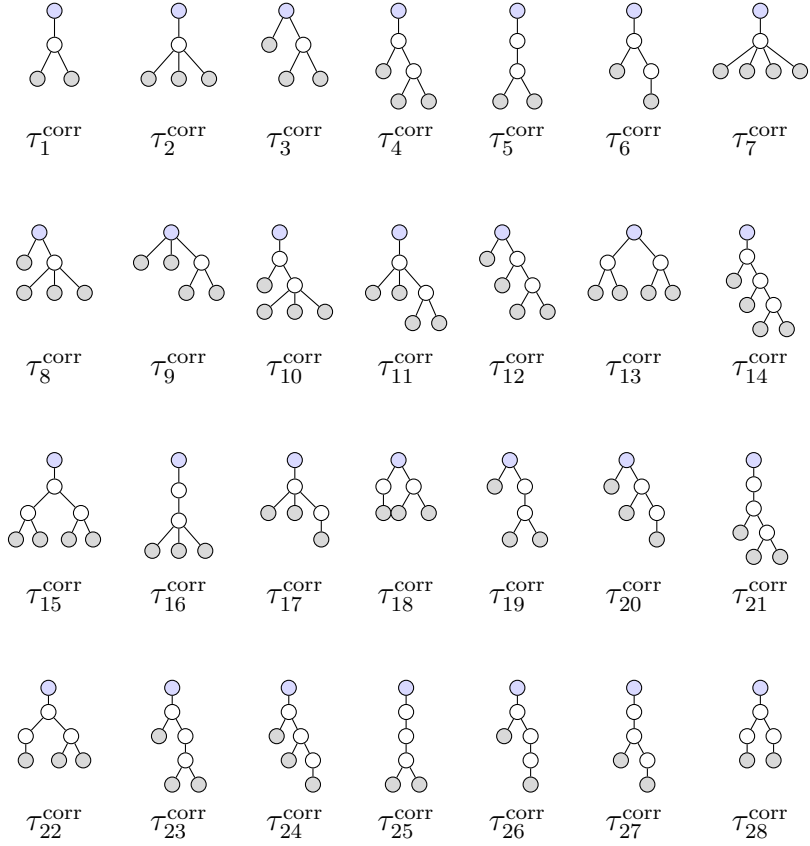

\subsubsection*{Example: $\kappa_\tau$ and $\sigma(\tau)$ for $\tau^{\rm corr}_{15}$}

We illustrate the construction of the term $\kappa_\tau$ from Section~\ref{sec:trees} and the recursive computation of $\sigma(\tau)$ from Appendix~\ref{app:primer} on the balanced tree $\tau^{\rm corr}_{15}$, in which the root has one child, that child has two children, and each of those grandchildren has two leaf children. Labeling the root by $W_1$, the top internal by $W_2$, the two middle internals by $W_3, W_4$, and the four leaves by $W_5, W_6, W_7, W_8$ (any choice of distinct copies works), bottom-up evaluation gives, at each leaf, $\Lambda\, g(W_j)$; at each middle internal, $\Lambda\, \partial_\eta^2 g(W_j)$ contracted with the values of its two leaf-children; at the top internal, $\Lambda\, \partial_\eta^2 g(W_2)$ contracted with the values of the two middle internals; and at the root, $m_\eta'$ contracted with the value of the top internal. The result is
\[
\kappa_{\tau^{\rm corr}_{15}} \;=\; m_\eta' \,\Lambda\, \partial_\eta^2 g(W_2)\Big[
\Lambda\, \partial_\eta^2 g(W_3)\big[\Lambda g(W_5),\, \Lambda g(W_6)\big],\,
\Lambda\, \partial_\eta^2 g(W_4)\big[\Lambda g(W_7),\, \Lambda g(W_8)\big]\Big].
\]

For $\sigma(\tau^{\rm corr}_{15})$, apply the recursion at each node: at each of the two middle internal nodes, the two leaf children form a single isomorphism class with $n = 2$, contributing $2! \cdot 1^2 = 2$, so each middle subtree has $\sigma = 2$. At the top internal node, the two middle children are isomorphic ($n = 2$, factor $2!$) and each has $\sigma = 2$, so the top subtree has $\sigma = 2! \cdot 2^2 = 8$. At the root, there is a single child ($n = 1$, factor $1!$), so $\sigma(\tau^{\rm corr}_{15}) = 1! \cdot 8 = 8$. With $|\tau^{\rm corr}_{15}| = 7$ and $d(\tau^{\rm corr}_{15}) = 4$, the closed form gives $c_{4, \tau^{\rm corr}_{15}} = (-1)^7 \binom{7}{7}/8 = -1/8$.

The same recursion gives the values $\sigma(\tau^{\rm aff}_5) = 1$ (chain of length $4$, no symmetry), $\sigma(\tau^{\rm aff}_{12}) = 4! = 24$ (root with four leaf children, all isomorphic), $\sigma(\tau^{\rm corr}_7) = 4! = 24$ (the root has one child, which has four leaf children),  $\sigma(\tau^{\rm corr}_9) = 2! \cdot 2 = 4$ (the root has three children, of which the two leaves form one class, while the third has two leaf children of its own and so contributes $\sigma = 2$), and  $\sigma(\tau^{\rm corr}_{13}) = 2! \cdot 2^2 = 8$ (root with two isomorphic children, each carrying two leaves).

\subsubsection*{Coefficient table}

Tables~\ref{tab:c4-aff} and~\ref{tab:c4-corr} list $(|\tau|, d(\tau), \sigma(\tau), c_{4, \tau})$ for every tree in Figures~\ref{fig:trees-q4-aff} and~\ref{fig:trees-q4-corr}, computed from~\eqref{eq:cq-formula} at $q = 4$.

\begin{table}[ht]
\centering
\renewcommand{\arraystretch}{1.05}
\begin{tabular}{|c|c|c|c|c|}
\hline
$\tau^{\rm aff}_i$ & $|\tau|$ & $d$ & $\sigma$ & $c_{4, \tau^{\rm aff}_i}$ \\
\hline
$\tau^{\rm aff}_1$ & 0 & 0 & 1 & $1$ \\
$\tau^{\rm aff}_2$ & 1 & 1 & 1 & $-4$ \\
$\tau^{\rm aff}_3$ & 2 & 2 & 1 & $6$ \\
$\tau^{\rm aff}_4$ & 3 & 3 & 1 & $-4$ \\
$\tau^{\rm aff}_5$ & 4 & 4 & 1 & $1$ \\
$\tau^{\rm aff}_6$ & 2 & 2 & 2 & $3$ \\
$\tau^{\rm aff}_7$ & 3 & 3 & 1 & $-4$ \\
$\tau^{\rm aff}_8$ & 4 & 4 & 1 & $1$ \\
$\tau^{\rm aff}_9$ & 4 & 4 & 2 & $1/2$ \\
$\tau^{\rm aff}_{10}$ & 3 & 3 & 6 & $-2/3$ \\
$\tau^{\rm aff}_{11}$ & 4 & 4 & 2 & $1/2$ \\
$\tau^{\rm aff}_{12}$ & 4 & 4 & 24 & $1/24$ \\
\hline
\end{tabular}
\caption{Closed-form coefficients $c_{4, \tau}$ for the 12 trees in $\mathcal{T}_4$ with no branches.}
\label{tab:c4-aff}
\end{table}

\begin{table}[ht]
\centering
\renewcommand{\arraystretch}{1.05}
\begin{tabular}{|c|c|c|c|c||c|c|c|c|c|}
\hline
$\tau^{\rm corr}_i$ & $|\tau|$ & $d$ & $\sigma$ & $c_{4, \tau^{\rm corr}_i}$
& $\tau^{\rm corr}_i$ & $|\tau|$ & $d$ & $\sigma$ & $c_{4, \tau^{\rm corr}_i}$ \\
\hline
$\tau^{\rm corr}_1$ & 3 & 2 & 2 & $-5$
& $\tau^{\rm corr}_{15}$ & 7 & 4 & 8 & $-1/8$ \\
$\tau^{\rm corr}_2$ & 4 & 3 & 6 & $5/6$
& $\tau^{\rm corr}_{16}$ & 5 & 4 & 6 & $-1/6$ \\
$\tau^{\rm corr}_3$ & 4 & 3 & 2 & $5/2$
& $\tau^{\rm corr}_{17}$ & 5 & 4 & 2 & $-1/2$ \\
$\tau^{\rm corr}_4$ & 5 & 3 & 2 & $-3$
& $\tau^{\rm corr}_{18}$ & 5 & 4 & 2 & $-1/2$ \\
$\tau^{\rm corr}_5$ & 4 & 3 & 2 & $5/2$
& $\tau^{\rm corr}_{19}$ & 5 & 4 & 2 & $-1/2$ \\
$\tau^{\rm corr}_6$ & 4 & 3 & 1 & $5$
& $\tau^{\rm corr}_{20}$ & 5 & 4 & 1 & $-1$ \\
$\tau^{\rm corr}_7$ & 5 & 4 & 24 & $-1/24$
& $\tau^{\rm corr}_{21}$ & 6 & 4 & 2 & $1/2$ \\
$\tau^{\rm corr}_8$ & 5 & 4 & 6 & $-1/6$
& $\tau^{\rm corr}_{22}$ & 6 & 4 & 2 & $1/2$ \\
$\tau^{\rm corr}_9$ & 5 & 4 & 4 & $-1/4$
& $\tau^{\rm corr}_{23}$ & 6 & 4 & 2 & $1/2$ \\
$\tau^{\rm corr}_{10}$ & 6 & 4 & 6 & $1/6$
& $\tau^{\rm corr}_{24}$ & 6 & 4 & 1 & $1$ \\
$\tau^{\rm corr}_{11}$ & 6 & 4 & 4 & $1/4$
& $\tau^{\rm corr}_{25}$ & 5 & 4 & 2 & $-1/2$ \\
$\tau^{\rm corr}_{12}$ & 6 & 4 & 2 & $1/2$
& $\tau^{\rm corr}_{26}$ & 5 & 4 & 1 & $-1$ \\
$\tau^{\rm corr}_{13}$ & 6 & 4 & 8 & $1/8$
& $\tau^{\rm corr}_{27}$ & 5 & 4 & 1 & $-1$ \\
$\tau^{\rm corr}_{14}$ & 7 & 4 & 2 & $-1/2$
& $\tau^{\rm corr}_{28}$ & 5 & 4 & 2 & $-1/2$ \\
\hline
\end{tabular}
\caption{Closed-form coefficients $c_{4, \tau}$ for the 28 trees in $\mathcal{T}_4$ with at least one branch. Entries are computed from~\eqref{eq:cq-formula} at $q = 4$.}
\label{tab:c4-corr}
\end{table}

\begin{remark}\label{rem:affine-fits-formula}
The 12 affine coefficients of Table~\ref{tab:c4-aff} are the binomial weights of the affine score~\eqref{eq:score} at $q = 4$. This is the special case where every non-root node is a leaf or chain (so $|\tau| = d(\tau)$). The closed form~\eqref{eq:cq-formula} reduces to $c_{q,\tau} = (-1)^{d(\tau)} \binom{q}{d(\tau)}/\sigma(\tau)$, which after grouping the ordered chain-length tuples $(k_1, \ldots, k_r)$ in~\eqref{eq:score} into tree classes recovers the affine binomial weights $(-1)^K \binom{q}{K}/\sigma(\tau)$, with $K = k_1 + \cdots + k_r = d(\tau)$.
\end{remark}

\subsubsection*{Affine moment function}

Taking the coefficients from Table~\ref{tab:c4-aff} and the terms from Figure~\ref{fig:trees-q4-aff} via the construction rule of Section~\ref{sec:trees},
\begin{align*}
\psi^{(4)}_{\rm aff}
&= m - 4\, m_\eta' \,\Lambda\, g(W_2) \\
&\quad + 6\, m_\eta' \,\big(\Lambda \,\partial_\eta g(W_3)\big)\, \Lambda\, g(W_2) \\
&\quad - 4\, m_\eta' \,\big(\Lambda \,\partial_\eta g(W_3)\big) \big(\Lambda \,\partial_\eta g(W_4)\big)\, \Lambda\, g(W_2) \\
&\quad + m_\eta' \,\big(\Lambda \,\partial_\eta g(W_3)\big) \big(\Lambda \,\partial_\eta g(W_4)\big) \big(\Lambda \,\partial_\eta g(W_5)\big)\, \Lambda\, g(W_2) \\
&\quad + 3\, \big(\Lambda\, g(W_2)\big)'\, m_{\eta\eta} \,\big(\Lambda\, g(W_3)\big) \\
&\quad - 4\, \big(\Lambda\, g(W_2)\big)'\, m_{\eta\eta} \,\big(\Lambda \,\partial_\eta g(W_4)\big)\, \Lambda\, g(W_3) \\
&\quad + \big(\Lambda\, g(W_2)\big)'\, m_{\eta\eta} \,\big(\Lambda \,\partial_\eta g(W_4)\big) \big(\Lambda \,\partial_\eta g(W_5)\big)\, \Lambda\, g(W_3) \\
&\quad + \tfrac{1}{2}\, \Big(\big(\Lambda \,\partial_\eta g(W_3)\big)\, \Lambda\, g(W_2)\Big)'\, m_{\eta\eta} \,\Big(\big(\Lambda \,\partial_\eta g(W_5)\big)\, \Lambda\, g(W_4)\Big) \\
&\quad - \tfrac{2}{3}\, m_{\eta\eta\eta}\big[\Lambda\, g(W_2),\, \Lambda\, g(W_3),\, \Lambda\, g(W_4)\big] \\
&\quad + \tfrac{1}{2}\, m_{\eta\eta\eta}\big[\Lambda\, g(W_2),\, \Lambda\, g(W_3),\, \big(\Lambda \,\partial_\eta g(W_5)\big) \Lambda\, g(W_4)\big] \\
&\quad + \tfrac{1}{24}\, m_{\eta\eta\eta\eta}\big[\Lambda\, g(W_2),\, \Lambda\, g(W_3),\, \Lambda\, g(W_4),\, \Lambda\, g(W_5)\big].
\end{align*}
The 12 terms correspond to $\tau^{\rm aff}_1, \tau^{\rm aff}_2, \tau^{\rm aff}_3, \tau^{\rm aff}_4, \tau^{\rm aff}_5, \tau^{\rm aff}_6, \tau^{\rm aff}_7, \tau^{\rm aff}_8, \tau^{\rm aff}_9, \tau^{\rm aff}_{10}, \tau^{\rm aff}_{11}, \tau^{\rm aff}_{12}$ in that order. The construction uses $L = 5$ independent copies $W_1, \ldots, W_5$.

\subsubsection*{Full moment function}

Taking the 28 correction-tree coefficients from Table~\ref{tab:c4-corr} and the terms from Figure~\ref{fig:trees-q4-corr} via the construction rule of Section~\ref{sec:trees},
\begin{align*}
\psi^{(4)}
&= \psi^{(4)}_{\rm aff} \\
&\quad - 5\, m_\eta'\, \Lambda\, \partial_\eta^2 g(W_2)\big[\Lambda g(W_3),\, \Lambda g(W_4)\big] \\
&\quad + \tfrac{5}{6}\, m_\eta'\, \Lambda\, \partial_\eta^3 g(W_2)\big[\Lambda g(W_3),\, \Lambda g(W_4),\, \Lambda g(W_5)\big] \\
&\quad + \tfrac{5}{2}\, \partial_\eta^2 m(W_1)\big[\Lambda g(W_2),\, \Lambda\, \partial_\eta^2 g(W_3)[\Lambda g(W_4),\, \Lambda g(W_5)]\big] \\
&\quad - 3\, m_\eta'\, \Lambda\, \partial_\eta^2 g(W_2)\big[\Lambda g(W_3),\, \Lambda\, \partial_\eta^2 g(W_4)[\Lambda g(W_5),\, \Lambda g(W_6)]\big] \\
&\quad + \tfrac{5}{2}\, m_\eta'\, \Lambda\, \partial_\eta g(W_2)\, \Lambda\, \partial_\eta^2 g(W_3)\big[\Lambda g(W_4),\, \Lambda g(W_5)\big] \\
&\quad + 5\, m_\eta'\, \Lambda\, \partial_\eta^2 g(W_2)\big[\Lambda\, \partial_\eta g(W_3)\, \Lambda g(W_4),\, \Lambda g(W_5)\big] \\
&\quad - \tfrac{1}{24}\, m_\eta'\, \Lambda\, \partial_\eta^4 g(W_2)\big[\Lambda g(W_3),\, \Lambda g(W_4),\, \Lambda g(W_5),\, \Lambda g(W_6)\big] \\
&\quad - \tfrac{1}{6}\, \partial_\eta^2 m(W_1)\big[\Lambda g(W_2),\, \Lambda\, \partial_\eta^3 g(W_3)[\Lambda g(W_4),\, \Lambda g(W_5),\, \Lambda g(W_6)]\big] \\
&\quad - \tfrac{1}{4}\, \partial_\eta^3 m(W_1)\big[\Lambda g(W_2),\, \Lambda g(W_3),\, \Lambda\, \partial_\eta^2 g(W_4)[\Lambda g(W_5),\, \Lambda g(W_6)]\big] \\
&\quad + \tfrac{1}{6}\, m_\eta'\, \Lambda\, \partial_\eta^2 g(W_2)\big[\Lambda g(W_3),\, \Lambda\, \partial_\eta^3 g(W_4)[\Lambda g(W_5),\, \Lambda g(W_6),\, \Lambda g(W_7)]\big] \\
&\quad + \tfrac{1}{4}\, m_\eta'\, \Lambda\, \partial_\eta^3 g(W_2)\big[\Lambda g(W_3),\, \Lambda g(W_4),\, \Lambda\, \partial_\eta^2 g(W_5)[\Lambda g(W_6),\, \Lambda g(W_7)]\big] \\
&\quad + \tfrac{1}{2}\, \partial_\eta^2 m(W_1)\big[\Lambda g(W_2),\, \Lambda\, \partial_\eta^2 g(W_3)[\Lambda g(W_4),\, \Lambda\, \partial_\eta^2 g(W_5)[\Lambda g(W_6),\, \Lambda g(W_7)]]\big] \\
&\quad + \tfrac{1}{8}\, \partial_\eta^2 m(W_1)\big[\Lambda\, \partial_\eta^2 g(W_2)[\Lambda g(W_3),\, \Lambda g(W_4)],\, \Lambda\, \partial_\eta^2 g(W_5)[\Lambda g(W_6),\, \Lambda g(W_7)]\big] \\
&\quad - \tfrac{1}{2}\, m_\eta'\, \Lambda\, \partial_\eta^2 g(W_2)\big[\Lambda g(W_3),\, \Lambda\, \partial_\eta^2 g(W_4)[\Lambda g(W_5),\, \Lambda\, \partial_\eta^2 g(W_6)[\Lambda g(W_7),\, \Lambda g(W_8)]]\big] \\
&\quad - \tfrac{1}{8}\, m_\eta'\, \Lambda\, \partial_\eta^2 g(W_2)\big[\Lambda\, \partial_\eta^2 g(W_3)[\Lambda g(W_5),\, \Lambda g(W_6)],\, \Lambda\, \partial_\eta^2 g(W_4)[\Lambda g(W_7),\, \Lambda g(W_8)]\big] \\
&\quad - \tfrac{1}{6}\, m_\eta'\, \Lambda\, \partial_\eta g(W_2)\, \Lambda\, \partial_\eta^3 g(W_3)\big[\Lambda g(W_4),\, \Lambda g(W_5),\, \Lambda g(W_6)\big] \\
&\quad - \tfrac{1}{2}\, m_\eta'\, \Lambda\, \partial_\eta^3 g(W_2)\big[\Lambda g(W_3),\, \Lambda g(W_4),\, \Lambda\, \partial_\eta g(W_5)\, \Lambda g(W_6)\big] \\
&\quad - \tfrac{1}{2}\, \partial_\eta^2 m(W_1)\big[\Lambda\, \partial_\eta g(W_2)\, \Lambda g(W_3),\, \Lambda\, \partial_\eta^2 g(W_4)[\Lambda g(W_5),\, \Lambda g(W_6)]\big] \\
&\quad - \tfrac{1}{2}\, \partial_\eta^2 m(W_1)\big[\Lambda g(W_2),\, \Lambda\, \partial_\eta g(W_3)\, \Lambda\, \partial_\eta^2 g(W_4)[\Lambda g(W_5),\, \Lambda g(W_6)]\big] \\
&\quad - \partial_\eta^2 m(W_1)\big[\Lambda g(W_2),\, \Lambda\, \partial_\eta^2 g(W_3)[\Lambda g(W_4),\, \Lambda\, \partial_\eta g(W_5)\, \Lambda g(W_6)]\big] \\
&\quad + \tfrac{1}{2}\, m_\eta'\, \Lambda\, \partial_\eta g(W_2)\, \Lambda\, \partial_\eta^2 g(W_3)\big[\Lambda g(W_4),\, \Lambda\, \partial_\eta^2 g(W_5)[\Lambda g(W_6),\, \Lambda g(W_7)]\big] \\
&\quad + \tfrac{1}{2}\, m_\eta'\, \Lambda\, \partial_\eta^2 g(W_2)\big[\Lambda\, \partial_\eta g(W_3)\, \Lambda g(W_4),\, \Lambda\, \partial_\eta^2 g(W_5)[\Lambda g(W_6),\, \Lambda g(W_7)]\big] \\
&\quad + \tfrac{1}{2}\, m_\eta'\, \Lambda\, \partial_\eta^2 g(W_2)\big[\Lambda g(W_3),\, \Lambda\, \partial_\eta g(W_4)\, \Lambda\, \partial_\eta^2 g(W_5)[\Lambda g(W_6),\, \Lambda g(W_7)]\big] \\
&\quad + m_\eta'\, \Lambda\, \partial_\eta^2 g(W_2)\big[\Lambda g(W_3),\, \Lambda\, \partial_\eta^2 g(W_4)[\Lambda g(W_5),\, \Lambda\, \partial_\eta g(W_6)\, \Lambda g(W_7)]\big] \\
&\quad - \tfrac{1}{2}\, m_\eta'\, \Lambda\, \partial_\eta g(W_2)\, \Lambda\, \partial_\eta g(W_3)\, \Lambda\, \partial_\eta^2 g(W_4)\big[\Lambda g(W_5),\, \Lambda g(W_6)\big] \\
&\quad - m_\eta'\, \Lambda\, \partial_\eta^2 g(W_2)\big[\Lambda g(W_3),\, \Lambda\, \partial_\eta g(W_4)\, \Lambda\, \partial_\eta g(W_5)\, \Lambda g(W_6)\big] \\
&\quad - m_\eta'\, \Lambda\, \partial_\eta g(W_2)\, \Lambda\, \partial_\eta^2 g(W_3)\big[\Lambda g(W_4),\, \Lambda\, \partial_\eta g(W_5)\, \Lambda g(W_6)\big] \\
&\quad - \tfrac{1}{2}\, m_\eta'\, \Lambda\, \partial_\eta^2 g(W_2)\big[\Lambda\, \partial_\eta g(W_3)\, \Lambda g(W_4),\, \Lambda\, \partial_\eta g(W_5)\, \Lambda g(W_6)\big].
\end{align*}
The 28 correction terms correspond to $\tau^{\rm corr}_1, \ldots, \tau^{\rm corr}_{28}$ in that order. The construction uses $L = 8$ independent copies of the data in the worst case ($\tau^{\rm corr}_{14}$ and $\tau^{\rm corr}_{15}$, each with $|\tau| = 7$ non-root nodes plus the root copy $W_1$).

\subsection{Asymptotic theory\label{sec:asymptotic}}

We now establish the asymptotic distribution of the estimator based on a $q$-th order Neyman-orthogonal moment function. We consider a grouped design with $N$ units, $T$ within-unit observations, and a unit-specific nuisance parameter, which covers the heterogeneous coefficients model of Section~\ref{sec:example}. We state the result for a $d_\theta$-dimensional target parameter $\theta_0$, with an orthogonal moment function $\psi$ taking values in $\mathbb{R}^{d_\psi}$ for some $d_\psi \geq d_\theta$. Under the assumptions below, the estimator is $\sqrt{N}$-consistent and asymptotically normal, where $N$ denotes the number of units (groups), even when the preliminary nuisance estimators converge at a rate slower than $\sqrt{N}$.

For each unit $i\in\{1,\ldots,N\}$, the $T$ observations $W_{i1}, \ldots, W_{iT}$ form the sample on which the orthogonal moment function $\psi(W_i; \theta, \nu_i)$ is constructed, and preliminary estimation of the unit-specific nuisance $\nu_{i0} = (\eta_{i0}, \lambda_{i0})$ is carried out on a held-out subset of observations for the same unit. The cross-sectional dimension $i = 1, \ldots, N$ enters only at the final step, through the sample average that defines the estimator of $\theta_0$. As a consequence, the $\sqrt{N}$-rate and the asymptotic normality in Theorem~\ref{th:asymptotic} below are driven by a cross-sectional central limit theorem. The within-unit structure of $\psi$ affects only the asymptotic variance. The proof follows Section~6 of \cite{BonhommeJochmansWeidner2025}, except for Lemma~\ref{lem:bridge}. 

\subsubsection{Setup}

We observe an independent cross-section of $N$ units. Each unit contributes a bundle $W_i = (W_{i1}, \ldots, W_{iT})$ of $T$ observations. Both $N$ and $T$ grow as $N \to \infty$, and we suppress the dependence of $T$ on $N$ in the notation. For each unit $i$, let $\nu_{i0} = (\eta_{i0}, \lambda_{i0})$ denote the true value of the combined nuisance parameter, where $\lambda_{i0}$ is the finite-dimensional parameter encoding the unit-level Jacobian inverse $\Lambda_{i0}$ as in Section~\ref{sec:setup}. We write $d_{\nu,i}$ for the dimension of $\nu_{i0}$ and, for $\nu = (\eta,\lambda)$, we abbreviate $\|\nu - \nu_{i0}\| := \|\eta - \eta_{i0}\| + \|\lambda - \lambda_{i0}\|$, so that the nuisance error is measured as in Definition~\ref{def:orthogonality}.

For each unit $i$ we partition the within-unit index set $\{1,\ldots,T\}$ into non-random disjoint subsets $\mathcal{S}_1$ and $\mathcal{S}_2$ with $|\mathcal{S}_2| \geq L$, suppressing their dependence on $i$ in the notation. Let $\psi^{(q)}$ be the moment function of Theorem~\ref{th:main} or Theorem~\ref{th:nonlinear}, which uses $L$ independent copies of an observation, and put
\[
\psi(W_i; \theta, \nu) := \frac{(|\mathcal{S}_2| - L)!}{|\mathcal{S}_2|!}
\sum_{\substack{t_1, \ldots, t_L \in \mathcal{S}_2 \\ \text{all distinct}}}
\psi^{(q)}(W_{it_1}, \ldots, W_{it_L}; \theta, \nu),
\qquad
\Psi_i(\theta,\nu) := \mathbb{E}\!\left[\psi(W_i; \theta, \nu)\right].
\]
The components of $\psi$ are the orthogonal moment functions obtained by applying the construction to each component of $m$ separately, so $\psi$ takes values in $\mathbb{R}^{d_\psi}$ and $d_\psi \geq d_\theta$ is permitted, which allows the target moment system to be overidentified. For every fixed $\nu$, the function $\Psi_i(\theta,\nu)$ is the population moment of Definition~\ref{def:orthogonality} formed under unit $i$'s within-unit distribution, because the $W_{it}$ are independent and identically distributed within the unit and each summand uses $L$ distinct indices.

The target parameter $\theta_0 \in \mathbb{R}^{d_\theta}$ is identified by
\begin{equation}
\frac{1}{N} \sum_{i=1}^N \mathbb{E}\!\left[\psi(W_i; \theta_0, \nu_{i0})\right] = 0.
\label{eq:asymp-identification}
\end{equation}
In the heterogeneous coefficients model of Section~\ref{sec:example} the target $\theta_0 = (1/N) \sum_i m(\eta_{i0})$ is a sample-dependent parameter. We allow $\theta_0$ to depend on $N$, so this target is covered by the results below.

Let $\widehat\nu_i$ be a preliminary estimator of $\nu_{i0}$. We assume throughout that $\widehat\nu_i$ is constructed from a held-out subset of observations, independent of the observations used to evaluate $\psi$. Concretely, $\widehat\nu_i$ is computed from $\{W_{it} : t \in \mathcal{S}_1\}$ and $\psi$ is evaluated on $\{W_{it} : t \in \mathcal{S}_2\}$. The sample is split within each unit. Efficiency can be improved by repeating the construction over several choices of $\mathcal{S}_2$ and averaging the resulting moment functions. When $|\mathcal{S}_1|$ and $|\mathcal{S}_2|$ are both proportional to $T$, this can be done by swapping their roles through cross-fitting. The analysis below applies to a single split.\footnote{Note that this differs from the leave-out scheme used in the simulations in Section~\ref{sec:implementation}.}

The estimator $\widehat\theta$ is defined by the GMM problem
\begin{equation}
\widehat\theta = \argmin_{\theta \in \Theta} \left\| \frac{1}{N} \sum_{i=1}^N \psi(W_i; \theta, \widehat\nu_i) \right\|_\Omega^2,
\label{eq:asymp-estimator}
\end{equation}
where $\Theta \subseteq \mathbb{R}^{d_\theta}$ is the parameter space, $\Omega$ is a symmetric positive-definite $d_\psi \times d_\psi$ weight matrix, and $\|x\|_\Omega^2 = x'\Omega x$. 

\subsubsection{Assumptions and main result}

Throughout, $\delta > 0$ is a non-random radius, fixed jointly in the two assumptions below, and $\mathcal{U}_i(\delta)$ denotes the closed ball of points $(\theta,\nu)$ with $\|\theta - \theta_0\| + \|\nu - \nu_{i0}\| \leq \delta$.

\begin{assumption}
\label{ass:asymp-reg}
\phantom{a}
\begin{enumerate}[(i)]
    \item $\Pr\bigl\{ (\widehat\theta,\widehat\nu_i) \in \mathcal{U}_i(\delta) \text{ for all } i \leq N \bigr\} \to 1$ as $N \to \infty$, and $\Theta$ contains every $\theta$ with $\|\theta - \theta_0\| < \delta$, for every $N$.
    \item $\max_i d_{\nu,i} = O(1)$.
    \item For all $i$ and $N$, $\psi(W_i;\theta,\nu)$ is continuously differentiable in $(\theta,\nu)$ on $\mathcal{U}_i(\delta)$, with $\partial_\theta \psi$ also continuously differentiable in $(\theta,\nu)$ there, and $\|\psi\|$, $\|\partial_\theta \psi\|$, $\|\partial^2_\theta \psi\|$, $\|\partial_\nu \psi\|$ and $\|\partial_\theta\partial_\nu \psi\|$ are bounded there by $M^\psi_{N,i}(W_i)$, a function of $\{W_{it} : t \in \mathcal{S}_2\}$ with $\sup_N \max_{i} \mathbb{E}[M^\psi_{N,i}(W_i)^2] < \infty$.
    \item $\widehat\theta - \theta_0 = o_P(1)$, and
    \[
    \frac{1}{N} \sum_{i=1}^N \mathbb{E}\!\left[\|\widehat\nu_i - \nu_{i0}\|^{2(q+1)}\right] = o(N^{-1}).
    \]
    \item The probability limit
    \[
    G = \plim_{N \to \infty} \frac{1}{N} \sum_{i=1}^N \frac{\partial \psi(W_i; \theta_0, \nu_{i0})}{\partial \theta'}
    \]
    exists as a $d_\psi \times d_\theta$ matrix, and $G'\Omega G$ is nonsingular.
\end{enumerate}
\end{assumption}

Part (i) requires $\widehat\theta$ and all nuisance estimates to lie in a fixed neighborhood of the truth with probability approaching one. Part (ii) mirrors the bounded-nuisance-dimension condition of \cite{BonhommeJochmansWeidner2025}, and ensures that the constants in the primitive conditions for part (iii) given below do not depend on $i$. In the heterogeneous coefficients model, $d_{\nu,i} = d_\eta + d_\eta(d_\eta+1)/2$ when $\lambda_i$ collects the distinct entries of the symmetric matrix $\Lambda_i$. This dimension does not depend on $N$ and $T$. Part (iii) is a condition on the orthogonal moment function. Since $\psi^{(q)}$ is a finite sum of contractions of derivatives of $m$ and $g$ of order at most $q$ in $\eta$, premultiplied by $\Lambda(\lambda)$, it holds if $m$ and $g$ have $q+1$ derivatives in $\eta$ and two in $\theta$, including mixed derivatives, $\Lambda$ is bounded and continuously differentiable with bounded derivative, and the derivatives of $m$ and $g$ are bounded by an envelope with second moments bounded uniformly in $i$ and $N$; $M^\psi_{N,i}$ is then a multiple of the average over $L$-tuples in $\mathcal{S}_2$ of products of these envelopes, whose second moment is bounded because the $W_{it}$ are independent within a unit. Part (iv) is the key rate condition on the preliminary estimator. In the grouped data example, $\widehat\nu_i$ is an OLS estimator computed on $|\mathcal{S}_1|$ observations. Under conditions ensuring that the relevant moments of $(|\mathcal{S}_1|^{-1}\sum_{t \in \mathcal{S}_1} X_{it}X_{it}')^{-1}$, $X_{it}$ and $X_{it}U_{it}$ are bounded uniformly in $i$ and $N$, $\mathbb{E}[\|\widehat\nu_i - \nu_{i0}\|^{2(q+1)}] = O(T^{-(q+1)})$ if $|\mathcal{S}_1|/T$ tends to a non-zero constant, and the rate condition is satisfied whenever $N = o(T^{q+1})$. Higher-order orthogonality thus permits inference in grouped settings where the group sizes are small relative to the number of groups. Part (v) is a standard rank condition.

\begin{assumption}
\label{ass:asymp-ortho}
\phantom{a}
\begin{enumerate}[(i)]
    \item There is a non-random $C < \infty$ such that, for every $i$, every $N$, and every $\nu = (\eta,\lambda)$ with $\|\eta - \eta_{i0}\| + \|\lambda - \lambda_{i0}\| \leq \delta$,
    \[
    \bigl\| \Psi_i(\theta_0,\eta,\lambda) - \Psi_i(\theta_0,\nu_{i0}) \bigr\| \;\leq\; C \bigl( \|\eta - \eta_{i0}\| + \|\lambda - \lambda_{i0}\| \bigr)^{q+1}.
    \]
    \item $\widehat\nu = (\widehat\nu_1,\ldots,\widehat\nu_N)$ is independent of $\{W_{it} : t \in \mathcal{S}_2, \; i = 1,\ldots,N\}$.
    \item $W_1, \ldots, W_N$ are independent across $i$, and $W_{i1},\ldots,W_{iT}$ are independent and identically distributed within each unit $i$.
    \item With $\xi_{N,i} = G'\Omega\,\psi(W_i; \theta_0, \nu_{i0}) \in \mathbb{R}^{d_\theta}$, the Lindeberg condition holds for $\xi_{N,i} - \mathbb{E}[\xi_{N,i}]$, and the limit
    \[
    V = \lim_{N \to \infty} \frac{1}{N} \sum_{i=1}^N \mathrm{Var}(\xi_{N,i})
    \]
    exists as a $d_\theta \times d_\theta$ positive definite matrix.
\end{enumerate}
\end{assumption}

Part (i) is the orthogonality bound of Definition~\ref{def:orthogonality} for unit $i$, required to hold uniformly. Definition~\ref{def:orthogonality} is stated for a scalar moment function, and part (i) applies it to each of the $d_\psi$ components and states the bound for the Euclidean norm. Part (ii) is achieved by sample splitting within each unit, as described above. Part (iii) requires independence across units. 

Part (iv) restricts the size of $\mathcal{S}_2$. The observations in $\mathcal{S}_2$ are independent and identically distributed, and $\psi$ averages $\psi^{(q)}$ over $L$-tuples of them, so $\psi(W_i; \theta_0, \nu_{i0})$ is a U-statistic of order $L$ in $|\mathcal{S}_2|$ observations. Hence $\xi_{N,i}$ is the average of $G'\Omega\, \psi^{(q)}$ over the same $L$-tuples, and $\mathrm{Var}(\xi_{N,i})$ is at most $L / |\mathcal{S}_2|$ times the second moment of $G'\Omega\, \psi^{(q)}$ at a single $L$-tuple of distinct observations. Under the primitive conditions given after Assumption~\ref{ass:asymp-reg}, that second moment is bounded uniformly in $i$ and $N$. When $|\mathcal{S}_2|$ is the same for all units, a positive definite $V$ thus requires $|\mathcal{S}_2|$ to stay bounded, for example $|\mathcal{S}_2| = L$. When this common $|\mathcal{S}_2|$ grows, the same steps go through with $\sqrt N$ replaced by $\sqrt{N |\mathcal{S}_2|}$, provided Assumption~\ref{ass:asymp-reg}(iii) also holds for $\psi^{(q)}$ at each $L$-tuple of distinct observations. The matrix $V$ becomes the limit of $N^{-1} \sum_{i=1}^N |\mathcal{S}_2|\, \mathrm{Var}(\xi_{N,i})$, and the Lindeberg condition is imposed on $\sqrt{|\mathcal{S}_2|}\,\xi_{N,i}$. The variance bounds in the proofs in Appendix~\ref{app:asymptotic-proof} must then also use the factor $L/|\mathcal{S}_2|$. The rate condition in Assumption~\ref{ass:asymp-reg}(iv) becomes $N^{-1} \sum_{i=1}^N \mathbb{E}[\|\widehat\nu_i - \nu_{i0}\|^{2(q+1)}] = o\bigl((N |\mathcal{S}_2|)^{-1}\bigr)$. In the heterogeneous coefficients model with $|\mathcal{S}_1|$ and $|\mathcal{S}_2|$ both proportional to $T$, that condition holds whenever $N = o(T^{q})$.

Write $\widehat\rho_i := \| \widehat\nu_i - \nu_{i0} \|$. On the event that $\widehat\rho_i \leq \delta$ for every $i$, all terms below are finite by Assumption~\ref{ass:asymp-reg}(iii), and replacing $\nu_{i0}$ by $\widehat\nu_i$ in the moment function gives
\begin{equation}
\frac{1}{\sqrt N}\sum_{i=1}^N \psi(W_i;\theta_0,\widehat\nu_i)
\;=\; \frac{1}{\sqrt N}\sum_{i=1}^N \psi(W_i;\theta_0,\nu_{i0}) \;+\; \mathcal{B}_N \;+\; \mathcal{E}_N ,
\label{eq:asymp-decomp}
\end{equation}
\begin{align*}
\mathcal{B}_N &:= \frac{1}{\sqrt N}\sum_{i=1}^N \bigl[ \Psi_i(\theta_0,\widehat\nu_i) - \Psi_i(\theta_0,\nu_{i0}) \bigr] , \\
\mathcal{E}_N &:= \frac{1}{\sqrt N}\sum_{i=1}^N \Bigl\{ \bigl[ \psi(W_i;\theta_0,\widehat\nu_i) - \Psi_i(\theta_0,\widehat\nu_i) \bigr] - \bigl[ \psi(W_i;\theta_0,\nu_{i0}) - \Psi_i(\theta_0,\nu_{i0}) \bigr] \Bigr\}.
\end{align*}
The event that $\widehat\rho_i \leq \delta$ for every $i$ has probability tending to one under Assumption~\ref{ass:asymp-reg}(iv), as shown in the proof of Lemma~\ref{lem:bridge}. The term $\mathcal{B}_N$ is the bias from evaluating the moment function at the preliminary estimates, and by Assumption~\ref{ass:asymp-ortho}(i) each of its summands is bounded in norm by $C \widehat\rho_i^{\,q+1}$. The term $\mathcal{E}_N$ is the change in the centered moment function. Both are negligible under our assumptions.

\begin{lemma}[Nuisance error under the orthogonality bound]
\label{lem:bridge}
Let Assumptions~\ref{ass:asymp-reg}(iii)--(iv) and~\ref{ass:asymp-ortho}(i)--(iii) hold with the same $q \in \{1,2,3,\ldots\}$. Then
\[
\frac{1}{\sqrt{N}} \sum_{i=1}^N \psi(W_i; \theta_0, \widehat\nu_i)
\;=\;
\frac{1}{\sqrt{N}} \sum_{i=1}^N \psi(W_i; \theta_0, \nu_{i0}) \;+\; o_P(1) .
\]
\end{lemma}

\begin{theorem}
\label{th:asymptotic}
Under Assumptions~\ref{ass:asymp-reg} and \ref{ass:asymp-ortho} and the identification condition~\eqref{eq:asymp-identification}, with the same $q \in \{1, 2, 3, \ldots\}$,
\[
\sqrt{N}\,(\widehat\theta - \theta_0) \xrightarrow{d} \mathcal{N}\!\left(0, \, (G'\Omega G)^{-1}\, V\, (G'\Omega G)^{-1}\right).
\]
\end{theorem}

The proofs of Lemma~\ref{lem:bridge} and Theorem~\ref{th:asymptotic} are in Appendix~\ref{app:asymptotic-proof}. The asymptotic variance depends on the orthogonality order $q$ through $\psi$, although we have left this dependence implicit in the notation. As in \cite{BonhommeJochmansWeidner2025}, the use of cross-fitting mitigates the variability introduced by a single split.

\section{Proofs\label{app:proofs}}

This appendix contains the proofs of the results stated in the main text and in Appendix~\ref{app:additional}.

\subsubsection*{Notation}

For integers $a \geq 0$ and $b \geq 0$ we write
\[
\binom{a}{b} \;:=\; \frac{a\,(a-1)\cdots(a-b+1)}{b!} .
\]
The numerator is the product of the $b$ factors $a, a-1, \ldots, a-b+1$. It is an empty product, equal to $1$, when $b = 0$, and it contains the factor $0$ when $a < b$. Hence $\binom{a}{0} = 1$, and $\binom{a}{b} = 0$ when $a < b$. For fixed $b$, the right-hand side of the display is a polynomial in $a$ of degree $b$.

The direct proof of Theorem~\ref{th:main} and the proof of Theorem~\ref{th:nonlinear} use power series to keep track of numerical coefficients. Only finitely many coefficients enter each argument, so convergence plays no role. For a formal power series $\varphi(y) = \sum_{n \geq 0} \varphi_n\, y^n$ we write $[y^n]\,\varphi(y) := \varphi_n$ for its coefficient of $y^n$, and we use the same notation for the variable $s$ in the proof of Theorem~\ref{th:nonlinear}. For example, the geometric series $1/(1-y) = \sum_{n \geq 0} y^n$ gives $[y^n]\,(1-y)^{-1} = 1$ for every $n \geq 0$.

\subsection{Proof of Proposition~\ref{prop:reconcile}\label{app:proof-prop1}}

Write $x = (\eta - \eta_0, \lambda - \lambda_0)$ and suppose $\Psi(\theta_0, \cdot)$ is $q$ times
differentiable at the truth. Let $B(x) = \sum_{k=0}^{q} B_k(x)$ be its Taylor polynomial of order
$q$ there, with $B_k$ homogeneous of degree $k$. Since $B_0 = \Psi(\theta_0,\eta_0,\lambda_0)$, the bound~\eqref{eq:ortho-bound} gives
$B(x) - B_0 = \Psi(\theta_0,\eta,\lambda) - \Psi(\theta_0,\eta_0,\lambda_0) + o(\|x\|^q) = O(\|x\|^{q+1}) + o(\|x\|^q)$, so
$B(x) - B_0 = o(\|x\|^q)$. Evaluating at $x = z v$ for fixed $v$ and letting $z \to 0$ gives
$\sum_{k=1}^{q} z^k B_k(v) = o(z^q)$, so $B_k(v) = 0$ for every $v$ and $1 \leq k \leq q$. Since $B_k$ vanishes for every $v$, its coefficients vanish, and these coefficients are the derivatives in~\eqref{eq:ortho} with $|\alpha|+|\beta| = k$, each divided by $\alpha!\,\beta!$.

\subsection{Proof of Lemma~\ref{lem:expansion}\label{app:proof-expansion}}

Fix $\nu$ with $\|\nu - \nu_0\| \leq r$ and write $\widehat P := \widehat P(\nu)$ and
$P(s) := (1-s)\widehat P + s P_0$, which lies in $\mathcal{P}$ for $s \in [0,1]$ by convexity. The
map $s \mapsto T(P(s))$ is $q+1$ times continuously differentiable on $[0,1]$ by the smoothness
assumed of $T$, so Taylor's theorem with integral remainder at $s = 0$, evaluated at $s = 1$, gives
\[
T(P_0) - \Psi^{(q)}(\widehat P, P_0)
\;=\; \frac{1}{q!} \int_0^1 (1-s)^{q}\, \frac{d^{\,q+1}}{ds^{\,q+1}} T(P(s))\, ds .
\]
The derivative in the integrand equals $D^{q+1}T(P(s))[H,\ldots,H]$ with $H = P_0 - \widehat P$, and
is therefore bounded in absolute value by $M \|H\|_{\overline P}^{\,q+1}$, where $M$ is the bound
assumed on the derivative of order $q+1$. Since $\int_0^1 (1-s)^q ds = 1/(q+1)$, this
gives~\eqref{eq:boundP} with $C = M/(q+1)!$.

For~\eqref{eq:boundnu}, the condition~\eqref{eq:lipschitz} gives $\|P_0 - \widehat P(\nu)\|_{\overline P} \leq c_{\widehat P} \|\nu - \nu_0\|$. Substituting into~\eqref{eq:boundP} gives~\eqref{eq:boundnu} with the
constant $M c_{\widehat P}^{\,q+1}/(q+1)!$.

\subsection{Proof of Lemma~\ref{lem:projection}}\label{app:proof-projection}

Throughout, $\nu = (\eta, \lambda)$ lies in a neighborhood of $\nu_0$ small enough that the rank condition holds at $\nu$, that $\eta$ lies in a ball around $\eta_0$ inside $\mathcal{N}$, and that $\lambda$ lies in a ball around $\lambda_0$ on which $\Lambda(\cdot)$ is bounded by $b_\Lambda$ and Lipschitz with constant $c_\Lambda$, as required by Assumption~\ref{ass:nonlinear}(i). 

The blocks of $h_q(w; \nu)$ in~\eqref{eq:hq} are $g(w, \theta_0, \eta)$, $\mathrm{vec}[\Lambda(\lambda)\, \partial_\eta g(w, \theta_0, \eta)]$ and $\mathrm{vec}[\partial^p_\eta g(w, \theta_0, \eta)]$ for $2 \leq p \leq q$, respectively, so $$\| h_q(w; \nu) \| \leq C \sum_{p \leq q} M^g_p(w),$$ for a constant $C$ depending on $b_\Lambda$ and the dimensions, which is square integrable under $\overline P$. Expectations of $\partial^p_\eta g(W_i, \theta_0, \eta)$, $0 \leq p \leq q+1$, under $\overline P$, under $P_0$ and under any $P \in \mathcal{P}$ are therefore finite, by the square-integrable envelopes and the Cauchy--Schwarz inequality $\int |f|\, |u|\, d\overline P \leq \| f \|_{L^2(\overline P)} \| u \|_{L^2(\overline P)}$ for $u = dP/d\overline P$, and the same inequality permits differentiation under the integral sign in $\eta$, up to order $q+1$, under every such $P$.

For part (i), write $$H(w) := (1, h_q(w; \nu)')',\quad \Gamma := \mathbb{E}_{\overline P}[ H(W_i) H(W_i)' ].$$ For $P \in \mathcal{P}$ with density $u := dP/d\overline P$, the restrictions~\eqref{eq:Phat-moments} and~\eqref{eq:derivatives} together with $\int dP = 1$ read $\mathbb{E}_{\overline P}[ H u ] = (1, t_q(\nu)')'$, where $t_q(\nu)$ stacks $0 \in \mathbb{R}^{d_g}$, $\mathrm{vec}[I_{d_\eta}]$ and $\mathrm{vec}\bigl[ \mathbb{E}_{P_0}[\partial^p_\eta g(W_i, \theta_0, \eta)] \bigr]$ for $2 \leq p \leq q$.

Since $\mathbb{E}_{\overline P}[ H\, dP_0/d\overline P ] = (1, \mathbb{E}_{P_0}[h_q(W_i; \nu)]')'$, whose blocks are $\overline g(\eta)$, $\mathrm{vec}[\partial_\eta \overline G(\eta, \lambda)]$ and $\mathrm{vec}\bigl[ \mathbb{E}_{P_0}[\partial^p_\eta g(W_i, \theta_0, \eta)] \bigr]$ for $2 \leq p \leq q$, subtracting it from both sides of $\mathbb{E}_{\overline P}[ H u ] = (1, t_q(\nu)')'$ shows that a member $P$ of $\mathcal{P}$ is feasible if and only if
\begin{equation}
\mathbb{E}_{\overline P}\Bigl[ H \Bigl( \frac{dP}{d\overline P} - \frac{dP_0}{d\overline P} \Bigr) \Bigr] \;=\; \bigl( 0, \zeta_q(\nu)' \bigr)' ,
\label{eq:feasible}
\end{equation}
with $\zeta_q(\nu)$ as in~\eqref{eq:zetaq}. In the difference on the right-hand side, the first entry is zero, and the blocks of order $p \geq 2$ are zero because~\eqref{eq:derivatives} sets their targets equal to their values under $P_0$. The two remaining blocks are $-\overline g(\eta)$ and $\mathrm{vec}[I_{d_\eta} - \partial_\eta \overline G(\eta, \lambda)]$, which are the blocks of $\zeta_q(\nu)$. The constraints are finitely many linear functionals of $u$, each continuous on $L^2(\overline P)$ because $H \in L^2(\overline P)$, so the feasible set is a closed affine subspace of $L^2(\overline P)$. The criterion is the $L^2(\overline P)$ distance from $dP_0/d\overline P$, so if the feasible set is nonempty the minimizer is the orthogonal projection, which is unique and has the form $dP_0/d\overline P + \mu' H$ for some $\mu \in \mathbb{R}^{1 + D_q}$, $D_q$ the dimension of $h_q$.

A function of the form $dP_0/d\overline P + \mu' H$ satisfies~\eqref{eq:feasible} if and only if $\Gamma \mu = (0, \zeta_q(\nu)')'$, because $\mathbb{E}_{\overline P}[ H H' \mu ] = \Gamma \mu$. Partition $\mu = (\mu_0, \mu_1')'$, so that
\[
\Gamma \mu \;=\; \begin{pmatrix} \mu_0 + \overline h_q(\nu)' \mu_1 \\ \overline h_q(\nu)\, \mu_0 + \mathbb{E}_{\overline P}[ h_q h_q' ]\, \mu_1 \end{pmatrix}.
\]
The first row of $\Gamma \mu = (0, \zeta_q(\nu)')'$ gives $\mu_0 = - \overline h_q(\nu)' \mu_1$, and substituting into the remaining rows gives $\overline\Sigma_q(\nu)\, \mu_1 = \zeta_q(\nu)$. Hence the system has a solution, and the feasible set is nonempty, if and only if $\zeta_q(\nu) \in \mathrm{range}\, \overline\Sigma_q(\nu)$.

It remains to show that $\zeta_q(\nu) \in \mathrm{range}\,\overline\Sigma_q(\nu)$. Since $\overline\Sigma_q(\nu)$ is symmetric and positive semidefinite, its range is the orthogonal complement of its null space, and $c$ lies in the null space if and only if $c' \overline\Sigma_q(\nu) c = \mathrm{Var}_{\overline P}( c' h_q(W_i; \nu) ) = 0$, that is, if and only if $c' h_q(\cdot\,; \nu)$ is $\overline P$-almost surely constant. By the rank condition, every such $c$ has zero entries in the first two blocks, while $\zeta_q(\nu)$ is zero outside these blocks, so $c' \zeta_q(\nu) = 0$. Hence $\zeta_q(\nu)$ lies in the range of $\overline\Sigma_q(\nu)$.

Every solution has $\mu_1 \in \overline\Sigma_q(\nu)^{+} \zeta_q(\nu) + \operatorname{null} \overline\Sigma_q(\nu)$. By the previous paragraph, $c' [ h_q(\cdot\,; \nu) - \overline h_q(\nu) ] = 0$ $\overline P$-almost surely for $c$ in the null space, so, using $\mu_0 = - \overline h_q(\nu)' \mu_1$,
\[
\mu' H \;=\; \mu_1' \bigl[ h_q(\cdot\,; \nu) - \overline h_q(\nu) \bigr]
\;=\; \zeta_q(\nu)'\, \overline\Sigma_q(\nu)^{+} \bigl[ h_q(\cdot\,; \nu) - \overline h_q(\nu) \bigr]
\]
for every solution, which is the second term of the density in~\eqref{eq:projection}. Moreover, $\| \widehat P(\nu) - P_0 \|^2_{\overline P} = \mathbb{E}_{\overline P}[ (\mu' H)^2 ] = \mu_1' \overline\Sigma_q(\nu) \mu_1 = \zeta_q(\nu)' \overline\Sigma_q(\nu)^{+} \zeta_q(\nu)$, using $\overline\Sigma_q \overline\Sigma_q^{+} \zeta_q = \zeta_q$ for $\zeta_q$ in the range. At $\nu = \nu_0$, $\zeta_q(\nu_0) = 0$ and the density is $dP_0/d\overline P$.

For part (ii), with $\zeta_q$ in the range of $\overline\Sigma_q$, $\zeta_q' \overline\Sigma_q^{+} \zeta_q \leq \gamma_q(\nu)^{-1} \| \zeta_q(\nu) \|^2 \leq \gamma^{-1} \| \zeta_q(\nu) \|^2$, and what remains is a Lipschitz bound on $\zeta_q$. Its two nonzero blocks are $-\overline g(\eta)$ and $I_{d_\eta} - \Lambda(\lambda)\, \partial_\eta \overline g(\eta)$. The mean value inequality under the integral sign gives $\| \overline g(\eta) - \overline g(\eta') \| \leq \mathbb{E}_{P_0}[M^g_1] \| \eta - \eta' \|$ and $\| \partial_\eta \overline g(\eta) - \partial_\eta \overline g(\eta') \| \leq \mathbb{E}_{P_0}[M^g_2] \| \eta - \eta' \|$ on the ball around $\eta_0$, with $\mathbb{E}_{P_0}[M^g_p] \leq \| M^g_p \|_{L^2(\overline P)} \| dP_0/d\overline P \|_{L^2(\overline P)} < \infty$. The second inequality uses $M^g_2$, which is available because $q + 1 \geq 2$. The first block is therefore Lipschitz in $\eta$ and does not depend on $\lambda$. 

The same envelopes bound $\partial_\eta \overline g$ on that ball, and $\Lambda(\cdot)$ is bounded by $b_\Lambda$ and Lipschitz with constant $c_\Lambda$ on the ball around $\lambda_0$, so the second block is a product of two bounded Lipschitz functions of $\nu$ and is itself Lipschitz. Hence $\nu \mapsto \zeta_q(\nu)$ is Lipschitz with some constant $c_\zeta$, and since $\zeta_q(\nu_0) = 0$, $\| \zeta_q(\nu) \| \leq c_\zeta \| \nu - \nu_0 \|$. The first block vanishes at $\nu_0$ because $\overline g(\eta_0) = 0$, the second because $\Lambda_0 J_0 = I_{d_\eta}$, and it is through the second block that $\| \lambda - \lambda_0 \|$ enters the bound of Lemma~\ref{lem:expansion}.

   For part (iii), fix $r > 0$ such that $\{\widetilde\eta : \|\widetilde\eta - \eta\| \leq r\} \subset \mathcal{N}$ for every $\eta$ in the neighborhood, and write $$F(\widetilde\eta) := \overline G_{\widehat P(\nu)}(\widetilde\eta, \lambda) = \Lambda(\lambda) \int g(w, \theta_0, \widetilde\eta)\, d\widehat P(\nu)(w).$$ Differentiation under the integral sign gives $\partial^p_{\widetilde\eta} F(\widetilde\eta) = \Lambda(\lambda) \int \partial^p_\eta g(w, \theta_0, \widetilde\eta)\, d\widehat P(\nu)(w)$ for $p \leq q+1$. At $\widetilde\eta = \eta$, the first condition in~\eqref{eq:Phat-moments} gives $F(\eta) = 0$, the second gives $\partial_{\widetilde\eta} F(\eta) = I_{d_\eta}$, and~\eqref{eq:derivatives}, premultiplied by $\Lambda(\lambda)$, gives $\partial^p_{\widetilde\eta} F(\eta) = \Lambda(\lambda)\, \mathbb{E}_{P_0}[\partial^p_\eta g(W_i, \theta_0, \eta)] = \partial^p_\eta \overline G(\eta, \lambda)$ for $2 \leq p \leq q$. The function $\widetilde\eta \mapsto (\widetilde\eta - \eta) + \overline R_\nu(\widetilde\eta)$ of~\eqref{eq:Gplug} has value $0$, Jacobian $I_{d_\eta}$ and derivatives $\partial^p_\eta \overline G(\eta, \lambda)$, $2 \leq p \leq q$, at $\widetilde\eta = \eta$, so the difference $D(\widetilde\eta) := F(\widetilde\eta) - (\widetilde\eta - \eta) - \overline R_\nu(\widetilde\eta)$ has vanishing derivatives at $\eta$ up to order $q$, and its derivative of order $q+1$ is $\partial^{q+1}_{\widetilde\eta} F - \partial^{q+1}_\eta \overline G(\cdot\,, \lambda)$. 
   
   Taylor's theorem with integral remainder along the segment from $\eta$ to $\widetilde\eta$, with the norm taken inside the integral, gives for $\| \widetilde\eta - \eta \| \leq r$
\[
\bigl\| D(\widetilde\eta) \bigr\|
\;\leq\; \frac{\| \widetilde\eta - \eta \|^{q+1}}{(q+1)!} \Bigl( \sup_{\| e - \eta \| \leq r} \bigl\| \partial^{q+1}_{\widetilde\eta} F(e) \bigr\| + \sup_{\| e - \eta \| \leq r} \bigl\| \partial^{q+1}_\eta \overline G(e, \lambda) \bigr\| \Bigr) ,
\]
where the norm of a derivative of order $q+1$ is its operator norm on $(q+1)$-tuples of unit vectors.
By the Cauchy--Schwarz inequality and the envelope $M^g_p$, for every $p \leq q+1$ and every $\widetilde\eta$ with $\| \widetilde\eta - \eta \| \leq r$,
\begin{align*}
\bigl\| \partial^{p}_{\widetilde\eta} F(\widetilde\eta) \bigr\|
&\;\leq\; b_\Lambda\, \| M^g_{p} \|_{L^2(\overline P)}\, \Bigl\| \frac{d\widehat P(\nu)}{d\overline P} \Bigr\|_{L^2(\overline P)} , \\
\Bigl\| \frac{d\widehat P(\nu)}{d\overline P} \Bigr\|_{L^2(\overline P)}
&\;\leq\; \Bigl\| \frac{dP_0}{d\overline P} \Bigr\|_{L^2(\overline P)} + \gamma^{-1/2} c_\zeta \| \nu - \nu_0 \|
\end{align*}
by part~(ii), which is bounded on the neighborhood, and the same argument with $\widehat P(\nu)$ replaced by $P_0$ bounds the second supremum. Taking $C$ to be the larger of the bound on $\sup \| \partial^p_{\widetilde\eta} F \|$ over $p \leq q+1$ and $1/(q+1)!$ times the sum of the two bounds gives~\eqref{eq:jet-agree} and the bound on the derivatives.

For part (iv), if $\| h_q(\cdot\,; \nu) - \overline h_q(\nu) \| \leq B$ almost surely under $\overline P$ and $dP_0/d\overline P \geq b > 0$, then by~\eqref{eq:projection} and part~(ii), $d\widehat P(\nu)/d\overline P \geq b - \| \overline\Sigma_q(\nu)^{+} \zeta_q(\nu) \| B \geq b - \gamma^{-1} c_\zeta B \| \nu - \nu_0 \|$, which is positive once $\| \nu - \nu_0 \| < b\, \gamma / (c_\zeta B)$.

\subsection{Proof of~\eqref{eq:Psiq}\label{app:affine}}

Fix $(\eta,\lambda)$ near $(\eta_0,\lambda_0)$, and let $\overline G$, $\overline H$ and $\delta_j = (-1)^j \overline H^{j-1}\overline G$ be as in Section~\ref{sec:affine}. Write $\mu_r := \mathbb{E}[\partial^r_\eta m(W_1,\theta_0,\eta)]$ for $r = 0, \ldots, q$, which is finite by Assumption~\ref{ass:affine}(iv) for $\eta$ close enough to $\eta_0$.

Since $I_{d_\eta}$ and $I_{d_\eta} + \overline H$ commute, the binomial theorem gives
\[
\overline H^{j-1} = \bigl((I_{d_\eta} + \overline H) - I_{d_\eta}\bigr)^{j-1}
= \sum_{k=1}^{j} \binom{j-1}{k-1} (-1)^{j-k} (I_{d_\eta} + \overline H)^{k-1}.
\]
Multiplying by $(-1)^j \overline G$ on the right, and using $(-1)^j(-1)^{j-k} = (-1)^k$,
\begin{equation}
\delta_j = \sum_{k=1}^{j} (-1)^{k} \binom{j-1}{k-1}\, (I_{d_\eta} + \overline{H})^{k-1}\, \overline G .
\label{eq:delta-binomial}
\end{equation}

At $\theta = \theta_0$, we have $\mathbb{E}[\Lambda(\lambda)\,\partial_\eta g(W_i,\theta_0,\eta)] = I_{d_\eta} + \overline H$ and $\mathbb{E}[\Lambda(\lambda)\, g(W_i,\theta_0,\eta)] = \overline G$, where both expectations are finite by Assumption~\ref{ass:affine}(v), since $g$ is affine in $\eta$. The factor $Z_k(W^{(b)};\theta_0,\eta,\lambda)$ in~\eqref{eq:Zfactor} is a product of $k$ such matrices evaluated at the $k$ independent copies in $W^{(b)}$. The expectation of a product of independent integrable random matrices is the product of their expectations, so
\begin{equation}
\mathbb{E}\bigl[Z_k(W^{(b)};\theta_0,\eta,\lambda)\bigr] = (I_{d_\eta} + \overline H)^{k-1}\, \overline G .
\label{eq:EZ}
\end{equation}

Substituting~\eqref{eq:delta-binomial} and~\eqref{eq:EZ} into~\eqref{eq:Psi-delta} and using multilinearity of $\langle \mu_r, \cdot \rangle$,
\[
\Psi^{(q)}(\theta_0,\eta,\lambda) = \sum_{r=0}^{q} \frac{1}{r!}
\sum_{\substack{j_1, \ldots, j_r \geq 1 \\ j_1 + \cdots + j_r \leq q}}
\ \sum_{k_1 = 1}^{j_1} \cdots \sum_{k_r = 1}^{j_r}
\ \prod_{b=1}^r \Bigl[ (-1)^{k_b} \binom{j_b - 1}{k_b - 1} \Bigr]
\Bigl\langle \mu_r,\, \bigotimes_{b=1}^r \mathbb{E}\bigl[Z_{k_b}(W^{(b)};\theta_0,\eta,\lambda)\bigr] \Bigr\rangle .
\]
The index set $\{(j_b,k_b)_{b=1}^r : 1 \leq k_b \leq j_b,\ \sum_b j_b \leq q\}$ equals $\{(j_b,k_b)_{b=1}^r : k_b \geq 1,\ \sum_b k_b \leq q,\ j_b \geq k_b,\ \sum_b j_b \leq q\}$. Interchanging the order of summation therefore gives
\[
\Psi^{(q)}(\theta_0,\eta,\lambda) = \sum_{r=0}^{q} \frac{1}{r!}
\sum_{\substack{k_1, \ldots, k_r \geq 1 \\ k_1 + \cdots + k_r \leq q}}
(-1)^{\sum_{b} k_b}\, c(k_1,\ldots,k_r)
\Bigl\langle \mu_r,\, \bigotimes_{b=1}^r \mathbb{E}\bigl[Z_{k_b}(W^{(b)};\theta_0,\eta,\lambda)\bigr] \Bigr\rangle,
\]
where
\[
c(k_1,\ldots,k_r) := \sum_{\substack{j_1 \geq k_1, \ldots, j_r \geq k_r \\ j_1 + \cdots + j_r \leq q}} \ \prod_{b=1}^r \binom{j_b - 1}{k_b - 1} = \binom{q}{k_1 + \cdots + k_r}
\]
by Lemma~\ref{lem:binomial} below.

In each term, the copies $W_1, W^{(1)}, \ldots, W^{(r)}$ are disjoint and hence mutually independent. Since $\langle A, v^{(1)} \otimes \cdots \otimes v^{(r)} \rangle$ is multilinear in $(A, v^{(1)}, \ldots, v^{(r)})$, with $|\langle A, v^{(1)} \otimes \cdots \otimes v^{(r)} \rangle| \leq \|A\| \prod_b \|v^{(b)}\|$, and each factor is integrable, Fubini's theorem gives
\[
\Bigl\langle \mu_r,\, \bigotimes_{b=1}^r \mathbb{E}\bigl[Z_{k_b}(W^{(b)};\theta_0,\eta,\lambda)\bigr] \Bigr\rangle
= \mathbb{E}\Bigl[ \Bigl\langle \partial^r_\eta m(W_1,\theta_0,\eta),\, \bigotimes_{b=1}^r Z_{k_b}(W^{(b)};\theta_0,\eta,\lambda) \Bigr\rangle \Bigr].
\]
For $r = 0$, the term is $\mu_0 = \mathbb{E}[m(W_1,\theta_0,\eta)]$. Summing over $r$ and $(k_1,\ldots,k_r)$ and comparing with~\eqref{eq:score} yields $\Psi^{(q)}(\theta_0,\eta,\lambda) = \mathbb{E}[\psi^{(q)}(W;\theta_0,\eta,\lambda)]$, implying~\eqref{eq:Psiq}. \qed

\begin{lemma}
\label{lem:binomial}
For integers $q \geq 1$ and $k_1, \ldots, k_r \geq 1$ with $K := k_1 + \cdots + k_r \leq q$,
\begin{equation}
\sum_{\substack{j_1 \geq k_1, \ldots, j_r \geq k_r \\ j_1 + \cdots + j_r \leq q}}
\ \prod_{b=1}^r \binom{j_b - 1}{k_b - 1} \;=\; \binom{q}{K} .
\label{eq:binomial-identity}
\end{equation}
\end{lemma}

\begin{proof}
For $k \geq 1$ and $|x| < 1$, $\sum_{j \geq k} \binom{j-1}{k-1} x^j = x^k (1-x)^{-k}$. Taking the product over $b = 1, \ldots, r$, the coefficient of $x^J$ in $x^K (1-x)^{-K}$ equals the sum of $\prod_b \binom{j_b-1}{k_b-1}$ over all $j_b \geq k_b$ with $j_1 + \cdots + j_r = J$. The left-hand side of~\eqref{eq:binomial-identity} is the sum of these coefficients over $J \leq q$, which is the coefficient of $x^q$ in $x^K (1-x)^{-K} (1-x)^{-1} = x^K (1-x)^{-(K+1)}$. Since the coefficient of $x^n$ in $(1-x)^{-(K+1)}$ is $\binom{n+K}{K}$, the coefficient of $x^q$ in $x^K(1-x)^{-(K+1)}$ is $\binom{q}{K}$.
\end{proof}

\subsection{Proof of Theorem~\ref{th:main}\label{app:proof-thm1}}

Theorem~\ref{th:main} follows from Theorem~\ref{th:nonlinear}, applied to a function $g$ that is affine in $\eta$. This is not immediate, because the two theorems write $\psi^{(q)}$ in different forms. In~\eqref{eq:score} it is a sum over tuples $(k_1, \ldots, k_r)$ with binomial weights, and in~\eqref{eq:general-psi} it is a sum over rooted trees $\tau$ with the coefficients $c_{q,\tau}$. In Appendix~\ref{app:proof-thm1-via-thm2} we check that the assumptions of Theorem~\ref{th:nonlinear} hold and that the two forms have the same population moment. Appendix~\ref{app:proof-thm1-short} gives a direct proof of Theorem~\ref{th:main}, which uses neither Theorem~\ref{th:nonlinear} nor rooted trees.

\subsubsection{Proof via Theorem~\ref{th:nonlinear}\label{app:proof-thm1-via-thm2}}

Assumption~\ref{ass:nonlinear} holds under Assumption~\ref{ass:affine}. Its part~(i) is parts~(ii)--(iv) of Assumption~\ref{ass:affine}, and its part~(ii) holds because a function that is affine in $\eta$ is infinitely differentiable in $\eta$. For its part~(iii), suppressing $\theta_0$ from the notation, we have $\partial_\eta^p g(W_i, \eta) = 0$ for $p \geq 2$, so the zero function is an envelope at those orders, while $\partial_\eta g(W_i, \eta) = \partial_\eta g(W_i, \eta_0)$ does not depend on $\eta$ and
\[
\sup_{\eta \in \mathcal{N}} \| g(W_i, \eta) \|
\;\leq\; \| g(W_i, \eta_0) \| \;+\; \Bigl( \sup_{\eta \in \mathcal{N}} \| \eta - \eta_0 \| \Bigr)\, \| \partial_\eta g(W_i, \eta_0) \|
\]
on a bounded neighborhood $\mathcal{N}$ of $\eta_0$, so that Assumption~\ref{ass:affine}(v) gives integrable envelopes at $p = 0$ and $p = 1$. Assumptions~\ref{ass:iid} and~\ref{ass:jacobian-inverse} are imposed in both theorems, and $q \geq 1$, so Theorem~\ref{th:nonlinear} applies.

When $g$ is affine in $\eta$, the moment function~\eqref{eq:general-psi} of Theorem~\ref{th:nonlinear} and the moment function~\eqref{eq:score} of Theorem~\ref{th:main} have the same population moment. A non-root node of $\tau$ with $p \geq 2$ children is evaluated at $\Lambda(\lambda)\, \partial_\eta^p g(W_j, \theta, \eta) = 0$, so $\kappa_\tau$ vanishes identically unless every non-root node of $\tau$ has at most one child, that is unless $|\tau| = d(\tau)$. The number of children of the root is unrestricted, since the root is evaluated at $\partial_\eta^r m(W_1, \theta, \eta)$ and $m$ may be nonlinear in $\eta$. A tree with $|\tau| = d(\tau)$ is a root with $r \geq 0$ children, the $b$-th of them the top of a chain of $k_b \geq 1$ nodes ending in a leaf. The $k_b - 1$ one-child nodes of that chain are evaluated at $\Lambda(\lambda)\, \partial_\eta g$ and its leaf at $\Lambda(\lambda)\, g$, so contracting the chains bottom up gives
\[
\kappa_\tau(W, \theta, \eta, \lambda)
\;=\; \Bigl\langle \partial_\eta^r m(W_1, \theta, \eta),\; \bigotimes_{b=1}^{r} Z_{k_b}(W^{(b)}; \theta, \eta, \lambda) \Bigr\rangle
\]
with $Z_{k_b}$ as in~\eqref{eq:Zfactor}. Here $|\tau| = d(\tau) = K$ with $K := \sum_{b=1}^{r} k_b$, so $\tau \in \mathcal{T}_q$ if and only if $K \leq q$ by~\eqref{eq:Tq-def}, and~\eqref{eq:cq-formula} reduces to $c_{q,\tau} = (-1)^{K} \binom{q}{K} / \sigma(\tau)$.

An automorphism of $\tau$ fixes the root and maps each chain onto a chain of equal length. It maps the $j$-th node of a chain, counted from the root, to the $j$-th node of the image chain, so an automorphism is the same as a permutation of the chains that preserves their lengths. Writing $n_\ell$ for the number of chains of length $\ell$, we therefore have $\sigma(\tau) = \prod_\ell n_\ell!$ and
\begin{equation}
c_{q,\tau} \;=\; \frac{(-1)^{K}}{\prod_\ell n_\ell!}\, \binom{q}{K} ,
\label{eq:t1-caff}
\end{equation}
which coincides with~\eqref{eq:cq-formula-affine}. The ordered tuples $(k_1, \ldots, k_r)$ whose entries form the multiset of chain lengths of $\tau$ number $r! / \prod_\ell n_\ell!$, and each of them gives the contraction above. The $r$-linear form $\partial_\eta^r m$ is symmetric, so permuting a tuple only relabels the blocks $W^{(1)}, \ldots, W^{(r)}$ of independent copies of the data among the chains, which leaves $\Psi^{(q)}(\theta_0, \eta, \lambda)$ unchanged. Summing~\eqref{eq:t1-caff} over the trees with $|\tau| = d(\tau)$ and $d(\tau) \leq q$ is therefore the same as summing over $r$ and over ordered tuples with $k_b \geq 1$ and $K \leq q$, with the weight $(-1)^{K} \binom{q}{K} / \prod_\ell n_\ell!$ replaced by $(-1)^{K} \binom{q}{K} / r!$. These are the weights in~\eqref{eq:score}, with $K = \sum_{b=1}^{r} k_b$, and the constraint $k_b \geq 1$ makes $r \leq K \leq q$, which is the range of the outer sum of~\eqref{eq:score}. Hence $\sum_{\tau \in \mathcal{T}_q} c_{q,\tau}\, \kappa_\tau$ and $\psi^{(q)}$ of~\eqref{eq:score} have the same population moment $\Psi^{(q)}(\theta_0, \eta, \lambda)$.

The surviving trees have $|\tau| = K \leq q$ non-root nodes. One copy of the data is used at the root and one at each non-root node, so $\psi^{(q)}$ uses at most $q+1$ independent copies of the data, or $q$ when $m$ does not depend on the data, as in Theorem~\ref{th:main}. Theorem~\ref{th:nonlinear} therefore gives the bound~\eqref{eq:ortho-bound} of Definition~\ref{def:orthogonality} for the $\Psi^{(q)}$ of Theorem~\ref{th:main}.

\subsubsection{Direct proof\label{app:proof-thm1-short}}

Write $u := \eta - \eta_0$ and $\rho := \| \eta - \eta_0 \| + \| \lambda - \lambda_0 \|$, let $\overline m$ and $\overline g$ be the population functions at $P_0$,
\begin{equation}
\overline m(\eta) := \mathbb{E}\bigl[ m(W_i, \theta_0, \eta) \bigr] ,
\qquad
\overline g(\eta) := \mathbb{E}\bigl[ g(W_i, \theta_0, \eta) \bigr] ,
\label{eq:pop-functions}
\end{equation}
and set $M := \Lambda(\lambda)\, J_0$ and $E := I_{d_\eta} - M$. In the notation of Section~\ref{sec:affine}, $M = I_{d_\eta} + \overline H$, $E = -\overline H$ and $\overline G = M u$. Let $C$ be a finite constant, possibly different at each occurrence, that depends only on $q$, on $\| J_0 \|$, on $\mathbb{E}[M_r(W_i)]$ for $r \leq q+1$ and on a Lipschitz constant of $\Lambda(\cdot)$ at $\lambda_0$. We work on a neighborhood of $(\eta_0, \lambda_0)$ small enough that $\rho \leq 1$, that $\eta$ lies in a ball around $\eta_0$ contained in the neighborhood $\mathcal{N}$ of Assumption~\ref{ass:affine}(iv), and that $\Lambda(\cdot)$ is Lipschitz at $\lambda_0$ on it. Assumption~\ref{ass:jacobian-inverse} gives $\Lambda_0 J_0 = I_{d_\eta}$, hence $E = (\Lambda_0 - \Lambda(\lambda))\, J_0$ and
\begin{equation}
\| E \| \;\leq\; \| J_0 \|\, \| \Lambda(\lambda) - \Lambda_0 \| \;\leq\; C \rho ,
\qquad
\| M \| \;\leq\; 1 + \| E \| \;\leq\; C .
\label{eq:t1s-E}
\end{equation}
Only Lipschitz continuity of $\Lambda(\cdot)$ at $\lambda_0$ is used here.

By Assumption~\ref{ass:affine}(i) and~(v), together with $\overline g(\eta_0) = 0$ from~\eqref{eq:model-general}, we have $\overline g(\eta) = J_0 u$, and $\partial_\eta g(W_i, \theta_0, \eta)$ does not depend on $\eta$. The $k$ factors of~\eqref{eq:Zfactor} are evaluated at distinct independent copies of the data by Assumption~\ref{ass:iid}, so $\mathbb{E}[Z_k] = M^{k-1}\, \Lambda(\lambda)\, \overline g(\eta) = M^k u$. The copy $W_1$ at which $\partial_\eta^r m$ is evaluated is independent of those as well, and Assumption~\ref{ass:affine}(ii) and~(iv) allow differentiation under the expectation at that copy, so taking expectations in~\eqref{eq:score} gives
\begin{equation}
\Psi^{(q)}(\theta_0, \eta, \lambda)
\;=\; \sum_{r=0}^{q} \frac{1}{r!} \sum_{\substack{k_1, \ldots, k_r \geq 1 \\ K \leq q}} (-1)^{K} \binom{q}{K}\,
\Bigl\langle \partial_\eta^r \overline m(\eta),\, \bigotimes_{b=1}^{r} M^{k_b} u \Bigr\rangle ,
\qquad K := \sum_{b=1}^{r} k_b .
\label{eq:t1s-Psi}
\end{equation}

Let $Q$ be the Taylor polynomial of $\overline m$ at $\eta_0$ of degree $q$, and let $R := \overline m - Q$. By Assumption~\ref{ass:affine}(ii) and~(iv), $\overline m$ is $(q+1)$-times continuously differentiable on $\mathcal{N}$ with $\| \partial_\eta^r \overline m \| \leq \mathbb{E}[M_r(W_i)]$ there, and $\partial_\eta^r Q$ is the Taylor polynomial of $\partial_\eta^r \overline m$ at $\eta_0$ of degree $q - r$, so Taylor's theorem gives $\| \partial_\eta^r R(\eta) \| \leq C \| u \|^{q+1-r}$ for $0 \leq r \leq q$. Replacing $\overline m$ by $R$ in~\eqref{eq:t1s-Psi} and using~\eqref{eq:t1s-E} bounds each summand there by $C\, \| u \|^{q+1-r}\, \| u \|^{r} \leq C \rho^{\,q+1}$, and the number of summands depends only on $q$. It remains to bound~\eqref{eq:t1s-Psi} with $Q$ in place of $\overline m$.

We use a generating function in a formal variable $y$. Set $v(y) := \sum_{k \geq 1} y^k M^k u$ and $w(y) := u + v(y)$, and for a formal power series $\varphi(y)$ let
\[
S[\varphi] \;:=\; \sum_{n \geq 0} (-1)^n \binom{q}{n}\, [y^n]\, \varphi(y),
\]
a finite sum because $\binom{q}{n} = 0$ for $n > q$. Since $Q$ is a polynomial of degree at most $q$, its Taylor expansions at $\eta$ and at $\eta_0$ are exact. Evaluating $Q$ at $\eta + v(y) = \eta_0 + w(y)$ in both ways gives the identity between formal power series
\begin{equation}
\sum_{r=0}^{q} \frac{1}{r!} \bigl\langle \partial_\eta^r Q(\eta),\, v(y)^{\otimes r} \bigr\rangle
\;=\; \sum_{r=0}^{q} \frac{1}{r!} \bigl\langle \partial_\eta^r \overline m(\eta_0),\, w(y)^{\otimes r} \bigr\rangle ,
\label{eq:t1s-two-sides}
\end{equation}
where $\partial_\eta^r Q(\eta_0) = \partial_\eta^r \overline m(\eta_0)$ for $r \leq q$ by the definition of $Q$. We apply $S$ to both sides.

On the left, multilinearity gives $v(y)^{\otimes r} = \sum_{k_1, \ldots, k_r \geq 1} y^{K} \bigotimes_{b=1}^{r} M^{k_b} u$, so the coefficient of $y^n$ collects the tuples $(k_1, \ldots, k_r)$ with $K = n$. The operator $S$ multiplies this coefficient by $(-1)^n \binom{q}{n}$, which is zero for $n > q$. Hence $S$ gives each tuple the weight $(-1)^{K} \binom{q}{K}$ and removes the tuples with $K > q$, so $S$ applied to the left-hand side of~\eqref{eq:t1s-two-sides} is~\eqref{eq:t1s-Psi} with $Q$ in place of $\overline m$.

On the right, $w(y) = \sum_{k \geq 0} y^k M^k u = (I_{d_\eta} - y M)^{-1} u$. Since $I_{d_\eta} - y M = (1-y)\, I_{d_\eta} + y E$, expanding the inverse in powers of $E$ gives
\[
w(y) \;=\; \sum_{i \geq 0} (-1)^i\, \frac{y^i}{(1-y)^{i+1}}\; E^i u ,
\qquad
w(y)^{\otimes r} \;=\; \sum_{i_1, \ldots, i_r \geq 0} (-1)^{i}\, \frac{y^{i}}{(1-y)^{i+r}}\; E^{i_1} u \otimes \cdots \otimes E^{i_r} u ,
\]
with $i := \sum_{b=1}^{r} i_b$. It remains to apply $S$ to these terms.

For $r \geq 1$ and $n \geq 0$ we have $[y^n]\, y^{i} (1-y)^{-(i+r)} = \binom{n+r-1}{i+r-1}$, a polynomial in $n$ of degree $i + r - 1$, which vanishes when $i > n$. Writing $(\Delta p)(n) := p(n+1) - p(n)$, we have $\sum_{n \geq 0} (-1)^n \binom{q}{n} p(n) = (-1)^q\, (\Delta^q p)(0)$, and $\Delta$ lowers the degree of a polynomial by one, so this sum is zero for every polynomial $p$ of degree less than $q$. The first identity confines the indices to $i \leq n \leq q$, while $r \leq q$ because $Q$ has degree at most $q$, leaving a number of terms that depends only on $q$. The second identity makes $S$ annihilate the term carrying $E^{i_1} u \otimes \cdots \otimes E^{i_r} u$ whenever $i + r \leq q$.  The term with $r = 0$ is $\overline m(\eta_0) = \Psi^{(q)}(\theta_0, \eta_0, \lambda_0)$, which is subtracted in~\eqref{eq:ortho-bound}. Every remaining term has $i + r \geq q+1$ and a numerical factor determined by $q$, $r$ and $i$, so by~\eqref{eq:t1s-E} it is bounded in absolute value by
\[
C\, \bigl\| \partial_\eta^r \overline m(\eta_0) \bigr\|\; \| E \|^{i}\; \| u \|^{r}
\;\leq\; C \rho^{\,i+r} \;\leq\; C \rho^{\,q+1} .
\]
Together with the bound for $R$ this gives $\bigl| \Psi^{(q)}(\theta_0, \eta, \lambda)  - \Psi^{(q)}(\theta_0, \eta_0, \lambda_0) \bigr| \leq C \rho^{\,q+1}$, which is the bound~\eqref{eq:ortho-bound} of Definition~\ref{def:orthogonality}.

The constant $C$ and the neighborhood in this proof depend on the distribution of $W_i$ only through $\| J_0 \|$ and $\mathbb{E}[M_r(W_i)]$ for $r \leq q+1$, and on the parametrization only through a Lipschitz constant of $\Lambda(\cdot)$ at $\lambda_0$, together with the radii of the neighborhoods in Assumption~\ref{ass:affine}(iii) and~(iv). In Example~\ref{ex:grouped-q2}, $\Lambda(\cdot)$ is the identity, with Lipschitz constant one, and $\| J_0 \| = \mathbb{E}[X_{it}^2]$. The constant $C$ of Definition~\ref{def:orthogonality} and the radius of a ball around $(\eta_{i0}, \Lambda_{i0})$ that serves as its neighborhood $\mathcal{N}$ can therefore be chosen independently of the unit $i$ if the neighborhoods in Assumption~\ref{ass:affine}(iv) contain balls whose radius does not depend on $i$, and the envelopes of Assumption~\ref{ass:affine}(iv) and $\mathbb{E}[X_{it}^2]$ are bounded uniformly in $i$.

\subsection{Proof of Theorem~\ref{th:nonlinear}\label{app:proof-thm2}}

A roadmap for the proof is as follows. We first write $\Psi^{(q)}(\theta_0, \eta, \lambda)$ as a sum over the trees $\tau \in \mathcal{T}_q$ of contractions of derivatives of $\overline m$ and of $\Lambda(\lambda)\, \overline g$, in~\eqref{eq:Psi-is-treesum}. We then construct a path $\eta_*(s)$, $s \in [0,1]$, with $\eta_*(0) = \eta$ and $\eta_*(1) = \eta_0$, in Lemma~\ref{lem:segment}. We show in~\eqref{eq:Psi-extract} and~\eqref{eq:Fq-final} that the tree sum equals the Taylor polynomial of degree $q$ of $s \mapsto \overline m(\eta_*(s))$ at $s = 0$, evaluated at $s = 1$. Taylor's theorem then bounds the difference between this value and $\overline m(\eta_0) = \Psi^{(q)}(\theta_0, \eta_0, \lambda_0)$ by a constant times $(\| \eta - \eta_0 \| + \| \lambda - \lambda_0 \|)^{q+1}$.

In this proof we write $g(W_i, \eta)$ for $g(W_i, \theta_0, \eta)$ and $m(W_i, \eta)$ for $m(W_i, \theta_0, \eta)$. Recall the population functions~\eqref{eq:pop-functions}, in which $\overline m(\eta) = \mathbb{E}[m(W_i, \eta)]$ is the target, and set
\[
A(\eta, \lambda) := \partial_\eta \overline G(\eta, \lambda), \qquad
E(\eta, \lambda) := I_{d_\eta} - A(\eta, \lambda), \qquad
\overline G_p(\eta, \lambda) := \partial_\eta^p \overline G(\eta, \lambda) \;\; (p \geq 2) ,
\]
with $A$ and $\overline G_p$ equal to $\Lambda(\lambda)$ times the corresponding derivative of $\overline g$. Each $\overline G_p$ is a symmetric $p$-linear map $(\mathbb{R}^{d_\eta})^p \to \mathbb{R}^{d_\eta}$, and we write $\overline G_p[v_1, \ldots, v_p]$ for its value and $v^{\otimes p}$ for $p$ equal arguments, and the same bracket is used for the symmetric $r$-linear form $\partial_\eta^r \overline m(\eta)$. Write $\rho := \| \eta - \eta_0 \| + \| \lambda - \lambda_0 \|$. The moment function $g$ and the parameter $\lambda$ enter only through the map $\overline G(\cdot\,, \lambda) : \mathbb{R}^{d_\eta} \to \mathbb{R}^{d_\eta}$ and its derivatives, so neither $\Lambda(\lambda)$ nor $J_0$ needs to be invertible, and $d_g \geq d_\eta$ is allowed.

Let $\overline\kappa_\tau(\eta, \lambda)$ be the contraction attached to the rooted tree $\tau$ by the rules of Section~\ref{sec:trees}, with $\partial_\eta^r \overline m$ at a root with $r$ children, $\Lambda(\lambda)\, \overline g$ at each leaf, $A$ at each non-root node with one child and $\overline G_p$ at each non-root node with $p \geq 2$ children, all evaluated at $(\eta, \lambda)$. We have
\begin{equation}
\Psi^{(q)}(\theta_0, \eta, \lambda) \;=\; \sum_{\tau \in \mathcal{T}_q} c_{q,\tau}\, \overline\kappa_\tau(\eta, \lambda) .
\label{eq:Psi-is-treesum}
\end{equation}
Indeed, $\kappa_\tau(W, \theta_0, \eta, \lambda)$ is multilinear in the tuple of tensors at the nodes of $\tau$, one per node, and those tensors are evaluated at distinct copies of the data, hence independent by Assumption~\ref{ass:iid}. The expectation of $\kappa_\tau$ is therefore the same contraction of the expected tensors. Assumption~\ref{ass:nonlinear}(iii) allows expectation and differentiation in $\eta$ to be interchanged up to order $q+1$, giving $\mathbb{E}[\Lambda(\lambda)\, \partial_\eta^p g(W_j, \eta)] = \Lambda(\lambda)\, \partial_\eta^p \overline g(\eta)$ at the non-root nodes, and Assumption~\ref{ass:affine}(ii) and~(iv), inherited through Assumption~\ref{ass:nonlinear}(i), allow the same interchange for $m$ at the root, giving $\mathbb{E}[\partial_\eta^r m(W_1, \eta)] = \partial_\eta^r \overline m(\eta)$. Summing over $\tau$ gives~\eqref{eq:Psi-is-treesum}.

We use throughout that $\overline g(\eta_0) = 0$ by~\eqref{eq:model-general}, that $\partial_\eta \overline g(\eta_0) = J_0$ has full column rank with $\Lambda_0 J_0 = I_{d_\eta}$ by Assumption~\ref{ass:jacobian-inverse}, that $\overline g$ and $\overline m$ are $(q+1)$-times continuously differentiable with derivatives bounded by $\mathbb{E}[M^g_p(W_i)]$ and $\mathbb{E}[M_r(W_i)]$ on a neighborhood $\mathcal{N}$ of $\eta_0$ by Assumption~\ref{ass:nonlinear}(ii)--(iii) and Assumption~\ref{ass:affine}(ii) and~(iv), and that $\Lambda(\cdot)$ is bounded and Lipschitz near $\lambda_0$ by Assumption~\ref{ass:affine}(iii). Hence $$\overline G(\eta, \lambda) = \Lambda(\lambda)\, [\overline g(\eta) - \overline g(\eta_0)]=O(\rho),\quad E(\eta, \lambda) = \Lambda(\lambda_0) J_0 - \Lambda(\lambda)\, \partial_\eta \overline g(\eta)=O(\rho),$$ because $\Lambda(\cdot)$ is Lipschitz at $\lambda_0$ and $\partial_\eta \overline g$ is Lipschitz at $\eta_0$.

Fix $\nu = (\eta, \lambda)$ near $\nu_0$ and let
\begin{equation}
\Gamma_\nu(\widetilde\eta) \;:=\; (\widetilde\eta - \eta) + \overline R_\nu(\widetilde\eta)
\;=\; \overline G(\widetilde\eta, \lambda) - \overline G(\eta, \lambda) + E(\eta, \lambda)\, (\widetilde\eta - \eta)
\label{eq:Gamma-def}
\end{equation}
with $\overline R_\nu(\widetilde\eta) := \overline G(\widetilde\eta, \lambda) - \overline G(\eta, \lambda) - \partial_\eta \overline G(\eta, \lambda)(\widetilde\eta - \eta)$, the nonlinear part of $\overline G(\cdot\,, \lambda)$ in~\eqref{eq:Gplug}. At $\widetilde\eta = \eta$ the function $\Gamma_\nu$ has value $0$, Jacobian $I_{d_\eta}$ and derivatives $\overline G_p(\eta, \lambda)$ of order $2 \leq p \leq q$, which are the values that $\overline G_{\widehat P(\nu)}(\cdot\,, \lambda)$ takes when $\widehat P(\nu)$ satisfies~\eqref{eq:Phat-moments} and~\eqref{eq:derivatives}. Replacing $\overline G_{\widehat P(\nu)}(\cdot\,, \lambda)$ by $\Gamma_\nu$ in the root equation~\eqref{eq:root} gives the equation $\Gamma_\nu(\widetilde\eta) + s\, [\overline G(\widetilde\eta, \lambda) - \Gamma_\nu(\widetilde\eta)] = 0$, and we write $\eta_*(s)$ for its solution near $\eta_0$, constructed in Lemma~\ref{lem:segment} below. Note that $\delta(s) := \eta_*(s) - \eta$ solves
\begin{equation}
\Theta_\nu\bigl( \eta + \delta(s),\, s \bigr) \;=\; 0,
\qquad
\Theta_\nu(\widetilde\eta, s) \;:=\; \overline G(\widetilde\eta, \lambda) + (1 - s)\, \bigl[ E(\eta, \lambda)\, (\widetilde\eta - \eta) - \overline G(\eta, \lambda) \bigr] .
\label{eq:Theta-def}
\end{equation}

The path $\eta_*(s)$ is not in general the path $\eta(P(s))$ of~\eqref{eq:root}, but the two have the same derivatives at $s = 0$ up to order $q$, and when $g$ is affine in $\eta$ they coincide. Since $\eta_*(0) = \eta$, the $k$-th derivative of $\eta_*$ at $s = 0$ involves the function that takes the place of $\overline G_{\widehat P(\nu)}(\cdot\,, \lambda)$ in~\eqref{eq:root} only through its derivatives at $\eta$ of order at most $k$. For $k \leq q$, the derivatives of $\eta_*$ at $s = 0$ are therefore the same for $\Gamma_\nu$ as for $\overline G_{\widehat P(\nu)}(\cdot\,, \lambda)$ with any $\widehat P(\nu)$ satisfying~\eqref{eq:Phat-moments} and~\eqref{eq:derivatives}, and so is the truncated expansion~\eqref{eq:PsiP} of $T(P(s)) = \overline m(\eta(P(s)))$. Replacing $\overline G_{\widehat P(\nu)}(\cdot\,, \lambda)$ by $\Gamma_\nu$ therefore leaves $\Psi^{(q)}(\nu)$ unchanged and does not require the existence of $\widehat P(\nu)$, which we establish in Lemma~\ref{lem:projection} only under a rank condition.

\begin{lemma}[The path $\eta_*(s)$]
\label{lem:segment}
Let $q \geq 1$ and let Assumptions~\ref{ass:jacobian-inverse} and~\ref{ass:nonlinear} hold. There are $r > 0$, $\epsilon > 0$ and $c < \infty$ such that for every $\nu$ with $\rho \leq \epsilon$ the following holds. For each $s \in [0,1]$ the function $\Theta_\nu(\cdot\,, s)$ of~\eqref{eq:Theta-def} has exactly one zero $\eta_*(s)$ in the ball $\{ \widetilde\eta : \| \widetilde\eta - \eta_0 \| \leq r \}$. It satisfies $\eta_*(0) = \eta$ and $\eta_*(1) = \eta_0$, the map $s \mapsto \eta_*(s)$ is $(q+1)$-times continuously differentiable on $[0,1]$, and, with $\delta(s) := \eta_*(s) - \eta$,
\begin{equation}
\| \delta(s) \| \;\leq\; c\, \rho
\qquad \text{and} \qquad
\bigl\| \delta^{(k)}(s) \bigr\| \;\leq\; c\, \rho^{\,k}
\quad \text{for } 1 \leq k \leq q+1 ,
\label{eq:delta-bounds}
\end{equation}
uniformly in $s \in [0,1]$.
\end{lemma}

\begin{proof}
Take $r$ small enough that the ball $B := \{ \widetilde\eta : \| \widetilde\eta - \eta_0 \| \leq r \}$ lies in $\mathcal{N}$ and that every $\lambda$ with $\| \lambda - \lambda_0 \| \leq r$ lies in the neighborhood of Assumption~\ref{ass:affine}(iii). The map $(\widetilde\eta, \lambda) \mapsto \partial_\eta \overline G(\widetilde\eta, \lambda) = \Lambda(\lambda)\, \partial_\eta \overline g(\widetilde\eta)$ is continuous and equals $I_{d_\eta}$ at $(\eta_0, \lambda_0)$, so after shrinking $r$ we have $\| \partial_\eta \overline G(\widetilde\eta, \lambda) - I_{d_\eta} \| \leq 1/4$ for $\widetilde\eta \in B$ and $\| \lambda - \lambda_0 \| \leq r$. We take $\epsilon \leq r$, so that $\eta \in B$ whenever $\rho \leq \epsilon$, and in particular $\| E(\eta, \lambda) \| \leq 1/4$. Since $\partial_{\widetilde\eta} \Theta_\nu(\widetilde\eta, s) = \partial_\eta \overline G(\widetilde\eta, \lambda) + (1-s)\, E(\eta, \lambda)$, this gives
\begin{equation}
\bigl\| \partial_{\widetilde\eta} \Theta_\nu(\widetilde\eta, s) - I_{d_\eta} \bigr\| \;\leq\; 1/2 ,
\qquad
\bigl\| \partial_{\widetilde\eta} \Theta_\nu(\widetilde\eta, s)^{-1} \bigr\| \;\leq\; 2 ,
\label{eq:theta-jac}
\end{equation}
for all $\widetilde\eta \in B$ and all $s \in [0,1]$. Hence the map $\varphi_s(\widetilde\eta) := \widetilde\eta - \Theta_\nu(\widetilde\eta, s)$, whose fixed points are the zeros of $\Theta_\nu(\cdot\,, s)$, is a contraction on $B$ with modulus $1/2$. At $\eta_0$, using $\overline G(\eta_0, \lambda) = \Lambda(\lambda)\, \overline g(\eta_0) = 0$,
\[
\varphi_s(\eta_0) - \eta_0 \;=\; -\Theta_\nu(\eta_0, s) \;=\; -(1-s) \bigl[ E(\eta, \lambda) (\eta_0 - \eta) - \overline G(\eta, \lambda) \bigr] ,
\]
and both $E(\eta, \lambda)$ and $\overline G(\eta, \lambda)$ are $O(\rho)$, so $\| \varphi_s(\eta_0) - \eta_0 \| \leq c_1 \rho$ uniformly in $s \in [0,1]$. For $\widetilde\eta \in B$ this gives
\[
\| \varphi_s(\widetilde\eta) - \eta_0 \| \;\leq\; \| \varphi_s(\widetilde\eta) - \varphi_s(\eta_0) \| + \| \varphi_s(\eta_0) - \eta_0 \| \;\leq\; \tfrac12\, r + c_1 \rho ,
\]
so for $\epsilon$ small enough that also $2 c_1 \epsilon \leq r$ the map $\varphi_s$ sends $B$ into itself. The Banach fixed point theorem then gives a unique zero $\eta_*(s) \in B$ of $\Theta_\nu(\cdot\,, s)$, and the same inequality at $\widetilde\eta = \eta_*(s)$ gives $\| \eta_*(s) - \eta_0 \| \leq \tfrac12 \| \eta_*(s) - \eta_0 \| + c_1 \rho$, that is, $\| \eta_*(s) - \eta_0 \| \leq 2 c_1 \rho$. Together with $\| \eta - \eta_0 \| \leq \rho$, this is the first bound in~\eqref{eq:delta-bounds}. Evaluating~\eqref{eq:Theta-def} gives $\Theta_\nu(\eta, 0) = 0$ and $\Theta_\nu(\eta_0, 1) = \overline G(\eta_0, \lambda) = 0$, so $\eta_*(0) = \eta$ and $\eta_*(1) = \eta_0$. The second of these uses only $\overline g(\eta_0) = 0$ and therefore holds at every $\lambda$.

For the derivative bounds, $\Theta_\nu$ is $(q+1)$-times continuously differentiable in $\widetilde\eta$ and affine in $s$, and its Jacobian in $\widetilde\eta$ is invertible on $B \times [0,1]$ by~\eqref{eq:theta-jac}, so $s \mapsto \eta_*(s)$ is $(q+1)$-times continuously differentiable there by the implicit function theorem. Write $M(s) := \partial_{\widetilde\eta} \Theta_\nu(\eta + \delta(s), s)$, so that $\| M(s)^{-1} \| \leq 2$. Since $\partial_s \Theta_\nu(\widetilde\eta, s) = \overline G(\eta, \lambda) - E(\eta, \lambda)(\widetilde\eta - \eta)$, differentiating~\eqref{eq:Theta-def} in $s$ gives
\[
M(s)\, \delta'(s) \;=\; E(\eta, \lambda)\, \delta(s) - \overline G(\eta, \lambda) ,
\]
whose right-hand side is $O(\rho)$, so $\| \delta'(s) \| \leq c \rho$. Differentiating this identity $k-1$ further times, for $2 \leq k \leq q+1$,
\[
M(s)\, \delta^{(k)}(s) \;=\; E(\eta, \lambda)\, \delta^{(k-1)}(s) \;-\; \sum_{j=1}^{k-1} \binom{k-1}{j}\, M^{(j)}(s)\, \delta^{(k-j)}(s) .
\]
For $j \geq 1$, by the chain rule, $M^{(j)}(s)$ is $-E(\eta, \lambda)\, \mathbf{1}\{ j = 1 \}$ plus a finite sum, with positive integer coefficients, of terms $\partial_\eta^{1+a} \overline G(\eta + \delta(s), \lambda) [ \delta^{(i_1)}(s), \ldots, \delta^{(i_a)}(s), \,\cdot\, ]$ with $1 \leq a \leq j$, $i_l \geq 1$ and $\sum_l i_l = j$. The derivatives $\partial_\eta^p \overline G(\cdot\,, \lambda)$ with $p \leq q+1$ are bounded on $B$, so an induction on $k$ gives $\| M^{(j)}(s) \| \leq c\, \rho^{\,j}$ for $1 \leq j \leq k-1$ and then $\| M(s)\, \delta^{(k)}(s) \| \leq c\, \rho^{\,k}$, uniformly in $s \in [0,1]$. Multiplying by $M(s)^{-1}$ gives~\eqref{eq:delta-bounds}.
\end{proof}

It remains to identify the truncated expansion~\eqref{eq:PsiP} with the tree sum and to bound it. We keep track of the Taylor coefficients of order at most $q$ at $s = 0$ by working in the ring $R_q := \mathbb{R}[s] / (s^{q+1})$ of polynomials truncated at degree $q$, that is, by computing modulo $s^{q+1}$. In $R_q$ the matrix $A_s := I_{d_\eta} - s E$ is invertible with inverse $A^{-1}_s = \sum_{j=0}^{q} s^j E^j$, so that every series below is a finite sum and no question of convergence arises. The class in $R_q$ of a function of $s$ that is $q$ times differentiable at $s = 0$ is the class of its Taylor polynomial of degree $q$ at $s = 0$. 

Let $\widehat\delta \in (R_q)^{d_\eta}$ be the class of $\delta$, which has zero constant term. Taylor's theorem applied to $\overline G(\cdot\,, \lambda)$ at $\eta$ to order $q$ leaves a remainder $O(\| \delta(s) \|^{q+1})$, which is $O(s^{q+1})$ because $\delta(0) = 0$ and $\delta$ is continuously differentiable, hence zero in $R_q$. Substituting the expansion into~\eqref{eq:Theta-def} and using $A + (1-s) E = I_{d_\eta} - sE$,
\begin{equation}
\bigl( I_{d_\eta} - s E \bigr)\, \widehat\delta
\;=\; -\, s\, \varepsilon \;-\; \sum_{p=2}^{q} \frac{1}{p!}\, \overline G_p\bigl[ \widehat\delta^{\,\otimes p} \bigr]
\qquad \text{in } R_q ,
\qquad
\varepsilon := \overline G(\eta, \lambda) .
\label{eq:root-Rq}
\end{equation}
Equivalently $\widehat\delta = \mathcal{F}(\widehat\delta)$ with $\mathcal{F}(x) := - A^{-1}_s [ s \varepsilon + \sum_{p=2}^q (1/p!)\, \overline G_p[x^{\otimes p}] ]$. If two elements of $(R_q)^{d_\eta}$ with zero constant term agree modulo $s^k$, their images under $\mathcal{F}$ agree modulo $s^{k+1}$, because $\overline G_p[x^{\otimes p}] - \overline G_p[y^{\otimes p}]$ is a sum of terms in which one argument is $x - y$ and at least one other argument has zero constant term. Iterating $\mathcal{F}$ from $0$ therefore reaches $\widehat\delta$ after $q$ steps, and the iteration is a tree recursion of the same kind as the one that determines the $B$-series coefficients of the implicitly defined stages of an implicit Runge--Kutta method (\citealp{HairerLubichWanner2006}, Ch.~III.1; \citealp{ChartierHairerVilmart2010}, Sect.~3).

\begin{lemma}[Tree expansion of $\overline m(\eta + \widehat\delta)$]
\label{lem:bseries-thm2}
In $R_q$, with $\overline m(\eta + \widehat\delta) := \sum_{r=0}^{q} (1/r!)\, \partial_\eta^r \overline m(\eta)\,[\widehat\delta^{\,\otimes r}]$,
\begin{equation}
\overline m\bigl( \eta + \widehat\delta \bigr) \;=\; \sum_{\tau'}\, \frac{(-1)^{|\tau'|}\, s^{\ell(\tau')}}{\sigma(\tau')}\, V_{\tau'}(\varepsilon, s; \eta, \lambda),
\label{eq:bseries-thm2}
\end{equation}
where the sum is over the rooted trees $\tau'$ with at most $q$ leaves in which the root has $r \geq 0$ children and every non-root node has either $0$ children (a leaf) or $\geq 2$ children (a branch). The contraction $V_{\tau'}$ is built in the same bottom-up way as $\overline\kappa_{\tau'}$, except that leaves carry $\varepsilon$ in place of $\Lambda(\lambda)\, \overline g$ and the value at each non-root node is left-multiplied by $A^{-1}_s = (I_{d_\eta} - sE)^{-1}$ before being passed up to its parent. A branch with $p$ children carrying $v_1, \ldots, v_p$ still evaluates to $\overline G_p[v_1, \ldots, v_p]$ and the root with $r$ children to $\partial_\eta^r \overline m\, [v_1, \ldots, v_r]$. Here $|\tau'|$ is the number of non-root nodes and $\ell(\tau')$ the number of leaves.
\end{lemma}

\noindent
The proof is given in Appendix~\ref{app:proof-thm2-lemmas}. Only the trees with $\ell(\tau') \leq q$ contribute, and a tree of this kind has $|\tau'| \leq 2 \ell(\tau') - 1$ once it has a non-root node, so the sum is finite. Each factor $A^{-1}_s$ is attached to the parent edge of a non-root node $v$. Expanding $A^{-1}_s = \sum_{j_v \geq 0} s^{j_v} E^{j_v}$ and each power as $E^{j_v} = (I - A)^{j_v} = \sum_{k_v = 0}^{j_v} \binom{j_v}{k_v} (-1)^{k_v} A^{k_v}$ replaces it by a sum over pairs $(j_v, k_v)$ with $0 \leq k_v \leq j_v$. Each factor of $A$ inserted on an edge is one node with a single child on that edge. We write $\tau'$ for the tree of Lemma~\ref{lem:bseries-thm2}, the \emph{skeleton}, and from here on reuse $\tau$ for the tree obtained from it by inserting $k_v$ one-child nodes on the parent edge of each non-root $v$. These chain-extended trees are the rooted trees of Section~\ref{sec:trees}, and the skeletons are those without non-root one-child nodes.

For the chain-extended tree, $\ell(\tau) = \ell(\tau')$, the number of branches is unchanged, and the number of inserted nodes is $K := \sum_v k_v$, so that $|\tau| = |\tau'| + K$ and $d(\tau) = \ell(\tau') + K$, whence $|\tau| - d(\tau) = |\tau'| - \ell(\tau')$. Substituting the expansions into~\eqref{eq:bseries-thm2}, and writing $\mathrm{nonroot}(\tau')$ for the set of non-root nodes of $\tau'$,
\begin{equation}
\overline m(\eta + \widehat\delta) \;=\; \sum_{\tau'}\, \frac{(-1)^{|\tau'|}}{\sigma(\tau')}\, s^{\ell(\tau')} \sum_{\substack{\{j_v, k_v\}\\ 0 \leq k_v \leq j_v}}\, \Bigl( \prod_{v \in \mathrm{nonroot}(\tau')} \binom{j_v}{k_v}\, (-1)^{k_v}\, s^{j_v} \Bigr)\, V_{\tau', \{k_v\}},
\label{eq:big-sum}
\end{equation}
where $V_{\tau', \{k_v\}}(\varepsilon; \eta, \lambda)$ is built like $V_{\tau'}$ but with each $A^{-1}_s$ replaced by the matrix power $A^{k_v}$ on that edge, so that it contains no matrix inverse.\footnote{Take $\tau'$ to be the tree whose root has one child, a branch with two leaves. Then $|\tau'| = 3$, $\ell(\tau') = 2$, $\sigma(\tau') = 2$, and $V_{\tau'} = \partial_\eta \overline m [ A^{-1}_s\, \overline G_2[A^{-1}_s \varepsilon, A^{-1}_s \varepsilon] ]$. With $k_v = 1$ on the branch's parent edge and $k_v = 0$ on the two leaf edges, $V_{\tau', \{k_v\}} = \partial_\eta \overline m [ A\, \overline G_2[\varepsilon, \varepsilon] ]$, the term of the chain-extended tree in which a single one-child node lies between the root and the branch.}

The element $\overline m(\eta + \widehat\delta)$ of $R_q$ in Lemma~\ref{lem:bseries-thm2} is the class of the function $h(s) := \overline m(\eta_*(s))$, because expanding $\overline m$ at $\eta$ to order $q$ leaves a remainder $O(\| \delta(s) \|^{q+1}) = O(s^{q+1})$, and $h$ is $(q+1)$-times continuously differentiable by Lemma~\ref{lem:segment}, so the coefficient of $s^k$ in $\overline m(\eta + \widehat\delta)$ is $h^{(k)}(0) / k!$ for $k \leq q$. Because $(1-s)^{-1} = \sum_{n=0}^{q} s^n$ in $R_q$, extracting the coefficient of $s^q$ from $\overline m(\eta + \widehat\delta) / (1-s)$ returns the partial sum of the coefficients of $s^0, \ldots, s^q$, so the expansion~\eqref{eq:PsiP} gives
\begin{equation}
\Psi^{(q)}(\nu) \;=\; \sum_{k=0}^{q} \frac{h^{(k)}(0)}{k!} \;=\; [s^q]\, \frac{\overline m(\eta + \widehat\delta)}{1 - s} .
\label{eq:Psi-extract}
\end{equation}

Apply the extraction~\eqref{eq:Psi-extract} to~\eqref{eq:big-sum}. A contribution has the power $s^{\ell(\tau') + \sum_v j_v}$, and dividing this power by $1-s$ and taking the coefficient $[s^q]$ gives $1$ if its exponent is at most $q$ and $0$ otherwise, so the extraction imposes $\sum_v j_v \leq q - \ell(\tau')$ and
\begin{align}
\Psi^{(q)}(\nu) \;=\;& \sum_{\tau'}\, \frac{(-1)^{|\tau'|}}{\sigma(\tau')}\, \sum_{\{k_v\}}\, (-1)^{\sum_v k_v}\, N_{\tau'}(\{k_v\}; q)\, V_{\tau', \{k_v\}},\notag\\
N_{\tau'}(\{k_v\}; q) \;:=\;& \sum_{\substack{j_v \geq k_v\\ \sum_v j_v \leq q - \ell(\tau')}}\, \prod_v \binom{j_v}{k_v}.
\label{eq:FqA}
\end{align}

\begin{lemma}[Multidimensional hockey-stick]
\label{lem:hockey}
For nonnegative integers $n$, $M$ and $a_1, \ldots, a_n$,
\begin{equation}
\sum_{\substack{x_1, \ldots, x_n \geq 0\\ \sum_i x_i \leq M}}\, \prod_{i = 1}^n \binom{x_i + a_i}{a_i} \;=\; \binom{M + n + \sum_i a_i}{n + \sum_i a_i}.
\label{eq:hockey}
\end{equation}
\end{lemma}

\noindent The proof is given in Appendix~\ref{app:proof-thm2-lemmas}. Substituting $j_v' := j_v - k_v \geq 0$ turns $\binom{j_v}{k_v}$ into $\binom{j_v' + k_v}{k_v}$ and the constraint into $\sum_v j_v' \leq q - \ell(\tau') - K$, so for $K \leq q - \ell(\tau')$ Lemma~\ref{lem:hockey} with $n = |\tau'|$, $a_v = k_v$ and $M = q - \ell(\tau') - K \geq 0$ gives
\[
N_{\tau'}(\{k_v\}; q) \;=\; \binom{q - \ell(\tau') + |\tau'|}{|\tau'| + K} \;=\; \binom{q + |\tau| - d(\tau)}{|\tau|} ,
\]
using $|\tau| = |\tau'| + K$ and $|\tau'| - \ell(\tau') = |\tau| - d(\tau)$. For $K > q - \ell(\tau')$ the sum defining $N_{\tau'}(\{k_v\}; q)$ is empty, and both binomial coefficients are zero, because their upper index $q - \ell(\tau') + |\tau'|$ is nonnegative, since the leaves are among the non-root nodes and so $|\tau'| \geq \ell(\tau')$, and smaller than their lower index. So the display holds for every $\{k_v\}$. The right-hand side depends on the chain pattern only through $|\tau|$ and $d(\tau)$. The sign in~\eqref{eq:FqA} is $(-1)^{|\tau'|} (-1)^{\sum_v k_v} = (-1)^{|\tau|}$. 

To reindex~\eqref{eq:FqA} by chain-extended trees, note that the automorphisms of $\tau'$ act on the assignments $v \mapsto k_v$ by permuting non-root nodes, that two assignments give the same $\tau$ when an automorphism of $\tau'$ relates them, and that no automorphism of $\tau$ permutes nodes within a single chain, since these are ordered by their depth. The stabilizer of an assignment inside $\mathrm{Aut}(\tau')$ is therefore isomorphic to $\mathrm{Aut}(\tau)$, and by the orbit--stabilizer theorem the number of assignments producing a given $\tau$ is $\sigma(\tau') / \sigma(\tau)$. Each of them carries the prefactor $1/\sigma(\tau')$, so together they contribute $1/\sigma(\tau)$. The constraint $k_v \leq j_v$ together with $\sum_v j_v \leq q - \ell(\tau')$ restricts the reindexed sum to $d(\tau) = \ell(\tau') + \sum_v k_v \leq q$, and every $\tau$ with $d(\tau) \leq q$ arises, from $j_v = k_v$. With $\varepsilon = \overline G(\eta, \lambda)$ the contraction $V_{\tau', \{k_v\}}$ is the term $\overline\kappa_\tau(\eta, \lambda)$, in which leaves carry $\Lambda(\lambda)\, \overline g$, one-child nodes carry $A$, branches carry $\overline G_p$, the root carries $\partial_\eta^r \overline m$, and the contraction runs along parent--child edges. Hence
\begin{equation}
\Psi^{(q)}(\nu) \;=\; \sum_{\tau \in \mathcal{T}_q}\, \frac{(-1)^{|\tau|}}{\sigma(\tau)}\, \binom{q + |\tau| - d(\tau)}{|\tau|}\, \overline\kappa_\tau(\eta, \lambda)
\;=\; \sum_{\tau \in \mathcal{T}_q}\, c_{q,\tau}\, \overline\kappa_\tau(\eta, \lambda) ,
\label{eq:Fq-final}
\end{equation}
by~\eqref{eq:Tq-def} and~\eqref{eq:cq-formula}, which is the right-hand side of~\eqref{eq:Psi-is-treesum}.

At $\nu_0$ the leaf factor $\Lambda(\lambda_0)\, \overline g(\eta_0)$ is zero, so only the single-node tree contributes to~\eqref{eq:Psi-is-treesum} and $\Psi^{(q)}(\theta_0,\eta_0,\lambda_0) = \overline m(\eta_0)$, which equals $h(1) = \overline m(\eta_*(1))$ by Lemma~\ref{lem:segment}. Taylor's theorem with integral remainder, applied to $h$ at $s = 0$ and evaluated at $s = 1$, then gives
\begin{equation}
\sum_{k=0}^{q} \frac{h^{(k)}(0)}{k!} \; -\; \Psi^{(q)}(\theta_0,\eta_0,\lambda_0)
\;=\; -\, \frac{1}{q!} \int_0^1 (1-s)^q\, h^{(q+1)}(s)\, ds .
\label{eq:taylor-remainder}
\end{equation}
By the chain rule, $h^{(q+1)}(s)$ is a finite sum, with positive integer coefficients that depend only on $q$, of terms $\partial_\eta^r \overline m(\eta_*(s)) [ \delta^{(i_1)}(s), \ldots, \delta^{(i_r)}(s) ]$ with $1 \leq r \leq q+1$, $i_l \geq 1$ and $\sum_l i_l = q+1$. The derivatives of $\overline m$ up to order $q+1$ are bounded on $\mathcal{N}$ by $\mathbb{E}[M_r(W_i)]$, and $\| \delta^{(i_l)}(s) \| \leq c\, \rho^{\,i_l}$ by~\eqref{eq:delta-bounds}, so each term is bounded by a constant times $\rho^{\,q+1}$ and $\sup_{s \in [0,1]} | h^{(q+1)}(s) | \leq C_1\, \rho^{\,q+1}$. The  sum on the left-hand side of~\eqref{eq:taylor-remainder} is $\Psi^{(q)}(\nu)$ by~\eqref{eq:Psi-extract}, which by~\eqref{eq:Fq-final} and~\eqref{eq:Psi-is-treesum} is $\Psi^{(q)}(\theta_0, \eta, \lambda)$, so
\[
\bigl| \Psi^{(q)}(\theta_0, \eta, \lambda)  - \Psi^{(q)}(\theta_0, \eta_0, \lambda_0) \bigr| \;\leq\; \frac{C_1}{(q+1)!}\, \bigl( \| \eta - \eta_0 \| + \| \lambda - \lambda_0 \| \bigr)^{q+1}
\]
for $\rho \leq \epsilon$, with $\epsilon$ the radius of Lemma~\ref{lem:segment}. This is the bound~\eqref{eq:ortho-bound} of Definition~\ref{def:orthogonality}, and Theorem~\ref{th:nonlinear} follows. \qed

\subsubsection{Proofs of Lemmas~\ref{lem:bseries-thm2} and~\ref{lem:hockey}\label{app:proof-thm2-lemmas}}

\begin{proof}[\bf Proof of Lemma~\ref{lem:bseries-thm2}]
Iterate $\widehat\delta = \mathcal{F}(\widehat\delta)$, that is
\[
\widehat\delta \;=\; -A^{-1}_s \biggl[s \varepsilon + \sum_{p = 2}^{q} \frac{1}{p!}\, \overline G_p[\widehat\delta^{\,\otimes p}]\biggr] .
\]
At each iteration every occurrence of $\widehat\delta$ on the right is replaced by one of two patterns. Either it becomes the leaf pattern $-A^{-1}_s \cdot s \varepsilon$, contributing one $\varepsilon$-leaf with $A^{-1}_s$ on its parent edge and nothing further to expand, or it becomes the branch pattern $-A^{-1}_s \cdot (1/p!)\, \overline G_p[\widehat\delta^{\,\otimes p}]$ for some $p \geq 2$, contributing a branch with $p$ children and $A^{-1}_s$ on its parent edge. 

The iteration terminates in $R_q$ after $q$ steps, and gives
\[
\widehat\delta \;=\; \sum_{\tau^{\rm ord}}\, (-1)^{|\tau^{\rm ord}|}\, s^{\ell(\tau^{\rm ord})}\, \prod_{v\,\text{branch}} \frac{1}{p_v!}\, V^\delta_{\tau^{\rm ord}}(\varepsilon, s; \eta, \lambda),
\]
where $\tau^{\rm ord}$ runs over the trees of the type described in the lemma whose root has one child, with the children of each node further distinguished by position. Here $V^\delta_{\tau^{\rm ord}}$ is the vector that the contraction of Lemma~\ref{lem:bseries-thm2} passes from that child to the root. Substituting into $\overline m(\eta + \widehat\delta) = \sum_{r=0}^{q} (1/r!)\, \partial_\eta^r \overline m(\eta)\,[\widehat\delta^{\,\otimes r}]$, every $\widehat\delta$ in $\widehat\delta^{\,\otimes r}$ becomes such a tree, and merging the $r$ roots into one gives
\begin{equation}
\overline m(\eta + \widehat\delta) \;=\; \sum_{\tau^{\rm ord}}\, (-1)^{|\tau^{\rm ord}|}\, s^{\ell(\tau^{\rm ord})}\, \frac{1}{r!}\, \prod_{v\,\text{branch}} \frac{1}{p_v!}\, V_{\tau^{\rm ord}}(\varepsilon, s; \eta, \lambda),
\label{eq:bseries-ordered}
\end{equation}
where $r$ is the number of children of the root of $\tau^{\rm ord}$. 

Group these trees by the underlying rooted tree $\tau'$, obtained by ignoring the order in which the children of each node are listed; for instance, a root whose children are a leaf and a branch arises from two ordered trees, one listing the leaf first and one listing the branch first. Fix a labeling of the nodes of $\tau'$. The orderings of the children at all nodes number $r!\, \prod_v p_v!$, two orderings give the same $\tau^{\rm ord}$ when an automorphism of $\tau'$ maps one to the other, and an automorphism that preserves an ordering fixes every node, by induction from the root. Hence the number of distinct $\tau^{\rm ord}$ representing a fixed $\tau'$ is $r!\, \prod_v p_v! / \sigma(\tau')$, and all of them have $V_{\tau^{\rm ord}} = V_{\tau'}$ because $\overline G_p$ and $\partial_\eta^r \overline m$ are symmetric. Multiplying this number by the per-ordering weight $\frac{1}{r!}\, \prod_v \frac{1}{p_v!}$ in~\eqref{eq:bseries-ordered} leaves $\frac{1}{\sigma(\tau')}$, which gives~\eqref{eq:bseries-thm2}.
\end{proof}

\begin{proof}[\bf Proof of Lemma~\ref{lem:hockey}]
   Write $S := n + \sum_i a_i$. For each $i$, the factor $\binom{x_i + a_i}{a_i}$ is the number of ways of writing $x_i$ as an ordered sum of $a_i + 1$ nonnegative integers. Introduce the slack variable $x_0 := M - \sum_i x_i \geq 0$, which absorbs the constraint $\sum_i x_i \leq M$ into an equality. The left-hand side of~\eqref{eq:hockey} then counts the ways of distributing $M$ indistinguishable units 
   among $\sum_i (a_i + 1) + 1 = S + 1$ ordered boxes, and since $M$ indistinguishable units can be distributed among $S + 1$ ordered boxes in $\binom{M + S}{S}$ ways, this implies~\eqref{eq:hockey}.
\end{proof}

\subsection{Proof of Corollaries~\ref{cor:derivatives-affine} and~\ref{cor:derivatives-nonlinear}\label{app:proof-cor}}

Under the strengthened conditions the integrable envelopes cover $2q$ derivatives of $m$, and in
Corollary~\ref{cor:derivatives-nonlinear} of $g$ as well, while $\Lambda(\cdot)$ retains the $q+1$
derivatives of Assumption~\ref{ass:affine}(iii). Each term of $\psi^{(q)}$ is a finite sum of
products whose factors carry at most $q$ derivatives of $m$ or of $g$, so $q$ further derivatives
in $(\eta,\lambda)$ require at most $2q$ derivatives of $m$ or of $g$ at any one factor, which
the strengthened envelopes dominate. Since the factors are evaluated at distinct independent copies of the data, the product of their envelopes is integrable. Differentiating under the expectation $q$ times is therefore
justified by dominated convergence, and $\Psi^{(q)}(\theta_0,\cdot)$ is $q$ times continuously
differentiable near $(\eta_0,\lambda_0)$.
Theorem~\ref{th:main}, respectively Theorem~\ref{th:nonlinear}, gives the bound of
Definition~\ref{def:orthogonality}, and Proposition~\ref{prop:reconcile} then yields~\eqref{eq:ortho}.

\subsection{Proofs of Lemma~\ref{lem:bridge} and Theorem~\ref{th:asymptotic}\label{app:asymptotic-proof}}

Let $M^\psi_{N,i}(W_i)$ be the envelope of Assumption~\ref{ass:asymp-reg}(iii). Since it depends on $W_i$ only through $\{ W_{it} : t \in \mathcal{S}_2 \}$, it is independent of $\widehat\nu$ by Assumption~\ref{ass:asymp-ortho}(ii). The set $\mathcal{U}_i(\delta)$ is convex and $\psi$ is continuously differentiable in $\nu$ there, so $\psi(W_i;\theta,\nu) - \psi(W_i;\theta,\nu'') = \int_0^1 \partial_\nu \psi(W_i;\theta,\nu'' + z(\nu - \nu''))[\nu - \nu'']\,dz$ and
\begin{equation}\label{eq:psi-lip}
\| \psi(W_i;\theta,\nu) - \psi(W_i;\theta,\nu'') \| \leq M^\psi_{N,i}(W_i)\, \| \nu - \nu'' \| , \qquad (\theta,\nu),\ (\theta,\nu'') \in \mathcal{U}_i(\delta) ,
\end{equation}
and, using also the bound on $\partial_\theta \psi$, the same bound holds for the joint difference with $\|\theta - \theta''\| + \|\nu - \nu''\|$ in place of $\|\nu - \nu''\|$. 

\begin{proof}[\bf Proof of Lemma~\ref{lem:bridge}]
Let $\mathcal{A}_N := \{\max_i \widehat\rho_i \leq \delta\}$, an event determined by $\widehat\nu$. By Markov's inequality, $\Pr(\mathcal{A}_N^c) \leq \sum_{i=1}^N \Pr(\widehat\rho_i > \delta) \leq \delta^{-2(q+1)}\sum_{i=1}^N \mathbb{E}[\widehat\rho_i^{\,2(q+1)}] = \delta^{-2(q+1)}\,N\,o(N^{-1}) = o(1)$ by Assumption~\ref{ass:asymp-reg}(iv). It therefore suffices to show that $\mathcal{B}_N \mathbf{1}_{\mathcal{A}_N}$ and $\mathcal{E}_N \mathbf{1}_{\mathcal{A}_N}$ are $o_P(1)$, and on $\mathcal{A}_N$ the identity~\eqref{eq:asymp-decomp} holds with all its terms finite. On $\mathcal{A}_N$, by Assumption~\ref{ass:asymp-ortho}(ii) and (iii), $\mathbb{E}[\psi(W_i;\theta_0,\widehat\nu_i) \mid \widehat\nu] = \Psi_i(\theta_0,\widehat\nu_i)$, and Assumption~\ref{ass:asymp-ortho}(i) gives $\|\mathcal{B}_N\| \leq C N^{-1/2}\sum_{i=1}^N \widehat\rho_i^{\,q+1}$. Since $\bigl(N^{-1}\sum_i \mathbb{E}[\widehat\rho_i^{\,q+1}]\bigr)^2 \leq N^{-1}\sum_i \bigl(\mathbb{E}[\widehat\rho_i^{\,q+1}]\bigr)^2 \leq N^{-1}\sum_i \mathbb{E}[\widehat\rho_i^{\,2(q+1)}]$, Assumption~\ref{ass:asymp-reg}(iv) gives $N^{-1}\sum_i \mathbb{E}[\widehat\rho_i^{\,q+1}] = o(N^{-1/2})$, so that $C\sqrt N\, N^{-1}\sum_i \mathbb{E}[\widehat\rho_i^{\,q+1}] = o(1)$. For any $\epsilon > 0$, $$\Pr(\|\mathcal{B}_N\| \mathbf{1}_{\mathcal{A}_N} > \epsilon) \leq \epsilon^{-1} C \sqrt N\, N^{-1}\sum_i \mathbb{E}[\widehat\rho_i^{\,q+1}] = o(1).$$

On $\mathcal{A}_N$, write $\Delta_i := \psi(W_i;\theta_0,\widehat\nu_i) - \psi(W_i;\theta_0,\nu_{i0})$ and $\widetilde\Delta_i := \Delta_i - \mathbb{E}[\Delta_i \mid \widehat\nu]$, so that $\mathcal{E}_N = N^{-1/2}\sum_{i=1}^N \widetilde\Delta_i$; the conditional expectation is finite there because $\|\Delta_i\| \leq M^\psi_{N,i}(W_i)\,\widehat\rho_i$ by~\eqref{eq:psi-lip}, $\widehat\nu_i$ lying in the $\nu$-section of $\mathcal{U}_i(\delta)$. By Assumption~\ref{ass:asymp-ortho}(ii) the vector $\widehat\nu$ is independent of the evaluation blocks $\{W_{it} : t \in \mathcal{S}_2\}$, $i = 1,\ldots,N$, and by Assumption~\ref{ass:asymp-ortho}(iii) those blocks are independent across $i$, so conditionally on $\widehat\nu$ the $\widetilde\Delta_i$ are independent across $i$ and each has conditional mean zero. For $i \neq j$ this gives $\mathbb{E}[\widetilde\Delta_i'\widetilde\Delta_j \mid \widehat\nu] = 0$, and since $\mathbb{E}[\|X - \mathbb{E}[X \mid \widehat\nu]\|^2 \mid \widehat\nu] \leq \mathbb{E}[\|X\|^2 \mid \widehat\nu]$, on $\mathcal{A}_N$
\[
\mathbb{E}[ \|\mathcal{E}_N\|^2 \mid \widehat\nu ] \;\leq\; N^{-1}\sum_{i=1}^N \mathbb{E}[ \|\Delta_i\|^2 \mid \widehat\nu ] \;\leq\; N^{-1}\sum_{i=1}^N \mathbb{E}[M^\psi_{N,i}(W_i)^2]\,\widehat\rho_i^{\,2} ,
\]
the last step because $M^\psi_{N,i}(W_i)$ is independent of $\widehat\nu$. Since $\mathcal{A}_N$ is determined by $\widehat\nu$, $\mathbb{E}[ \|\mathcal{E}_N\|^2 \mathbf{1}_{\mathcal{A}_N} ] \leq ( \sup_N \max_{i} \mathbb{E}[ M^\psi_{N,i}(W_i)^2 ] )\, N^{-1}\sum_{i=1}^N \mathbb{E}[ \widehat\rho_i^{\,2} ]$. For each $i$, $\mathbb{E}[\widehat\rho_i^{\,2}] \leq (\mathbb{E}[\widehat\rho_i^{\,2(q+1)}])^{1/(q+1)}$, and $x \mapsto x^{1/(q+1)}$ is concave on $[0,\infty)$, so the average of those roots is at most the root of the average and $N^{-1}\sum_i \mathbb{E}[\widehat\rho_i^{\,2}] \leq (N^{-1}\sum_i \mathbb{E}[\widehat\rho_i^{\,2(q+1)}])^{1/(q+1)} = (o(N^{-1}))^{1/(q+1)} = o(1)$ by Assumption~\ref{ass:asymp-reg}(iv). For any $\epsilon > 0$, $\Pr(\|\mathcal{E}_N\| \mathbf{1}_{\mathcal{A}_N} > \epsilon) \leq \epsilon^{-2}\mathbb{E}[\|\mathcal{E}_N\|^2 \mathbf{1}_{\mathcal{A}_N}] = o(1)$. Together with $\Pr(\mathcal{A}_N^c) = o(1)$ and the identity~\eqref{eq:asymp-decomp} on $\mathcal{A}_N$, this proves the lemma.
\end{proof}

\begin{proof}[\bf Proof of Theorem~\ref{th:asymptotic}]
On the event of Assumption~\ref{ass:asymp-reg}(i), whose probability tends to one, $(\widehat\theta,\widehat\nu_i) \in \mathcal{U}_i(\delta)$ for every $i$, and $\psi(W_i;\cdot,\widehat\nu_i)$ is continuously differentiable on the segment from $\theta_0$ to $\widehat\theta$, so $N^{-1}\sum_i \psi(W_i;\widehat\theta,\widehat\nu_i) = N^{-1}\sum_i \psi(W_i;\theta_0,\widehat\nu_i) + \widehat{\mathcal{J}}_N (\widehat\theta - \theta_0)$, where $\widehat{\mathcal{J}}_N := N^{-1}\sum_i \int_0^1 \partial \psi(W_i;\theta_0 + z(\widehat\theta - \theta_0),\widehat\nu_i)/\partial\theta'\,dz$. By Assumption~\ref{ass:asymp-reg}(iii) applied to $\partial^2_\theta \psi$ and $\partial_\theta \partial_\nu \psi$, and the same integral form along the segment from $(\theta_0,\nu_{i0})$, the integrand differs from $\partial \psi(W_i;\theta_0,\nu_{i0})/\partial\theta'$ by at most $M^\psi_{N,i}(W_i)(\|\widehat\theta - \theta_0\| + \widehat\rho_i)$, so
\[
\biggl\| \widehat{\mathcal{J}}_N - \frac{1}{N}\sum_{i=1}^N \frac{\partial \psi(W_i;\theta_0,\nu_{i0})}{\partial\theta'} \biggr\|
\;\leq\; \biggl( \frac{1}{N}\sum_{i=1}^N M^\psi_{N,i}(W_i) \biggr) \Bigl( \|\widehat\theta - \theta_0\| + \max_i \widehat\rho_i \Bigr) .
\] 
The first factor has expectation at most $\sup_N \max_{i}(\mathbb{E}[M^\psi_{N,i}(W_i)^2])^{1/2}$ and is $O_P(1)$. In the second, $\|\widehat\theta - \theta_0\| = o_P(1)$ by Assumption~\ref{ass:asymp-reg}(iv) and $\Pr(\max_i \widehat\rho_i > \epsilon) \leq \epsilon^{-2(q+1)}\sum_i \mathbb{E}[\widehat\rho_i^{\,2(q+1)}] = o(1)$ by the same part. Hence $\widehat{\mathcal{J}}_N \to_P G$ by Assumption~\ref{ass:asymp-reg}(v), and the same argument gives $$\widetilde{\mathcal{J}}_N := N^{-1}\sum_i \partial \psi(W_i;\widehat\theta,\widehat\nu_i)/\partial\theta' \to_P G.$$

Next, $\|\psi(W_i;\theta_0,\nu_{i0})\| \leq M^\psi_{N,i}(W_i)$, so the summands $\psi(W_i;\theta_0,\nu_{i0}) - \Psi_i(\theta_0,\nu_{i0})$ are independent across $i$ by Assumption~\ref{ass:asymp-ortho}(iii), have mean zero, and have second moments bounded by $\sup_N \max_{i}\mathbb{E}[M^\psi_{N,i}(W_i)^2] < \infty$. By~\eqref{eq:asymp-identification} and Chebyshev's inequality,
\[
S_N := \frac{1}{\sqrt N}\sum_{i=1}^N \psi(W_i;\theta_0,\nu_{i0}) \;=\; \frac{1}{\sqrt N}\sum_{i=1}^N \bigl[ \psi(W_i;\theta_0,\nu_{i0}) - \Psi_i(\theta_0,\nu_{i0}) \bigr] \;=\; O_P(1) .
\]
Lemma~\ref{lem:bridge} gives $N^{-1/2}\sum_i \psi(W_i;\theta_0,\widehat\nu_i) = S_N + o_P(1)$. Since $\Theta$ contains every $\theta$ with $\|\theta - \theta_0\| < \delta$, and both $\|\widehat\theta - \theta_0\|$ and $\max_i \widehat\rho_i$ are $o_P(1)$, with probability tending to one $\widehat\theta$ is interior and $\|\widehat\theta - \theta_0\| + \max_i \widehat\rho_i < \delta$. Then the objective in~\eqref{eq:asymp-estimator} is differentiable near $\widehat\theta$ by Assumption~\ref{ass:asymp-reg}(iii), and the first-order condition $$\widetilde{\mathcal{J}}_N'\Omega\, N^{-1}\sum_i \psi(W_i;\widehat\theta,\widehat\nu_i) = 0$$ holds. Substituting the expansion above and multiplying by $\sqrt N$,
\[
0 \;=\; \widetilde{\mathcal{J}}_N'\Omega\,\bigl( S_N + o_P(1) \bigr) + \widetilde{\mathcal{J}}_N'\Omega\,\widehat{\mathcal{J}}_N\, \sqrt N(\widehat\theta - \theta_0) .
\]
Because $S_N = O_P(1)$ and $\widetilde{\mathcal{J}}_N, \widehat{\mathcal{J}}_N \to_P G$, this gives $0 = G'\Omega S_N + (G'\Omega G)\sqrt N(\widehat\theta - \theta_0) + o_P(1 + \|\sqrt N(\widehat\theta - \theta_0)\|)$. By Assumption~\ref{ass:asymp-reg}(v), $G'\Omega G$ is nonsingular, so first $\sqrt N(\widehat\theta - \theta_0) = O_P(1)$ and then the remainder is $o_P(1)$, leaving $\sqrt N(\widehat\theta - \theta_0) = -(G'\Omega G)^{-1}G'\Omega S_N + o_P(1)$. Here $G'\Omega S_N = N^{-1/2}\sum_i \xi_{N,i} = N^{-1/2}\sum_i (\xi_{N,i} - \mathbb{E}[\xi_{N,i}])$, the summands $\xi_{N,i}$ being centered on average because $\sum_i \mathbb{E}[\xi_{N,i}] = G'\Omega \sum_i \Psi_i(\theta_0,\nu_{i0}) = 0$ by~\eqref{eq:asymp-identification}, although not individually. By Assumption~\ref{ass:asymp-ortho}(iii) and (iv) and the Lindeberg--Feller central limit theorem, $G'\Omega S_N \to_d \mathcal{N}(0,V)$, and the conclusion follows by the continuous mapping theorem and Slutsky's lemma.
\end{proof}

\end{document}